\documentclass{entcs} 
\usepackage{entcsmacro}
\usepackage{graphicx}
\usepackage[T1]{fontenc}

\usepackage{amsmath}
\usepackage{amssymb}
\usepackage{array}
\usepackage{bbold}              
\usepackage{centernot}
\usepackage{cite}
\usepackage{enumerate}
\usepackage[mathscr]{euscript}  
\let\mathcur\mathscr
\usepackage{hyperref}
\usepackage{mathrsfs}           
\usepackage{multirow}
\usepackage{paralist}
\usepackage{prooftree}
\usepackage{relsize}
\usepackage{stmaryrd}
\SetSymbolFont{stmry}{bold}{U}{stmry}{m}{n}
\usepackage{wrapfig}
\usepackage[all,cmtip,poly]{xy}

\makeatletter
\def\old@comma{,}
\catcode`\,=13
\def,{%
  \ifmmode%
    \old@comma\discretionary{}{}{}%
  \else%
    \old@comma%
  \fi%
}
\makeatother

\newcommand{\capp}{\textsf{CAP}}
\newcommand{\ppc}{\textsf{PPC}}

\newcommand{\ie}{{\em i.e.~}}

\newcommand{\cf}{{\em cf.~}}

\newcommand{\coloneq}{\ensuremath{:=}}
\newcommand{\Coloneq}{\ensuremath{:\coloneq}}
\newcommand{\llbrace}{\{\!\!\{}
\newcommand{\rrbrace}{\}\!\!\}}

\newcommand{\pair}[2]{\ensuremath{\langle{#1},{#2}\rangle}}

\newcommand{\qqqquad}{\qquad\qquad}

\newcommand{\mathLarger}[1]{\ensuremath{\mathlarger{\mathlarger{#1}}}}

\newcommand{\sequP}[2]{\ensuremath{{#1}\vdash_{\mathsf{p}}{#2}}}
\newcommand{\sequT}[2]{\ensuremath{{#1}\vdash{#2}}}
\newcommand{\sequPDeriv}[2]{\ensuremath{{#1}\rhd_{\mathsf{p}}{#2}}}
\newcommand{\sequTDeriv}[2]{\ensuremath{ {#1}\rhd{#2}}}
\newcommand{\sequC}[2]{\ensuremath{#2}}
\newcommand{\sequTE}[2]{\ensuremath{{#1}\vdash{#2}}}

\newcommand{\emphdef}[1]{\textbf{{#1}}}
\newcommand{\eqdef}{\ensuremath{\triangleq}}

\newcommand{\Rule}[3]{
    \prooftree
         {#1}
    \justifies  
         {#2}
    \thickness=0.05em
    \using
         {#3}
    \endprooftree}
\newcommand{\RuleCo}[3]{
    \prooftree
         {#1}
    \Justifies  
         {#2}
    \thickness=0.1em
    \using
         {#3}
    \endprooftree}

\newcommand{\ruleName}[1]{\ensuremath{\textsc{({#1})}}}

\newcommand{\ruleBeta}{\ruleName{$\beta$}}

\newcommand{\rulePComp}{\ruleName{p-comp}}
\newcommand{\rulePConst}{\ruleName{p-const}}
\newcommand{\rulePMatch}{\ruleName{p-match}}

\newcommand{\ruleTAbs}{\ruleName{t-abs}}
\newcommand{\ruleTApp}{\ruleName{t-app}}
\newcommand{\ruleTComp}{\ruleName{t-comp}}
\newcommand{\ruleTConst}{\ruleName{t-const}}
\newcommand{\ruleTSubs}{\ruleName{t-subs}}
\newcommand{\ruleTVar}{\ruleName{t-var}}

\newcommand{\ruleEqmuComp}{\ruleName{e-comp}}
\newcommand{\ruleEqmuContr}{\ruleName{e-contr}}
\newcommand{\ruleEqmuFold}{\ruleName{e-fold}}
\newcommand{\ruleEqmuFunc}{\ruleName{e-func}}
\newcommand{\ruleEqmuRec}{\ruleName{e-rec}}
\newcommand{\ruleEqmuRefl}{\ruleName{e-refl}}
\newcommand{\ruleEqmuSymm}{\ruleName{e-symm}}
\newcommand{\ruleEqmuTrans}{\ruleName{e-trans}}
\newcommand{\ruleEqmuUnion}{\ruleName{e-union}}
\newcommand{\ruleEqmuUnionAssoc}{\ruleName{e-union-assoc}}
\newcommand{\ruleEqmuUnionComm}{\ruleName{e-union-comm}}
\newcommand{\ruleEqmuUnionIdem}{\ruleName{e-union-idem}}

\newcommand{\ruleEqcoComp}{\ruleName{e-comp-t}}
\newcommand{\ruleEqcoFunc}{\ruleName{e-func-t}}
\newcommand{\ruleEqcoRefl}{\ruleName{e-refl-t}}
\newcommand{\ruleEqcoUnion}{\ruleName{e-union-t}}

\newcommand{\ruleSubmuComp}{\ruleName{s-comp}}
\newcommand{\ruleSubmuEq}{\ruleName{s-eq}}
\newcommand{\ruleSubmuFunc}{\ruleName{s-func}}
\newcommand{\ruleSubmuHyp}{\ruleName{s-hyp}}
\newcommand{\ruleSubmuRec}{\ruleName{s-rec}}
\newcommand{\ruleSubmuRefl}{\ruleName{s-refl}}
\newcommand{\ruleSubmuTrans}{\ruleName{s-trans}}
\newcommand{\ruleSubmuUnionL}{\ruleName{s-union-l}}
\newcommand{\ruleSubmuUnionRL}{\ruleName{s-union-r1}}
\newcommand{\ruleSubmuUnionRR}{\ruleName{s-union-r2}}

\newcommand{\ruleSubcoComp}{\ruleName{s-comp-t}}
\newcommand{\ruleSubcoFunc}{\ruleName{s-func-t}}
\newcommand{\ruleSubcoRefl}{\ruleName{s-refl-t}}
\newcommand{\ruleSubcoUnion}{\ruleName{s-union-t}}

\newcommand{\set}[1]{\ensuremath{\left\{{#1}\right\}}}
\newcommand{\lista}[1]{\ensuremath{\left[{#1}\right]}}

\newcommand{\Variable}{\ensuremath{{\mathbb{V}}}}
\newcommand{\Constant}{\ensuremath{{\mathbb{C}}}}
\newcommand{\Pattern}{\ensuremath{{\mathbb{P}}}}
\newcommand{\DataStructure}{\ensuremath{{\mathbb{D}}}}
\newcommand{\Term}{\ensuremath{{\mathbb{T}}}}
\newcommand{\MatchableForms}{\ensuremath{{\mathbb{M}}}}

\newcommand{\TypeVariable}{\ensuremath{\mathcal{V}}}
\newcommand{\TypeConstant}{\ensuremath{\mathcal{C}}}
\newcommand{\Type}{\ensuremath{\mathcal{T}}}

\newcommand{\DataTypeVariable}{\ensuremath{\TypeVariable_{D}}}
\newcommand{\DataType}{\ensuremath{\Type_{D}}}

\newcommand{\Tree}{\ensuremath{\mathfrak{T}}}
\newcommand{\TreeFinite}{\ensuremath{\Tree^{\mathit{fin}}}}
\newcommand{\TreeRegular}{\ensuremath{\Tree^{\mathit{reg}}}}

\newcommand{\Parts}[1]{\ensuremath{\wp\left({#1}\right)}}

\newcommand{\Natural}{\ensuremath{\mathbb{N}}}

\newcommand{\Refl}{\ensuremath{\mathit{Refl}}}
\newcommand{\Symm}{\ensuremath{\mathit{Symm}}}
\newcommand{\Trans}{\ensuremath{\mathit{Trans}}}

\newcommand{\Phieqtypeco}{\ensuremath{\Phi_{\eqtypeco}}}

\newcommand{\Phisubtypeco}{\ensuremath{\Phi_{\subtypeco}}}

\newcommand{\M}{\ensuremath{{\cal M}}}
\renewcommand{\S}{\ensuremath{{\cal S}}}
\newcommand{\U}{\ensuremath{{\cal U}}}
\newcommand{\X}{\ensuremath{{\cal X}}}

\newcommand{\treeFont}[1]{\mathcur{#1}}
\newcommand{\tA}{\ensuremath{\treeFont{A}}}
\newcommand{\tB}{\ensuremath{\treeFont{B}}}
\newcommand{\tC}{\ensuremath{\treeFont{C}}}
\newcommand{\tD}{\ensuremath{\treeFont{D}}}

\newcommand{\idatatype}{\ensuremath{\mathrel@}}
\newcommand{\ifunctype}{\ensuremath{\supset}}
\newcommand{\iuniontype}{\ensuremath{\oplus}}
\newcommand{\irectype}{\ensuremath{\mu}}

\newcommand{\consttype}[1]{\ensuremath{\mathbb{#1}}}
\newcommand{\datatype}[2]{\ensuremath{{#1}\idatatype{#2}}}
\newcommand{\functype}[2]{\ensuremath{{#1}\ifunctype{#2}}}
\newcommand{\uniontype}[2]{\ensuremath{{#1}\iuniontype{#2}}}
\newcommand{\maxuniontype}[2]{\ensuremath{\mathLarger{\iuniontype}_{\substack{#1}}{#2}}}
\newcommand{\rectype}[2]{\ensuremath{\irectype{#1}.{#2}}}

\newcommand{\subtype}{\ensuremath{\preceq}}

\newcommand{\eqtype}{\ensuremath{\simeq}}

\newcommand{\subtypemu}{\ensuremath{\subtype_{\mu}}}

\newcommand{\eqtypemu}{\ensuremath{\eqtype_{\mu}}}

\newcommand{\subtypeco}{\ensuremath{\subtype_{\Tree}}}

\newcommand{\eqtypeco}{\ensuremath{\eqtype_{\Tree}}}

\newcommand{\bool}{\ensuremath{\mathsf{Bool}}}
\newcommand{\nat}{\ensuremath{\mathsf{Nat}}}

\newcommand{\iappterm}{\ensuremath{\,}}
\newcommand{\idataterm}{\ensuremath{\,}}
\newcommand{\icaseterm}{\ensuremath{\mathrel|}}
\newcommand{\ifuncterm}{\ensuremath{\shortrightarrow}}

\newcommand{\matchable}[1]{\ensuremath{{#1}}}
\newcommand{\constterm}[1]{\ensuremath{\mathtt{#1}}}
\newcommand{\appterm}[2]{\ensuremath{{#1}\iappterm{#2}}}
\newcommand{\dataterm}[2]{\ensuremath{{#1}\idataterm{#2}}}
\newcommand{\absterm}[3]{\ensuremath{{#1}\ifuncterm_{#3}{#2}}}

\newcommand{\ifThenElse}[3]{\texttt{if }{#1}\texttt{ then }{#2}\texttt{ else }{#3}}

\newcommand{\fail}{\ensuremath{\mathtt{fail}}}
\newcommand{\wait}{\ensuremath{\mathtt{wait}}}
\newcommand{\avoids}{\ensuremath{\mathrel\mathtt{avoids}}}

\newcommand{\reduce}{\ensuremath{\rightarrow}}

\newcommand{\fv}[1]{\ensuremath{\mathsf{fv}\!\left({#1}\right)}}
\newcommand{\fm}[1]{\ensuremath{\mathsf{fm}\!\left({#1}\right)}}

\newcommand{\dom}[1]{\ensuremath{\mathsf{dom}\left({#1}\right)}}
\newcommand{\img}[1]{\ensuremath{\mathsf{img}\left({#1}\right)}}

\newcommand{\rename}[2]{\ensuremath{\left\{{#2}/{#1}\right\}}}
\newcommand{\substitute}[3]{\rename{#1}{#2}{#3}}
\newcommand{\basicmatch}[2]{\ensuremath{\llbrace{#2}/{#1}\rrbrace}}
\newcommand{\match}[2]{\basicmatch{#1}{#2}}

\newcommand{\card}[2]{\ensuremath{\#_{#1}\!\left({#2}\right)}}
\newcommand{\cut}[2]{\ensuremath{{#1}|_{{#2}}}}
\newcommand{\size}[1]{\ensuremath{|{#1}|}}
\newcommand{\length}[1]{\size{#1}}
\newcommand{\res}[2]{\ensuremath{{#1}|_{#2}}}
\newcommand{\unfoldf}[3]{\ensuremath{{#1}_{#2}^{#3}}}

\newcommand{\toBTree}[1]{\ensuremath{\llbracket{#1}\rrbracket^{\Tree}}}

\newcommand{\Psicomp}[2]{\ensuremath{{\cal P}_{\mathsf{comp}}\!\left({#1},{#2}\right)}}

\newcommand{\compatible}[2]{\ensuremath{{#1}\lll{#2}}}
\newcommand{\matches}[2]{\ensuremath{{#1}\mathbin{\vartriangleleft}{#2}}}
\newcommand{\nmatches}[2]{\ensuremath{{#1}\mathbin{\centernot\vartriangleleft}{#2}}}

\newcommand{\at}[2]{{#1}\ensuremath{|_{#2}}}
\newcommand{\lookup}[2]{{#1}\ensuremath{\|_{#2}}}
\newcommand{\pos}[1]{\ensuremath{\mathsf{pos}\!\left({#1}\right)}}
\newcommand{\cpos}[2]{\ensuremath{\mathsf{mmpos}\!\left({#1},{#2}\right)}}
\newcommand{\mpos}[1]{\ensuremath{\mathsf{maxpos}\!\left({#1}\right)}}

\def\lastname{Viso -- Bonelli -- Ayala-Rincon}

\begin{document}
\begin{frontmatter}
\title{Type Soundness for Path Polymorphism\thanksref{STICAmSud}}

\thanks[STICAmSud]{Work partially funded by the international projects DeCOPA
STIC-AmSud 146/2012, CONICET, CAPES, CNRS; and ECOS-Sud A12E04, CONICET, CNRS.}

\author[UBA]{Andr\'es Viso\thanksref{emailAndres}}
\author[UNQ]{Eduardo Bonelli\thanksref{emailEduardo}}
\author[UnB]{Mauricio Ayala-Rinc\'{o}n\thanksref{emailMauricio}}

\address[UBA]{Consejo Nacional de Investigaciones Cient\'{i}ficas y
T\'{e}cnicas -- CONICET \\
Departamento de Computaci\'{o}n \\
Facultad de Ciencias Exactas y Naturales \\
Universidad de Buenos Aires -- UBA \\
Buenos Aires, Argentina}
\address[UNQ]{Consejo Nacional de Investigaciones Cient\'{i}ficas y
T\'{e}cnicas -- CONICET \\
Departamento de Ciencia y Tecnolog\'{i}a \\
Universidad Nacional de Quilmes -- UNQ \\
Bernal, Argentina}
\address[UnB]{Departamentos de Matem\'{a}tica e Ci\^{e}ncia da
Computa\c{c}\~{a}o \\
Universidade de Bras\'{\i}lia -- UnB \\
Bras\'{\i}lia D.F., Brasil}

\thanks[emailAndres]{Email:
  \href{mailto:aeviso@dc.uba.ar}{\texttt{\normalshape aeviso@dc.uba.ar}}} 
\thanks[emailEduardo]{Email:
  \href{mailto:eabonelli@gmail.com}{\texttt{\normalshape eabonelli@gmail.com}}}
\thanks[emailMauricio]{Email:
  \href{mailto:ayala@unb.br}{\texttt{\normalshape ayala@unb.br}}}

\begin{abstract}
\emph{Path polymorphism} is the ability to define functions that can operate
uniformly over arbitrary recursively specified data structures. Its essence is
captured by patterns of the form $\dataterm{x}{y}$ which decompose a compound
data structure into its parts. Typing these kinds of patterns is challenging
since the type of a compound should determine the type of its components.
We propose a static type system (\ie no run-time analysis) for a pattern
calculus that captures this feature. Our solution combines type application,
constants as types, union types and recursive types. We address the fundamental
properties of Subject Reduction and Progress that guarantee a well-behaved
dynamics. Both these results rely crucially on a notion of \emph{pattern
compatibility} and also on a coinductive characterisation of subtyping.
\end{abstract}


\end{frontmatter}

\maketitle

\section{Introduction}
\label{sec:intro}

Applicative representation of data structures in functional programming
languages consists in applying variable arity constructors to arguments.
Examples are: $$
\arraycolsep=1.5pt\def\arraystretch{1}
\begin{array}{rcl}
s & = &
\dataterm{\dataterm{\constterm{cons}}{(\dataterm{\constterm{vl}}{v_1})}}{(\dataterm{\dataterm{\constterm{cons}}{(\dataterm{\constterm{vl}}{v_2})}}{\constterm{nil}})} \\
t & = & \dataterm{\dataterm{\dataterm{\constterm{node}}{(\dataterm{\constterm{vl}}{v_3})}}{(\dataterm{\dataterm{\dataterm{\constterm{node}}{(\dataterm{\constterm{vl}}{v_4})}}{\constterm{nil}}}{\constterm{nil}})}}{(\dataterm{\dataterm{\dataterm{\constterm{node}}{(\dataterm{\constterm{vl}}{v_5})}}{\constterm{nil}}}{\constterm{nil}})}
\end{array} $$
These are data structures that hold values, prefixed by the constructor
\constterm{vl} for ``value'' ($v_{1,2}$ in the first case, and $v_{3,4,5}$ in
the second). Consider the following function for updating the values of any of
these two structures by applying some user-supplied function $f$ to it: 
\begin{equation}
\mathsf{upd} = \absterm{\matchable{f}}{
\arraycolsep=1.5pt\def\arraystretch{1}
\begin{array}[t]{clll}
(          & \dataterm{\constterm{vl}}{\matchable{z}} & \ifuncterm_{\set{z:A}}      & \dataterm{\constterm{vl}}{(\appterm{f}{z})} \\
\icaseterm & \dataterm{\matchable{x}}{\matchable{y}}  & \ifuncterm_{\set{x:C, y:D}} & \dataterm{(\appterm{\appterm{\mathsf{upd}}{f}}{x})}{(\appterm{\appterm{\mathsf{upd}}{f}}{y})} \\
\icaseterm & \matchable{w}                            & \ifuncterm_{\set{w:E}}      & w)
\end{array}
}{\set{f:\functype{A}{B}}}
\label{eq:intro:upd}
\end{equation}
Both $\appterm{\appterm{\mathsf{upd}}{(+1)}}{s}$ and
$\appterm{\appterm{\mathsf{upd}}{(+1)}}{t}$ may be evaluated. The expression to
the right of ``='' is called an \emph{abstraction} and consists of a unique
\emph{branch}; this branch in turn is formed from a pattern (\matchable{f}), a
user-specified type declaration for the variables in the pattern
($\set{f:\functype{A}{B}}$), and a body (in this case the body is itself
another abstraction that consists of three branches). An  argument to an
abstraction is matched against the patterns, in the order in which they are
written, and the appropriate body is selected. Notice the pattern
$\dataterm{\matchable{x}}{\matchable{y}}$. This pattern embodies the essence of
what is known as \emph{path
polymorphism}~\cite{DBLP:journals/jfp/JayK09,Jay:2009} since it abstracts a
path being ``split''.  The starting point of this paper is how to type a
calculus, let us call it $\capp$ for \emph{Calculus of Applicative Patterns},
that admits such examples. We next show why the problem is challenging, explain
our contribution and also discuss why the current literature falls short of
addressing it. We do so with an introduction-by-example approach, for the full
syntax and semantics of the calculus refer to Sec.~\ref{sec:terms}.

\paragraph*{Preliminaries on typing patterns expressing path polymorphism}  
Consider these two simple examples:
\begin{equation}
  \appterm{(\absterm{\constterm{nil}}{0}{})}{\constterm{cons}}
  \qqqquad
  \appterm{(\absterm{\dataterm{\constterm{vl}}{x}}{x+1}{\set{x:\nat}})}
          {(\dataterm{\constterm{vl}}{\constterm{true}})}
  \label{eq:intro:cx}
\end{equation}
They should clearly not be typable. In the first case, the abstraction is not
capable of handling $\constterm{cons}$. This is avoided by introducing
singleton types in the form of the constructors themselves: $\constterm{nil}$
is given type $\consttype{nil}$ while $\constterm{cons}$ is given type
$\consttype{cons}$; these are then compared. In the second case,
$\matchable{x}$ in the pattern is required to be $\nat$ yet the type of the
argument to $\constterm{vl}$ in $\dataterm{\constterm{vl}}{\constterm{true}}$
is $\bool$. This is avoided by introducing type
application~\cite{DBLP:journals/corr/abs-1009-3429} into types:
$\dataterm{\constterm{vl}}{x}$ is assigned a type of the form
$\datatype{\consttype{vl}}{\nat}$ while
$\dataterm{\constterm{vl}}{\constterm{true}}$ is assigned type
$\datatype{\consttype{vl}}{\bool}$; these are then compared.

Consider next the pattern $\dataterm{\matchable{x}}{\matchable{y}}$ of
$\mathsf{upd}$. It can be instantiated with different applicative terms in each
recursive call to $\mathsf{upd}$. For example, suppose $A = B = \nat$, that
$v_1$ and $v_2$ are numbers and consider
$\appterm{\appterm{\mathsf{upd}}{(+1)}}{s}$. The following table illustrates
some of the terms with which $\matchable{x}$ and $\matchable{y}$ are
instantiated during the evaluation of
$\appterm{\appterm{\mathsf{upd}}{(+1)}}{s}$:
\begin{center}
\begin{tabular}{l|c|c}
 & $\matchable{x}$ & $\matchable{y}$ \\ 
\hline
$\appterm{\appterm{\mathsf{upd}}{(+1)}}{s}$                                &
$\dataterm{\constterm{cons}}{(\dataterm{\constterm{vl}}{v_1})}$            &
$\dataterm{\dataterm{\constterm{cons}}{(\dataterm{\constterm{vl}}{v_2})}}
           {\constterm{nil}}$                                              \\
$\appterm{\appterm{\mathsf{upd}}{(+1)}}{(\dataterm{\constterm{cons}}
         {(\dataterm{\constterm{vl}}{v_1})})}$                             &
$\constterm{cons}$                                                         &
$\dataterm{\constterm{vl}}{v_1}$
                                                                           \\
$\appterm{\appterm{\mathsf{upd}}{(+1)}}{(
  \dataterm{\dataterm{\constterm{cons}}{(\dataterm{\constterm{vl}}{v_2})}}
           {\constterm{nil}}
)}$                                                                        &
$\dataterm{\constterm{cons}}{(\dataterm{\constterm{vl}}{v_2})}$            &
$\constterm{nil}$
\end{tabular}
\end{center}
The type assigned to $\matchable{x}$ (and $\matchable{y}$) should encompass all
terms in its respective column. This suggests  adopting a union type for
$\matchable{x}$. On the assumption that the programmer has provided an
exhaustive coverage, the type of $\matchable{x}$ in $\mathsf{upd}$ is: $$
\rectype{\alpha}{
  \uniontype{
    \uniontype{(\datatype{\consttype{vl}}{A})}{(\datatype{\alpha}{\alpha})}
  }{(
    \uniontype{\consttype{cons}}{\uniontype{\consttype{node}}{\consttype{nil}}}
  )}
}$$ Here $\irectype$ is the recursive type constructor and $\iuniontype$ the
union type constructor. The variable $y$ in the pattern $\dataterm{x}{y}$ will
also be assigned the same type. Note that $\mathsf{upd}$ itself is assigned
type $\functype{(\functype{A}{B})}{(\functype{F_A}{F_B})}$, where $F_X$ is
$\rectype{\alpha}{
  \uniontype{
    \uniontype{(\datatype{\consttype{vl}}{X})}{(\datatype{\alpha}{\alpha})}
  }{(
    \uniontype{\consttype{cons}}{\uniontype{\consttype{node}}{\consttype{nil}}}
  )}
}$.
Thus variables in applicative patterns will be assigned union types.

Recursive types are useful to give static semantics to fixpoint combinators,
which embodies the essence of recursion and thus \emph{path polymorphism}. Together
with unions, they allow to model recursively defined data types. Combining these
ideas with type application allows to define data types in a more intuitive
manner, like for example lists and trees $$
\rectype{\alpha}{
  \uniontype{\consttype{nil}}
            {(\datatype{\datatype{\consttype{cons}}{A}}{\alpha})}
}
\qquad
\rectype{\alpha}{
  \uniontype{\consttype{nil}}{(
    \datatype{\datatype{\datatype{\consttype{node}}{A}}{\alpha}}{\alpha}
  )}
}$$ The advantage of this approach is that the type expression reflects the
structure of the terms that inhabit it
(\cf~Fig.~\ref{fig:typingSchemesForPatternsAndTerms}). This will prove to be
convenient for our proposed notion of \emph{pattern compatibility}.

Compatibility is the key for ensuring Safety (Subject Reduction, SR for short,
and Progress). Consider the following example:
\begin{equation}
(\absterm{\dataterm{\constterm{vl}}{\matchable{x}}}
         {\ifThenElse{x}{1}{0}}
         {\set{x:\bool}})
\icaseterm
(\absterm{\dataterm{\constterm{vl}}{\matchable{y}}}
         {y+1}
         {\set{y:\nat}})
\label{eq:example:compatib:i}
\end{equation}
Although there is a branch capable of handling a term such as
$\appterm{\constterm{vl}}{4}$, namely the second one,  evaluation in $\capp$
takes place in left-to-right order following standard practice in functional
programming languages. Since the term $\appterm{\constterm{vl}}{4}$ \emph{also}
matches the pattern $\dataterm{\constterm{vl}}{\matchable{x}}$, we would obtain
the (incorrect) reduct $\ifThenElse{4}{1}{0}$. We thus must relate the types of
$\dataterm{\constterm{vl}}{\matchable{x}}$ and
$\dataterm{\constterm{vl}}{\matchable{y}}$ in order to avoid failure of SR.
Since $\dataterm{\constterm{vl}}{\matchable{y}}$ is an instance of
$\dataterm{\constterm{vl}}{\matchable{x}}$, we require the type of the latter
to be a subtype of the type of the former since it will always have priority:
$\datatype{\consttype{vl}}{\nat}\subtype \datatype{\consttype{vl}}{\bool}$.
Fortunately, this is not the case since $\nat\centernot \subtype \bool$,
rendering this example untypable.

Consider now, a term such as: 
\begin{equation}
\matchable{f} \ifuncterm_{\set{f:\functype{A}{B}}}
\arraycolsep=1.5pt\def\arraystretch{1}
\begin{array}[t]{cllll}
(          & \dataterm{\constterm{vl}}{\matchable{z}} & \ifuncterm_{\set{z:A}}      & \dataterm{\constterm{vl}}{(\appterm{f}{z})} \\
\icaseterm & \dataterm{\matchable{x}}{\matchable{y}}  & \ifuncterm_{\set{x:C, y:D}} & \dataterm{x}{y})
\end{array}
\label{eq:example:compatib:ii}
\end{equation}
This function takes an argument $\matchable{f}$ and pattern-matches with a data
structure to apply $f$ only when this data structure is an application with the
constructor $\constterm{vl}$ on the left-hand side. Assigning $\matchable{x}$ in
the second branch the type $C = \consttype{vl}$ is a potential source of failure
of SR since the function would accept arguments of type
$\datatype{\consttype{vl}}{D}$. Our proposed notion of compatibility will check
the \emph{types} occurring at offending positions in the \emph{types} of both
patterns. In this case, if $C = \consttype{vl}$ then $\datatype{C}{D} \subtype
\datatype{\consttype{vl}}{A}$ is enforced. Note that if $C$ were a type such as
$\rectype{\alpha}{\uniontype{\consttype{vl}}{\datatype{\alpha}{\alpha}}}$,
then also the same condition would be enforced.    

Let us return to example (\ref{eq:intro:upd}). The type declarations would be
$C = D = \rectype{\alpha}{\!}
\uniontype{
  \uniontype{(\datatype{\consttype{vl}}{A})}{(\datatype{\alpha}{\alpha})}
}{(
  \uniontype{\consttype{cons}}{\uniontype{\consttype{node}}{\consttype{nil}}}
)}$ and $E =
\uniontype{\consttype{cons}}{\uniontype{\consttype{node}}{\consttype{nil}}}$.
We now illustrate how compatibility determines any possible source of failure
of SR. Let us call $p, q$ and $r$ the three patterns of the innermost
abstraction of (\ref{eq:intro:upd}), resp. Since pattern $p$ does not subsume
$q$, we determine the (maximal) positions in both patterns which are sources of
failure of subsumption. In this case, it is that of $\constterm{vl}$ in $p$ and
$\matchable{x}$ in $q$. We now consider the \emph{subtype} at that position in
$\datatype{\consttype{vl}}{A}$, the type of $p$, and the \emph{subtype} at the
same position in $\datatype{F_A}{F_A}$, the type of $q$: the first is
$\consttype{vl}$ and the second is $F_A$. Since $F_A$ does not admit
$\consttype{vl}$ (\cf~Def.~\ref{def:compatibility}), these branches are
immediately declared compatible. In the case of $p$ and $r$, $\epsilon$ is the
offending position in the failure of $p$ subsuming $r$: since the type
application constructor $\idatatype$ located at position $\epsilon$ in
$\datatype{\consttype{vl}}{A}$ is not admitted by $E$, the type of $r$, these
branches are immediately declared compatible. Finally, a similar analysis
between $q$ and $r$ entails that these are compatible too. The type system and
its proof of Safety will therefore assure us that this example preserves
typability.

\paragraph*{Summary of contributions:} 
\begin{itemize}
  \item A typing discipline for \capp. We statically guarantee safety for path
  polymorphism in its purest form (other, more standard forms of polymorphism
  such as parametric polymorphism which we believe to be easier to handle, are
  out of the scope of this paper). 
  
  \item A proof of safety for the resulting system. It relies on the syntactic
  notion of pattern compatibility mentioned above, hence no runtime analysis is
  required. 

  \item Invertibility of subtyping of recursive types. This is crucial for the
  proof of safety. It relies on an equivalent coinductive formulation for which
  invertibility implies invertibility of subtyping of recursive types.
\end{itemize}

\paragraph*{Related work}
The literature on (typed) pattern calculi is extensive; we mention the most
relevant ones (see~\cite{DBLP:journals/jfp/JayK09,Jay:2009} for a more thorough
listing). In~\cite{DBLP:journals/jfp/ArbiserMR09} the constructor calculus is
proposed. It has a different notion of pattern matching: it uses a case
construct $\set{c_1 \mapsto s_1, \ldots, c_n \mapsto s_n} \cdot t$ in which
certain occurrences of the constructors $c_i$ in $t$ are replaced by their
corresponding terms. \cite{DBLP:journals/corr/abs-1009-3429} studies typing to
ensure that these constructor substitutions never block on a constant not in
their domain. Recursive types are not considered (nor is path polymorphism).
Two further closely related efforts merit comments: the first is the work by
Jay and Kesner and the second is that of the $\rho$-calculus by Kirchner and
colleagues.

In~\cite{DBLP:conf/esop/JayK06,DBLP:journals/jfp/JayK09} the Pure Pattern
Calculus (\ppc) is studied. It allows patterns to be computed dynamically (they
may contain free variables). A type system for a \ppc\ like calculus is given
in~\cite{Jay:2009} however neither recursive nor union types are considered.
\cite{Jay:2009} also studies a simple static pattern calculus. However, there
are numerous differing aspects w.r.t. this work among which we can mention the
following. First, the typed version of~\cite{Jay:2009} (the \emph{Query
Calculus}) omits recursive types and union types. Then, although it admits a
form of path polymorphism, this is at the cost of matching types at runtime and
thus changing the operational semantics of the untyped calculus; our system is
purely static, no runtime analysis is required.

The $\rho$-calculus~\cite{DBLP:journals/igpl/CirsteaK01} is a generic pattern
matching calculus parameterized over a matching theory. There has been
extensive work exploring numerous
extensions~\cite{DBLP:conf/rta/CirsteaKL01,DBLP:conf/fossacs/CirsteaKL01,DBLP:journals/entcs/CirsteaKL02,DBLP:conf/types/CirsteaLW03,DBLP:journals/entcs/LiquoriW05,DBLP:conf/popl/BartheCKL03}.
None addresses path polymorphism however. Indeed, none of the above allow
patterns of the form $\dataterm{\matchable{x}}{\matchable{y}}$. This limitation
seems to be due to the alternative approach to typing
$\dataterm{\constterm{c}}{\matchable{x}}$ adopted in the literature on the
$\rho$-calculus where $\constterm{c}$ is assigned a \emph{fixed} functional
type. This approach seems incompatible with path polymorphism, as we see it, in
that it suggests no obvious way of typing patterns of the form
$\dataterm{\matchable{x}}{\matchable{y}}$ where $\matchable{x}$ denotes an
arbitrary piece of unstructured \emph{data}. Additional differences with our
work are:
\begin{itemize}
  \item \cite{DBLP:journals/entcs/CirsteaKL02}: It does not introduce union
  types. No runtime matching error detection takes place (this is achieved via
  Progress in our paper).
  
  \item \cite{DBLP:conf/rta/CirsteaKL01}: It deals with an untyped
  $\rho$-calculus. Hence no SR.
  \item \cite{DBLP:conf/fossacs/CirsteaKL01,DBLP:conf/popl/BartheCKL03}:
  Neither union nor recursive types are considered.
\end{itemize}

\textbf{Structure of the paper.} Sec.~\ref{sec:terms} introduces the terms and
operational semantics of $\capp$. The typing system is developed in
Sec.~\ref{sec:typingSystem} together with a precise definition of
compatibility. Sec.~\ref{sec:safety} studies Safety: SR and Progress. Finally,
we conclude. The document you are reading is the report including full proofs.


\section{Syntax and Operational Semantics of \capp}
\label{sec:terms}

We assume given an infinite set of term variables $\Variable$ and constants $\Constant$. The syntax
of \capp\ consists of four syntactic categories, namely \emphdef{patterns} ($p, q, \ldots$), \emphdef{terms} ($s, t, \ldots$), \emphdef{data structures} ($d, e, \ldots$) and \emphdef{matchable forms} ($m,n,\ldots$): $$
\begin{array}{ll}
\begin{array}[t]{rlll}
p & \Coloneq & \matchable{x}   & \text{(matchable)} \\
  & |        & \constterm{c}   & \text{(constant)} \\
  & |        & \dataterm{p}{p} & \text{(compound)} \\
\end{array}
\quad &
\begin{array}[t]{rlll}
t & \Coloneq  & x                                                                       & \text{(variable)} \\
  & |         & \constterm{c}                                                           & \text{(constant)} \\
  & |         & \appterm{t}{t}                                                          & \text{(application)} \\
  & |         & \absterm{p}{t}{\theta} \icaseterm \dots \icaseterm \absterm{p}{t}{\theta} & \text{(abstraction)}
\end{array}
\\
\\
\begin{array}[t]{rlll}
d & \Coloneq & \constterm{c}   & \text{(constant)} \\
  & |        & \dataterm{d}{t} & \text{(compound)} 
\end{array}
\quad &
\begin{array}[t]{rlll}
m & \Coloneq & d                                                                       & \text{(data structure)} \\
  & |        & \absterm{p}{t}{\theta} \icaseterm \dots \icaseterm \absterm{p}{t}{\theta} & \text{(abstraction)}
\end{array}
\end{array} $$

The set of patterns, terms, data structures and matchable forms are denoted $\Pattern$, $\Term$,
$\DataStructure$ and $\MatchableForms$, resp. Variables occurring in patterns are called
\emphdef{matchables}. We often abbreviate
$\absterm{p_1}{s_1}{\theta_1} \icaseterm \ldots \icaseterm \absterm{p_n}{s_n}{\theta_n}$ with
$(\absterm{p_i}{s_i}{\theta_i})_{i \in 1..n}$. The $\theta_i$ are typing contexts annotating the
type assignments for the variables in $p_i$ (\cf Sec.~\ref{sec:typingSystem}). The \emphdef{free
variables} of a term $t$ (notation $\fv{t}$) are defined as expected; in a pattern $p$ we call
them \emphdef{free matchables} ($\fm{p}$). All free matchables in each $p_i$ are assumed to be
bound in their respective bodies $s_i$. Positions in patterns and terms are defined as expected
and denoted $\pi,\pi',\ldots$ ($\epsilon$ denotes the root position). We write \pos{s} for the
set of positions of $s$ and $\at{s}{\pi}$ for the subterm of $s$ occurring at position $\pi$.

A \emphdef{substitution} ($\sigma, \sigma_i, \ldots$) is a partial function from term variables
to terms.  If it assigns $u_i$ to $x_i$, $i\in 1..n$, then we write $\{u_1/x_1,\ldots,u_n/x_n\}$.
Its domain (\dom{\sigma}) is $\set{x_1,\ldots,x_n}$. Also, $\set{}$ is the identity substitution.
We write $\sigma s$ for the result of applying $\sigma$ to term $s$. \emphdef{Matchable forms}
are required for defining the \emphdef{matching operation}, described next. 

Given a pattern $p$ and a term $s$, the matching operation $\basicmatch{p}{s}$ determines whether
$s$ matches $p$. It may have one of three outcomes: success, fail (in which case it returns the
special symbol $\fail$) or undetermined (in which case it returns the special symbol $\wait$).
We say $\basicmatch{p}{s}$ is \emphdef{decided} if it is either successful or it fails. In the
former it yields a substitution $\sigma$; in this case we write $\basicmatch{p}{s} = \sigma$.
The disjoint union of matching outcomes is given as follows (``\eqdef'' is used for definitional equality): $$
\begin{array}{c}
\begin{array}{rll}
\fail\uplus o            & \eqdef & \fail \\
o \uplus \fail           & \eqdef & \fail \\
\sigma_1 \uplus \sigma_2 & \eqdef & \sigma
\end{array}
\hspace{.5cm}
\begin{array}{rll}
\wait \uplus \sigma & \eqdef & \wait \\
\sigma \uplus \wait & \eqdef & \wait \\
\wait \uplus \wait  & \eqdef & \wait
\end{array}
\end{array} $$
where $o$ denotes any possible output and $\sigma_1\uplus\sigma_2 \eqdef \sigma$ if the domains of
$\sigma_1$ and $\sigma_2$ are disjoint. This always holds given that patterns are assumed to be
linear (at most one occurrence of any matchable). The matching operation is defined as follows,
where the defining clauses below are evaluated from top to bottom\footnote{This is simplification
to the static patterns case of the matching operation introduced in~\cite{DBLP:journals/jfp/JayK09}.}: $$
\begin{array}{llll}
\basicmatch{\matchable{x}}{u}                 & \eqdef & \rename{x}{u} \\
\basicmatch{\constterm{c}}{\constterm{c}}     & \eqdef & \set{} \\
\basicmatch{\dataterm{p}{q}}{\dataterm{u}{v}} & \eqdef & \basicmatch{p}{u} \uplus \basicmatch{q}{v} &
  \quad \text{if $\dataterm{u}{v}$ is a \emph{matchable form}} \\
\basicmatch{p}{u}                             & \eqdef & \fail &
  \quad \text{if $u$ is a \emph{matchable form}} \\
\basicmatch{p}{u}                             & \eqdef & \wait
\end{array} $$
For example: $\basicmatch{\constterm{c}}{\absterm{\matchable{x}}{s}{}} = \fail$;
$\basicmatch{\constterm{c}}{\constterm{d}} = \fail$; $\basicmatch{\constterm{c}}{x} = \wait$
and $\basicmatch{\dataterm{\constterm{c}}{\constterm{c}}}{\dataterm{x}{\constterm{d}}} = \fail$.
We now turn to the only reduction axiom of $\capp$: $$
\begin{array}{c}
\Rule{\match{p_i}{u} = \fail \text{ for all } i < j \quad \match{p_j}{u} = \sigma_j \quad j \in 1..n}
     {\appterm{(\absterm{p_i}{s_i}{\theta_i})_{i \in 1..n}}{u} \reduce \sigma_j s_j}
     {\ruleBeta}
\end{array} $$
It may be applied under any context and states that if the argument $u$ to an abstraction
$(\absterm{p_i}{s_i}{\theta_i})_{i\in 1..n}$ fails to match all patterns $p_i$ with $i<j$ and 
successfully matches pattern $p_j$ (producing a substitution $\sigma_j$), then the term
$\appterm{(\absterm{p_i}{s_i}{\theta_i})_{i \in 1..n}}{u}$ reduces to $\sigma_j s_j$. 

The following example illustrates the use of the reduction rule and the matching operation:
\begin{equation}
\arraycolsep=1.4pt\def\arraystretch{1}
\begin{array}[t]{l}
\appterm{(\absterm{\constterm{true}}{1}{} \icaseterm \absterm{\constterm{false}}{0}{})}{(\appterm{(\absterm{\constterm{true}}{\constterm{false}}{} \icaseterm \absterm{\constterm{false}}{\constterm{true}}{})}{\constterm{true}})} \\
\begin{array}{rcll}
\qquad & \reduce & \appterm{(\absterm{\constterm{true}}{1}{} \icaseterm \absterm{\constterm{false}}{0}{})}{\match{\constterm{true}}{\constterm{true}}\ \constterm{false}} \\
       & =       & \appterm{(\absterm{\constterm{true}}{1}{} \icaseterm \absterm{\constterm{false}}{0}{})}{\constterm{false}} \\
       & \reduce & \match{\constterm{false}}{\constterm{false}}\ 0 & \quad\match{\constterm{true}}{\constterm{false}} = \fail \\
       & =       & 0
\end{array}
\end{array}
\label{eq:terms:ejemploReduccion}
\end{equation}

\begin{proposition}
Reduction in \capp\ is confluent (CR).
\end{proposition}

This result follows from a straightforward adaptation of the CR proof presented
in~\cite{DBLP:journals/jfp/JayK09} to our calculus. The key step is proving that the matching
operation satifies the \emph{Rigid Matching Condition (RMC)} proposed in the cited work.
Note that $\capp$ is just the static patterns fragment of $\ppc$ where instead of the usual
abstraction we have alternatives (\ie we abstract multiple branches with the same constructor).
Our contribution is on the typed variant of the calculus.


\section{Typing System}
\label{sec:typingSystem}

This section presents $\mu$-types, the finite type expressions that shall be used for typing
terms in \capp, their associated notions of equivalence and subtyping and then the
typing schemes. Also, further examples and definitions associated to compatibility are included. 

\subsection{Types}
\label{sec:typingMu}

In order to ensure that patterns such as $\dataterm{\matchable{x}}{\matchable{y}}$ decompose only data structures rather than arbitrary terms, we shall introduce two sorts of typing expressions: \emph{types} and \emph{datatypes}, the latter being strictly included in the former. 
%
We assume given countably infinite sets $\DataTypeVariable$ of \emphdef{datatype variables} ($\alpha, \beta, \ldots$), $\TypeVariable_{A}$ of \emphdef{type variables} $(X, Y, \ldots$) and $\TypeConstant$ of \emphdef{type constants} $(\consttype{c}, \consttype{d}, \ldots)$. We define
$\TypeVariable \eqdef  \TypeVariable_{A} \cup \DataTypeVariable$ and use metavariables $V, W, \ldots$ to denote an arbitrary element in it. Likewise, we write $a, b, \ldots$ for elements in $\TypeVariable \cup \TypeConstant$. The sets $\DataType$ of \emphdef{$\mu$-datatypes} and $\Type$ of \emphdef{$\mu$-types}, resp., are inductively defined as follows: $$
\begin{array}{cc}
\begin{array}{rlll}
D & \Coloneq & \alpha              & \quad \text{(datatype variable)} \\
  & |        & \consttype{c}       & \quad \text{(atom)} \\
  & |        & \datatype{D}{A}     & \quad \text{(compound)} \\
  & |        & \uniontype{D}{D}    & \quad \text{(union)} \\
  & |        & \rectype{\alpha}{D} & \quad \text{(recursion)} 
\end{array}
& \qquad
\begin{array}{rlll}
A & \Coloneq & X                & \quad \text{(type variable)} \\
  & |        & D                & \quad \text{(datatype)} \\
  & |        & \functype{A}{A}  & \quad \text{(type abstraction)} \\
  & |        & \uniontype{A}{A} & \quad \text{(union)} \\
  & |        & \rectype{X}{A}   & \quad \text{(recursion)}
\end{array}
\end{array} $$

\begin{remark}
A type of the form $\rectype{\alpha}{A}$ is not valid in general since it may produce invalid unfoldings. For example, $\rectype{\alpha}{\functype{\alpha}{\alpha}} = \functype{(\rectype{\alpha}{\functype{\alpha}{\alpha}})}{(\rectype{\alpha}{\functype{\alpha}{\alpha}})}$. On the other hand, types of the form $\rectype{X}{D}$ are not necessary since they denote the solution to the equation $X = D$, hence $X$ is a variable representing a datatype.
\end{remark}

We consider $\iuniontype$ to bind tighter than $\ifunctype$, while $\idatatype$ binds
tighter than $\iuniontype$. Therefore $\functype{\uniontype{\datatype{D}{A}}{A'}}{B}$ means
$\functype{(\uniontype{(\datatype{D}{A})}{A'})}{B}$. Additionally, when refering to a finite series of
consecutive unions such as $\uniontype{\uniontype{A_1}{\ldots}}{A_n}$ we will use the simplified
notation $\maxuniontype{i \in 1..n}{A_i}$. This notation is not strict on how subexpressions $A_i$
are associated hence, in principle, it refers to any of all possible associations. In the next
section we present an equivalence relation on $\mu$-types that will identify all these associations. 
We often write $\rectype{V}{A}$ to mean either $\rectype{\alpha}{D}$ or $\rectype{X}{A}$. A \emphdef{non-union $\mu$-type} $A$ is a $\mu$-type of one of the following forms: $\alpha$, $\consttype{c}$, $\datatype{D}{A}$, $X$, $\functype{A}{B}$ or $\rectype{V}{A}$ with $A$ a non-union $\mu$-type. We assume $\mu$-types are \emphdef{contractive}: $\rectype{V}{A}$ is contractive if $V$ occurs in $A$ only under a type constructor $\ifunctype$ or $\idatatype$, if at all. We henceforth redefine $\Type$ to be the set of \emphdef{contractive $\mu$-types}. $\mu$-types come equipped with a notion of equivalence $\eqtypemu$ and subtyping $\subtypemu$.

\begin{figure}[t] $$
\begin{array}{c}
\Rule{}
     {\sequTE{}{A \eqtypemu A}}
     {\ruleEqmuRefl}
\quad
\Rule{\sequTE{}{A \eqtypemu B}
      \quad
      \sequTE{}{B \eqtypemu C}}
     {\sequTE{}{A \eqtypemu C}}{\ruleEqmuTrans}
\quad
\Rule{\sequTE{}{A \eqtypemu B}}
     {\sequTE{}{B \eqtypemu A}}
     {\ruleEqmuSymm}
\\
\\
\Rule{\sequTE{}{A \eqtypemu A'}
      \quad
      \sequTE{}{B \eqtypemu B'}}
     {\sequTE{}{\functype{A}{B} \eqtypemu \functype{A'}{B'}}}
     {\ruleEqmuFunc}
\qquad
\Rule{\sequTE{}{D \eqtypemu D'}
      \quad
      \sequTE{}{A \eqtypemu A'}}
     {\sequTE{}{\datatype{D}{A} \eqtypemu \datatype{D'}{A'}}}
     {\ruleEqmuComp}
\\
\\
\Rule{}
     {\sequTE{}{\uniontype{A}{A} \eqtypemu A}}
     {\ruleEqmuUnionIdem}
\qquad
\Rule{}
     {\sequTE{}{\uniontype{A}{B} \eqtypemu \uniontype{B}{A}}}
     {\ruleEqmuUnionComm}
\\
\\
\Rule{}
     {\sequTE{}{\uniontype{A}{(\uniontype{B}{C})} \eqtypemu \uniontype{(\uniontype{A}{B})}{C}}}
     {\ruleEqmuUnionAssoc}
\\
\\
\Rule{\sequTE{}{A \eqtypemu A'}
      \quad
      \sequTE{}{B \eqtypemu B'}}
     {\sequTE{}{\uniontype{A}{B} \eqtypemu \uniontype{A'}{B'}}}
     {\ruleEqmuUnion}
\qquad
\Rule{\sequTE{}{A \eqtypemu B}}
     {\sequTE{}{\rectype{V}{A} \eqtypemu \rectype{V}{B}}}
     {\ruleEqmuRec}
\\
\\
\Rule{}
     {\sequTE{}{\rectype{V}{A} \eqtypemu \substitute{V}{\rectype{V}{A}}{A}}}
     {\ruleEqmuFold}
\quad
\Rule{\sequTE{}{A \eqtypemu \substitute{V}{A}{B}}
      \quad
      \rectype{V}{B} \text{ contractive}}
     {\sequTE{}{A \eqtypemu \rectype{V}{B}}}
     {\ruleEqmuContr}
\end{array} $$
\caption{Type equivalence for $\mu$-types}
\label{fig:equivalenceSchemesMu}
\end{figure}

\begin{figure}[t] $$
\begin{array}{c}
\Rule{}
     {\sequTE{\Sigma}{A \subtypemu A}}
     {\ruleSubmuRefl}
\qquad
\Rule{}
     {\sequTE{\Sigma, V \subtypemu W}{V\subtypemu W}}
     {\ruleSubmuHyp}
\qquad
\Rule{\sequTE{}{A \eqtypemu B}}
     {\sequTE{\Sigma}{A \subtypemu B}}
     {\ruleSubmuEq}
\\
\\
\Rule{\sequTE{\Sigma}{A \subtypemu B} 
      \quad
      \sequTE{\Sigma}{B \subtypemu C}}
     {\sequTE{\Sigma}{A \subtypemu C}}
     {\ruleSubmuTrans} 
\qquad
\Rule{\sequTE{\Sigma}{D \subtypemu D'} 
      \quad 
      \sequTE{\Sigma}{A \subtypemu A'}}
     {\sequTE{\Sigma}{\datatype{D}{A} \subtypemu \datatype{D'}{A'}}}
     {\ruleSubmuComp}
\\
\\
\Rule{\sequTE{\Sigma}{A \subtypemu A'}
      \quad 
      \sequTE{\Sigma}{B \subtypemu B'}}
     {\sequTE{\Sigma}{\functype{A'}{B} \subtypemu \functype{A}{B'}}}
     {\ruleSubmuFunc} 
\qquad
\Rule{\sequTE{\Sigma}{A \subtypemu C}
      \quad 
      \sequTE{\Sigma}{B \subtypemu C}}
     {\sequTE{\Sigma}{\uniontype{A}{B} \subtypemu C}}
     {\ruleSubmuUnionL}
\\
\\
\Rule{\sequTE{\Sigma}{A \subtypemu B}}
     {\sequTE{\Sigma}{A \subtypemu \uniontype{B}{C}}}
     {\ruleSubmuUnionRL}
\qquad
\Rule{\sequTE{\Sigma}{A \subtypemu C}}
     {\sequTE{\Sigma}{A \subtypemu \uniontype{B}{C}}}
     {\ruleSubmuUnionRR} 
\\
\\
\Rule{\sequTE{\Sigma, V \subtypemu W}{A \subtypemu B}
      \quad
      W \notin \fv{A} \quad V \notin \fv{B}
     }
     {\sequTE{\Sigma}{\rectype{V}{A} \subtypemu \rectype{W}{B}}}
     {\ruleSubmuRec}
\end{array} $$
\caption{Strong subtyping for $\mu$-types}\label{fig:subtypingSchemesMu}
\end{figure}

\begin{definition}
\begin{enumerate}
\item $\eqtypemu$ is defined by the schemes in Fig.~\ref{fig:equivalenceSchemesMu}.
\item $\subtypemu$ is defined by the schemes in Fig.~\ref{fig:subtypingSchemesMu} where a subtyping context $\Sigma$ is a set of assumptions over type variables of the form $V \subtypemu W$ with $V, W \in \TypeVariable$. 
\end{enumerate}
\end{definition}

$\ruleEqmuRec$ actually encodes two rules, one for datatypes ($\rectype{\alpha}{D}$) and one for arbitrary types ($\rectype{X}{A}$). Likewise for $\ruleEqmuFold$ and $\ruleEqmuContr$. The relation resulting from dropping $\ruleEqmuContr$~\cite{DBLP:journals/fuin/AriolaK96,DBLP:journals/fuin/BrandtH98} is called weak type equivalence~\cite{DBLP:conf/caap/Cardone92} and is known to be too weak to capture equivalence of its coinductive formulation (required for our proof of invertibility of subtyping \cf Prop.~\ref{prop:subtypingIsInvertible}); for example, types $\rectype{X}{\functype{A}{\functype{A}{X}}}$ and $\rectype{X}{\functype{A}{X}}$ cannot be equated.

Regarding the subtyping rules, we adopt those for union of~\cite{DBLP:conf/csl/Vouillon04}. It should be noted that the na\"ive variant of $\ruleSubmuRec$ in which \sequTE{\Sigma}{\rectype{V}{A} \subtypemu \rectype{V}{B}} is deduced from \sequTE{\Sigma}{A\subtypemu B}, is known to be unsound~\cite{DBLP:journals/toplas/AmadioC93}. We often abbreviate \sequTE{}{A\subtypemu B} as $A\subtypemu B$.

We can now use notation $\maxuniontype{i \in I}{A_i}$ on contractive $\mu$-types to denote several consecutive applications of the binary operator $\iuniontype$ irrespective of how they are associated. All such associations yield equivalent $\mu$-types. Such expressions will be useful to prove the correspondence between the types as trees formulation and the contractive $\mu$-types of the current section. To that end we introduce the following lemmas that extend the associative, commutative and idempotent properties to arbitrary unions.

To simplify the presentation of the proofs, we often resort to the following reasoning (or its symmetric variant) $$
\Rule{
  \Rule{\vdots}{A \eqtypemu B}{X}
  \quad
  \Rule{}{C \eqtypemu C}{\ruleEqmuRefl}
}{\uniontype{A}{C} \eqtypemu \uniontype{B}{C}}{\ruleEqmuUnion}
$$ by only stating $\ruleName{X}$ (\ie a rule, lemma, inductive hypothesis, etc.). Thus, we say that $\uniontype{A}{C} \eqtypemu \uniontype{B}{C}$ by $\ruleName{X}$ or, in other words, apply $\ruleName{X}$ within a \emph{union context}.

\begin{lemma}
\label{lem:unionIsAssociative}
Let $A$ and $A'$ be two distinct associations of $\maxuniontype{i \in 1..n}{A_i}$. Then, $A \eqtypemu A'$.
\end{lemma}

\begin{proof}
Direct consequence of $\ruleEqmuUnionAssoc$.
%
  
\end{proof}

\begin{lemma}
\label{lem:unionIsCommutative}
Let $p$ be a permutation over $1..n$. Then, $\maxuniontype{i \in 1..n}{A_i} \eqtypemu \maxuniontype{i \in 1..n}{A_{p(i)}}$.
\end{lemma}

\begin{proof}
By induction on $n$. 
\begin{itemize}
  \item $n = 1$. This case is immediate since $p = \mathit{id}$.
  
  \item $n > 1$. Without loss of generality we can consider $p$ to be the function $$p(i) = \left\{ 
\begin{array}{ll}
p'(i)   & \quad \text{if $i < k$} \\
n       & \quad \text{if $i = k$} \\
p'(i-1) & \quad \text{if $i > k$}
\end{array} \right.$$ where $p'$ is a permutation over $1..n-1$ and $k \in 1..n$. That is, $p$ permutes $k$ with $n$ and behaves like $p'$ on every other position. Then, $$
\begin{array}{rcll}
\maxuniontype{i \in 1..n}{A_i} & \eqtypemu & \uniontype{(\maxuniontype{i \in 1..n-1}{A_i})}{A_n} & \quad \text{Lem.~\ref{lem:unionIsAssociative}} \\
                               & \eqtypemu & \uniontype{(\maxuniontype{i \in 1..n-1}{A_{p'(i)}})}{A_n} & \quad \text{by IH}
\end{array} $$
  If $k = n$ we are done, since $\uniontype{(\maxuniontype{i \in 1..n-1}{A_{p'(i)}})}{A_n} \eqtypemu \maxuniontype{i \in 1..n}{A_{p(i)}}$ by Lem.~\ref{lem:unionIsAssociative}. If not (\ie $k \in 1..n-1$) we just need to apply $\ruleEqmuUnionComm$ to the proper subexpression $$
\begin{array}{rcll}
\maxuniontype{i \in 1..n}{A_i} & \eqtypemu & \uniontype{(\maxuniontype{i \in 1..n-1}{A_{p'(i)}})}{A_n} \\
                               & \eqtypemu & \uniontype{(\maxuniontype{i \in 1..k-1}{A_{p'(i)}})}{(\uniontype{(\maxuniontype{i \in k..n-1}{A_{p'(i)}})}{A_n})} & \quad \text{Lem.~\ref{lem:unionIsAssociative}} \\
                               & \eqtypemu & \uniontype{(\maxuniontype{i \in 1..k-1}{A_{p'(i)}})}{(\uniontype{A_n}{(\maxuniontype{i \in k..n-1}{A_{p'(i)}})})} & \quad \text{\ruleEqmuUnionComm}\\
                               & \eqtypemu & \maxuniontype{i \in 1..n}{A_{p(i)}} & \quad \text{Lem.~\ref{lem:unionIsAssociative}} \\
\end{array} $$
\end{itemize}
\end{proof}

\begin{lemma}
\label{lem:unionIsIdempotent}
Let $J_m = \pair{J}{m}$ be a finite multiset\footnote{Recall that a \emph{multiset} is a pair $\M = \pair{\X}{m}$ where $\X$ is de \emph{underlying set} of $\M$ and $m : \X \to \Natural$ is its \emph{multiplicity function}. We will usually denote $\M$ with $\X$ when there is no ambiguity or the meaning is clear from the context.} such that $J \subseteq 1..n$, then $\maxuniontype{i \in 1..n}{A_i} \eqtypemu \uniontype{(\maxuniontype{i \in 1..n}{A_i})}{(\maxuniontype{j \in J_m}{A_j})}$.
\end{lemma}

\begin{proof}
This proof is by induction on $\#(J_m)$ (the cardinality of the multiset $J_m$).
\begin{itemize}
  \item $\#(J_m) = 0$. This case is immediate by Lem.~\ref{lem:unionIsAssociative} (note that both sides of the equivalence may be associated differently, thus $\ruleEqmuRefl$ is not enough).
  
  \item $\#(J_m) > 0$. Let $k \in J_m$. Then $$\kern-2em
\begin{array}{rcll}
\uniontype{(\maxuniontype{i \in 1..n}{A_i})}{(\maxuniontype{j \in J_m}{A_j})}
  & \eqtypemu & \uniontype{(\uniontype{(\maxuniontype{i \in 1..n}{A_i})}{(\maxuniontype{j \in (J_m \setminus \set{k})}{A_j})})}{A_k} & \enskip \text{Lem.~\ref{lem:unionIsCommutative}} \\
  & \eqtypemu & \uniontype{(\maxuniontype{i \in 1..n}{A_i})}{A_k}                                & \enskip \text{by IH} \\
  & \eqtypemu & \uniontype{(\maxuniontype{i \in 1..n \\ i \neq k}{A_i})}{(\uniontype{A_k}{A_k})} & \enskip \text{Lem.~\ref{lem:unionIsCommutative}} \\
  & \eqtypemu & \uniontype{(\maxuniontype{i \in 1..n \\ i \neq k}{A_i})}{A_k}                    & \enskip \text{\ruleEqmuUnionIdem} \\
  & \eqtypemu & \maxuniontype{i \in 1..n}{A_i}                                                   & \enskip \text{Lem.~\ref{lem:unionIsCommutative}}
\end{array} $$
\end{itemize}
\end{proof}

The following lemma presents an admissible rule regarding union types that shall be used later to relate $\eqtypemu$ with its
coinductive characterisation. Note that in this case there is no need for types $A_i, B_j$ to be
non-union types below.

\begin{lemma}
\label{lem:unionEquivalence}
Let $A = \maxuniontype{i \in 1..n}{A_i}, B = \maxuniontype{j \in 1..m}{B_j}$ and $ f : 1..n \to 1..m$, $g : 1..m \to 1..n$ functions such that $A_i \eqtypemu B_{f(i)}$ and $A_{g(j)} \eqtypemu B_j$ for every $i \in 1..n, j \in 1..m$. Then, $\maxuniontype{i \in 1..n}{A_i} \eqtypemu \maxuniontype{j \in 1..m}{B_j}$.
\end{lemma}

\begin{proof}
It is immediate to see that for every multiset of indexes $I \subseteq 1..n$, $\maxuniontype{i \in I}{A_i} \eqtypemu \maxuniontype{i \in I}{B_{f(i)}}$, by applying $\ruleEqmuUnion$ as many times as needed and resorting to Lem.~\ref{lem:unionIsAssociative} if necessary. Similarly, $\maxuniontype{j \in J}{B_j} \eqtypemu \maxuniontype{j \in J}{A_{g(j)}}$ for $J \subseteq 1..m$. So lets consider some multisets and see how they relate to each other to finish our analysis $$
\begin{array}{rcl}
I  & \eqdef & \set{i    \mathrel| i \in 1..n, i \in    \img{g}} \\
I' & \eqdef & \set{i    \mathrel| i \in 1..n, i \notin \img{g}} \\
G  & \eqdef & \set{g(j) \mathrel| j \in 1..m} \\
F  & \eqdef & \set{f(i) \mathrel| i \in 1..n, i \notin \img{g}}
\end{array} $$
  First notice that, by definition, $I$ and $I'$ have no repeated elements and
  \begin{equation}
    \label{eq:eqiso:g}
    G = I \cup G' \quad \text{with} \quad G' \subseteq I
  \end{equation}
where $G'$ simply holds the repeated elements of $G$. Additionaly we have
  \begin{equation}
    \label{eq:eqiso:f}
    F \subseteq 1..m
  \end{equation}
Finally, we can conclude by resorting to some previous results $$
\begin{array}{rcll}
A & =         & \maxuniontype{i \in 1..n}{A_i} \\
  & \eqtypemu & \uniontype{(\maxuniontype{i \in I}{A_i})}{(\maxuniontype{i \in I'}{A_i})} & \quad \text{Lem.~\ref{lem:unionIsCommutative}} \\
  & \eqtypemu & \uniontype{(\uniontype{(\maxuniontype{i \in I}{A_i})}{(\maxuniontype{i \in G'}{A_i})})}{(\maxuniontype{i \in I'}{A_i})} & \quad \text{Lem.~\ref{lem:unionIsIdempotent} with (\ref{eq:eqiso:g})} \\
  & =         & \uniontype{(\maxuniontype{i \in G}{A_i})}{(\maxuniontype{i \in I'}{A_i})} \\
  & =         & \uniontype{(\maxuniontype{j \in 1..m}{A_{g(j)}})}{(\maxuniontype{i \in I'}{A_i})} \\
  & \eqtypemu & \uniontype{(\maxuniontype{j \in 1..m}{B_j})}{(\maxuniontype{i \in I'}{A_i})} \\
  & =         & \uniontype{(\maxuniontype{j \in 1..m}{B_j})}{(\maxuniontype{i \in 1..n \\ i \notin \img{g}}{A_i})} \\
  & \eqtypemu & \uniontype{(\maxuniontype{j \in 1..m}{B_j})}{(\maxuniontype{i \in 1..n \\ i \notin \img{g}}{B_{f(i)}})} \\
  & =         & \uniontype{(\maxuniontype{j \in 1..m}{B_j})}{(\maxuniontype{j \in F}{B_j})} \\
  & \eqtypemu & \maxuniontype{j \in 1..m}{B_j} & \quad \text{Lem.~\ref{lem:unionIsIdempotent} with (\ref{eq:eqiso:f})} \\
  & =         & B
\end{array} $$
\end{proof}


\subsubsection{Types as trees}
\label{sec:typingCoinductive}

Type safety, addressed in the Sec.\ref{sec:safety}, also relies on $\subtypemu$
enjoying the fundamental property of \emph{invertibility} of non-union types
(\cf Prop.~\ref{prop:subtypingIsInvertible}):
\begin{enumerate}
  \item If $\datatype{D}{A} \subtypemu \datatype{D'}{A'}$, then $D \subtypemu
  D'$ and $A \subtypemu A'$.
  \item If $\functype{A}{B} \subtypemu \functype{A'}{B'}$, then $A' \subtypemu
  A$ and $B \subtypemu B'$.
\end{enumerate}

To prove this we appeal to the standard tree interpretation of terms and
formulate an equivalent coinductive definition of equivalence and subtyping
($\subtypeco$). For the latter, invertibility of non-union types is proved
coinductively, (Lem.~\ref{lem:subtypingIsInvertible}), entailing Prop.~\ref{prop:subtypingIsInvertible}.

Consider \emphdef{type constructors} $\idatatype$ and $\ifunctype$ together
with \emphdef{type connector} $\iuniontype$ and the ranked alphabet
$\mathfrak{L} \eqdef \set{a^0 \mathrel| a \in \TypeVariable \cup \TypeConstant}
\cup \set{\idatatype^2, \ifunctype^2, \iuniontype^2}$.  We write $\Tree$ for
the set of (possibly) \emphdef{infinite types} with symbols in $\mathfrak{L}$.
This is a standard construction~\cite{terese03,journals/tcs/Courcelle83} given
by the metric completion based on a simple depth function measuring the
distance from the root to the minimum conflicting node in two trees.
Perhaps worth mentioning is that the type connector $\iuniontype$ does not
contribute to the depth (hence the reason for calling it a connector rather
than a constructor) excluding types consisting of infinite branches of
$\iuniontype$, such as $\uniontype{(\uniontype{\ldots}{\ldots})}{(\uniontype{\ldots}{\ldots})}$,
from $\Tree$. We use meta-variables $\tA, \tB, \ldots$ to denote elements of
$\Tree$.

\begin{remark}
\label{rem:maximalUnionTypes}
For any $\star \in \mathfrak{L}$, we write $\tA \neq \star$ to mean that
$\tA(\epsilon) \neq \star$, $\epsilon$ being the root position of the tree. For
example, $A \neq \iuniontype$ means that $A$ is a \emph{non-union type}. Any
type $\tA$ can be written as $\tA = \maxuniontype{i \in 1..n}{\tA_i}$ (dubbed a
\emph{maximal} union type) where $\tA_i \neq \iuniontype$ for all $i \in 1..n$
with $n\in \Natural$, irrespective of how their arguments are associated. All
such associations yield equivalent infinite types in a sense to be made precise
shortly.
\end{remark}

\subsubsection{Equivalence of Infinite Types}

\begin{figure}[t] $$
\begin{array}{c}
\RuleCo{}
       {a \eqtypeco a}
       {\ruleEqcoRefl}
\\
\\
\RuleCo{\tA \eqtypeco \tA' \quad \tB \eqtypeco \tB'}
       {\functype{\tA}{\tB} \eqtypeco \functype{\tA'}{\tB'}}
       {\ruleEqcoFunc}
\qquad
\RuleCo{\tD \eqtypeco \tD' \quad \tA \eqtypeco \tA'}
       {\datatype{\tD}{\tA} \eqtypeco \datatype{\tD'}{\tA'}}
       {\ruleEqcoComp}
\\
\\
\RuleCo{\begin{array}{ll}
          \tA_i \eqtypeco \tB_{f(i)} & \quad f : 1..n \to 1..m \\
          \tA_{g(j)} \eqtypeco \tB_j & \quad g : 1..m \to 1..n
        \end{array}
        \quad
        \tA_i, \tB_j \neq \iuniontype
        \quad
        n + m > 2}
       {\maxuniontype{i \in 1..n}{\tA_i} \eqtypeco \maxuniontype{j \in 1..m}{\tB_j}}
       {\ruleEqcoUnion}
\end{array} $$
\caption{Equivalence relation for infinite types}
\label{fig:equivalenceSchemesCo}
\end{figure}

\begin{definition}
Infinite type equivalence, written $\mathbin{\eqtypeco}$, is defined by the
coinductive interpretation of the schemes of
Fig.~\ref{fig:equivalenceSchemesCo}.
\end{definition}

Note that $\ruleEqcoUnion$ is actually a rule scheme, representing all
possible associations within maximal union types $\tA =
\maxuniontype{i \in 1..n}{\tA_i}$ and $\tB = \maxuniontype{j \in 1..m}{\tB_j}$.
Each instance of the rule states that every $A_i$ must be equivalent to some
$\tB_j$ via a function $f : 1..n \to 1..m$ and vice versa (with $g: 1..m \to
1..n$). Note that the  type connector $\iuniontype$ is seen to be not only
associative and commutative but also idempotent.

Formally, let $\Phieqtypeco : \Parts{\Tree \times \Tree} \to \Parts{\Tree
\times \Tree}$ be the functional associated to the rules in
Fig.~\ref{fig:equivalenceSchemesCo}, defined as follows: $$
\begin{array}{rcl}
\Phieqtypeco(\S) & =    & \set{\pair{a}{a} \mathrel| a \in \TypeVariable \cup \TypeConstant} \\
                 & \cup & \set{\pair{\datatype{\tD}{\tA}}{\datatype{\tD'}{\tA'}} \mathrel| \pair{\tD}{\tD'}, \pair{\tA}{\tA'} \in \S} \\
                 & \cup & \set{\pair{\functype{\tA}{\tB}}{\functype{\tA'}{\tB'}} \mathrel| \pair{\tA}{\tA'}, \pair{\tB}{\tB'} \in \S} \\
                 & \cup & \{\pair{\maxuniontype{i \in 1..n}{\tA_i}}{\maxuniontype{j \in 1..m}{\tB_j}} \mathrel| \tA_i, \tB_j \neq \iuniontype, n + m > 2 \\
                 &      & \qquad \exists f : 1..n \to 1..m \text{ s.t. } \pair{\tA_i}{\tB_{f(i)}} \in \S, \\
                 &      & \qquad \exists g : 1..m \to 1..n \text{ s.t. } \pair{\tA_{g(j)}}{\tB_j} \in \S\}
\end{array} $$

Then $\mathbin{\eqtypeco} \eqdef \nu\Phieqtypeco$. Now we show that it is
indeed an equivalence relation.

\begin{lemma}
$\eqtypeco$ is an equivalence relation (\ie reflexive, symmetric and
transitive).
\end{lemma}

\begin{proof}
The three properties are proved be showing that the sets defining them are
$\Phieqtypeco$-dense. Then we conclude by the coinductive
principle\footnote{\emph{Coinductive principle}: if $\X$ is $\Phi$-dense, then
$\X \subseteq \nu\Phi$.} that the properties hold on $\eqtypeco$.
\begin{itemize}
  \item Reflexivity: $\Refl\ \eqdef \set{\pair{\tA}{\tA} \mathrel| \tA \in
  \Tree}$. Let $\pair{\tA}{\tA} \in \Refl$. We proceed by analyzing the shape
  of $\tA$:
  \begin{itemize}
    \item $\tA = a$. Immediate since $\pair{a}{a} \in \Phieqtypeco(\Refl)$ for
    every $a \in \TypeVariable \cup \TypeConstant$.
    
    \item $\tA = \datatype{\tD}{\tA'}$. By definition of reflexivity
    $\pair{\tD}{\tD}, \pair{\tA'}{\tA'} \in \Refl$. Then $\pair{A}{A} \in
    \Phieqtypeco\!\left(\Refl\right)$.
    
    \item $\tA = \functype{\tA'}{\tA''}$. Similarly to the previous case, we
    have $\pair{\tA'}{\tA'}, \pair{\tA''}{\tA''} \in \Refl$. Hence
    $\pair{\tA}{\tA} \in \Phieqtypeco\!\left(\Refl\right)$.
    
    \item $\tA = \maxuniontype{i \in 1..n}{\tA_i}$ with $\tA_i \neq \iuniontype$
    for $i \in i..n, n > 1$. Then, since $\pair{\tA_i}{\tA_i} \in \Refl$ and
    $n + n > 2$, we conclude
    $\pair{\maxuniontype{i \in 1..n}{\tA_i}}{\maxuniontype{i \in 1..n}{\tA_i}}
    \in \Phieqtypeco\!\left(\Refl\right)$ by considering $f = g = \mathit{id}$
    (the identity function).
  \end{itemize}

  \item Symmetry: $\Symm\!\left(\S\right) \eqdef \set{\pair{\tB}{\tA} \mathrel|
  \pair{\tA}{\tB} \in \S}$. We show that $\Symm\!\left(\eqtypeco\right)
  \subseteq \mathbin{\eqtypeco}$.
  
  Let $\pair{\tA}{\tB} \in \Symm\!\left(\eqtypeco\right)$, then
  $\pair{\tB}{\tA} \in \mathbin{\eqtypeco} =
  \Phieqtypeco\!\left(\eqtypeco\right)$. By Rem.~\ref{rem:maximalUnionTypes} we
  can consider maximal union types $$
\begin{array}{r@{\quad\text{with}\quad}l}
\tA = \maxuniontype{i \in 1..n}{\tA_i} & \tA_i \neq \iuniontype, i \in 1..n \\
\tB = \maxuniontype{j \in 1..m}{\tB_j} & \tB_j \neq \iuniontype, j \in 1..m
\end{array} $$ and we have two separate cases to analyze:
  \begin{enumerate}
    \item If $n = m = 1$, then both $\tA$ and $\tB$ are non-union types. Now we
    proceed by analyzing the shape of $\tB$:
    \begin{itemize}
      \item $\tB = a$. Then $\tA = a$ by definition of $\Phieqtypeco$ and the
      result is immediate since $\pair{a}{a} \in
      \Phieqtypeco\!\left(\Symm(\eqtypeco)\right)$ for every $a \in
      \TypeVariable \cup \TypeConstant$.
      
      \item $\tB = \datatype{\tD'}{\tB'}$. Again, by definition, we have
      $\tA = \datatype{\tD}{\tA'}$ with $\pair{\tD'}{\tD}, \pair{\tB'}{\tA'}
      \in \mathbin{\eqtypeco}$. Then $\pair{\tD}{\tD'}, \pair{\tA'}{\tB'} \in
      \Symm\!\left(\eqtypeco\right)$ and we conclude $\pair{\tA}{\tB} \in
      \Phieqtypeco\!\left(\Symm(\eqtypeco)\right)$.
      
      \item $\tB = \functype{\tB'}{\tB''}$. Similarly, $\tA =
      \functype{\tA'}{\tA''}$ with $ \pair{\tB'}{\tA'}, \pair{\tB''}{\tA''} \in
      \mathbin{\eqtypeco}$. Hence $\pair{\tA'}{\tB'}, \pair{\tA''}{\tB''} \in
      \Symm\!\left(\eqtypeco\right)$ and we conclude $\pair{\tA}{\tB} \in
      \Phieqtypeco\!\left(\Symm(\eqtypeco)\right)$.
    \end{itemize}
    
    \item If not, we have $n + m > 2$ and only the rule $\ruleEqcoUnion$
    applies. Then $$
\begin{array}{r@{\quad\text{s.t.}\quad}c@{\quad\text{for every}\quad}l}
\exists g : 1..m \to 1..n & \pair{\tB_j}{\tA_{g(j)}} \in \mathbin{\eqtypeco} & j \in 1..m \\
\exists f : 1..n \to 1..m & \pair{\tB_{f(i)}}{\tA_i} \in \mathbin{\eqtypeco} & i \in 1..n
\end{array} $$ Applying symmetry we get $\pair{\tA_i}{\tB_{f(i)}},
    \pair{\tA_{g(j)}}{\tB_j} \in \Symm\!\left(\eqtypeco\right)$ for every
    $i \in 1..n, j \in 1..m$. Thus, we conclude $\pair{\tA}{\tB} \in
    \Phieqtypeco\!\left(\Symm(\eqtypeco)\right)$.
  \end{enumerate}

  \item Transitivity: $\Trans\!\left(\S\right) \eqdef \set{\pair{\tA}{\tB}
  \mathrel| \exists \tC \in \Tree. \pair{\tA}{\tC}, \pair{\tC}{\tB} \in \S}$.
  As before, we show that $\Trans\!\left(\eqtypeco\right) \subseteq
  \mathbin{\eqtypeco}$. Let $\pair{\tA}{\tB} \in
  \Trans\!\left(\eqtypeco\right)$, then there exists $\tC \in \Tree$ such that
  $\pair{\tA}{\tC}, \pair{\tC}{\tB} \in \mathbin{\eqtypeco} =
  \Phieqtypeco\!\left(\eqtypeco\right)$. Again, we resort to
  Rem.~\ref{rem:maximalUnionTypes} an consider maximal union types $$
\begin{array}{r@{\quad\text{with}\quad}l}
\tA = \maxuniontype{i \in 1..n}{\tA_i} & \tA_i \neq \iuniontype, i \in 1..n \\
\tB = \maxuniontype{j \in 1..m}{\tB_j} & \tB_j \neq \iuniontype, j \in 1..m \\
\tC = \maxuniontype{k \in 1..l}{\tC_k} & \tC_k \neq \iuniontype, k \in 1..l
\end{array} $$
  \begin{enumerate}
    \item If $n = m = l = 1$ (\ie all three are non-union types), we proceed by
    analyzing the shape of $\tC$:
  \begin{itemize}
	  \item $\tC = a$. By definition of $\Phieqtypeco$, $\tA = a$ and $\tB =
	  a$. Then $\pair{\tA}{\tB} = \pair{a}{a} \in
	  \Phieqtypeco\!\left(\Trans(\eqtypeco)\right)$.
	  
	  \item $\tC = \datatype{\tD''}{\tC'}$. Once again by definition of
	  $\Phieqtypeco$, $\tA = \datatype{\tD}{\tA'}$ with $\pair{\tD}{\tD''},
	  \pair{\tA'}{\tC'} \in \mathbin{\eqtypeco}$ and $\tB =
	  \datatype{\tD'}{\tB'}$ with $\pair{\tD''}{\tD'}, \pair{\tC'}{\tB'} \in
	  \mathbin{\eqtypeco}$. Then $\pair{\tD}{\tD'}, \pair{\tA'}{\tB'} \in
	  \Trans\!\left(\subtypeco\right)$ and we conclude $\pair{\tA}{\tB} \in
	  \Phieqtypeco\!\left(\Trans(\eqtypeco)\right)$.
	  
	  \item $\tC = \functype{\tC'}{\tC''}$. Similarly, we have $\tA =
	  \functype{\tA'}{\tA''}$ and $\tB = \functype{\tB'}{\tB''}$ with
	  $\pair{\tA'}{\tC'}, \pair{\tA''}{\tC''}, \pair{\tC'}{\tB'},
	  \pair{\tC''}{\tB''} \in \mathbin{\eqtypeco}$. By transitivity
	  $\pair{\tA'}{\tB'}, \pair{\tA''}{\tB''} \in
	  \Trans\!\left(\eqtypeco\right)$ and $\pair{\tA}{\tB} \in
	  \Phieqtypeco\!\left(\Trans(\eqtypeco)\right)$.
  \end{itemize}
  
	\item If not (\ie $n + m + l > 3$), we have three different situations to
	consider:
    \begin{inparaenum}[(i)]
      \item $n + l > 2$ and $m + l > 2$;
      \item $n > 1$ and $m = l = 1$; or
      \item $m > 1$ and $n = l = 1$.
    \end{inparaenum}
    In terms of applied rules to derive $\tA \eqtypeco \tC$ and $\tC \eqtypeco
    \tB$, in the former case the only possibility is $\ruleEqcoUnion$ on both
    sides, while in the latter two we have $\ruleEqcoUnion$ on one side and any
    of the other three rules ($\ruleEqcoRefl$, $\ruleEqcoComp$,
    $\ruleEqcoFunc$) on the other. Note that this last two cases are symmetric,
    therefore we only analyse cases (i) and (ii) below:
    \begin{enumerate}[(i)]
      \item $n + l > 2$ and $m + l > 2$. By definition of $\Phieqtypeco$ $$
\begin{array}{r@{\quad\text{s.t.}\quad}c@{\quad\text{for every}\quad}l}
\exists f : 1..n \to 1..l  & \pair{\tA_i}{\tC_{f(i)}}  \in \mathbin{\eqtypeco} & i \in 1..n \\
\exists g : 1..l \to 1..n  & \pair{\tA_{g(k)}}{\tC_k}  \in \mathbin{\eqtypeco} & k \in 1..l \\
\exists f' : 1..l \to 1..m & \pair{\tC_k}{\tB_{f'(k)}} \in \mathbin{\eqtypeco} & k \in 1..l \\
\exists g' : 1..m \to 1..l & \pair{\tC_{g'(j)}}{\tB_j} \in \mathbin{\eqtypeco} & j \in 1..m
\end{array} $$
      Then, we have $\pair{\tA_i}{\tC_{f(i)}},
      \pair{\tC_{f(i)}}{\tB_{f'(f(i))}} \in \mathbin{\eqtypeco}$ for every $i
      \in 1..n$, and $\pair{\tA_{g(g'(j))}}{\tC_{g'(j)}},
      \pair{\tC_{g'(j)}}{\tB_j} \in \mathbin{\eqtypeco}$ for every $j \in
      1..m$.
      
      Here we have two possible situations. If $n = m = 1$ (hence $l > 1$) it
      is necessarily the case $\pair{\tA}{\tC_{f(1)}}, \pair{\tC_{f(1)}}{\tB}
      \in \mathbin{\eqtypeco}$ with all three non-union types. Then we can
      safely conclude, by the previous analysis made in case 1, that
      $\pair{\tA}{\tB} \in \Phieqtypeco\!\left(\Trans(\eqtypeco)\right)$.
      
      If not (\ie $n + m > 2$), taking $f'' = f' \circ f : 1..n \to 1..m$ we
      get $\pair{\tA_i}{\tB_{f''(i)}} \in \Trans\!\left(\eqtypeco\right)$.
      Similarly, $\pair{\tA_{g''(j)}}{\tB_j} \in
      \Trans\!\left(\eqtypeco\right)$ for every $j \in 1..m$ with $g'' = g
      \circ g' : 1..m \to 1..n$. Finally we conclude by $\ruleEqcoUnion$,
      $\pair{\tA}{\tB} \in \Phieqtypeco\!\left(\Trans(\eqtypeco)\right)$.

      \item $n > 1$ and $m = l = 1$. Then, by definition of $\Phieqtypeco$,
      $f : 1..n \to 1$ is a constant function and we have $\pair{\tA_i}{\tC}
      \in \mathbin{\eqtypeco}$ for every $i \in 1..n$. On the other hand
      $\pair{\tC}{\tB} \in \mathbin{\eqtypeco}$ by hypothesis. By transitivity
      once again we get $\pair{\tA_i}{\tB} \in \Trans\!\left(\eqtypeco\right)$
      and we conclude with the same constant function $f$, $\pair{\tA}{\tB} \in
      \Phieqtypeco\!\left(\Trans(\eqtypeco)\right)$.
    \end{enumerate}
  \end{enumerate}
\end{itemize}
\end{proof}

\begin{lemma}[Equality of non-union types is invertible]
\label{lem:equalityIsInvertible}
Let $\tA \eqtypeco \tB$ be two non-union types.
\begin{enumerate}
  \item If $\tA = a$, then $\tB = a$.
  \item If $\tA = \datatype{\tD}{\tA'}$, then $\tB = \datatype{\tD'}{\tB'}$
  with $\tD \eqtypeco \tD'$ and $\tA' \eqtypeco \tB'$.
  \item If $\tA = \functype{\tA'}{\tA''}$, then $\tB = \functype{\tB'}{\tB''}$
  with $\tA' \eqtypeco \tB'$ and $\tA'' \eqtypeco \tB''$.
\end{enumerate}
\end{lemma}

\begin{proof}
Immediate from the definition of subtyping. Note that there's only one
applicable rule in each case.
\end{proof}

Along the document we often resort to the following definition and properties
of the substitution operator over infinite trees:

\begin{definition}
\label{def:treeSubstitution}
The \emphdef{substitution} of a variable $V$ by a tree $\tB$ in $\tA$ (notation
$\substitute{V}{\tB}{\tA}$) is defined as: $$
\begin{array}{r@{\quad \eqdef\quad}l@{\quad}l}
(\substitute{V}{\tB}{\tA})(\pi)     & \tA(\pi)  & \text{if $\tA(\pi)$ defined and $\tA(\pi) \neq V$} \\
(\substitute{V}{\tB}{\tA})(\pi\pi') & \tB(\pi') & \text{if $\tA(\pi)$ defined and $\tA(\pi) = V$}
\end{array} $$
\end{definition}

The following lemma provides a more convenient characterisation of the
substitution.

\begin{lemma}
\label{lem:treeSubstitution}
\begin{enumerate}[(i)]
  \item $\substitute{V}{\tB}{V} = \tB$.
  \item $\substitute{V}{\tB}{a} = a$ for $V \neq a \in \TypeVariable \cup
  \TypeConstant$.
  \item $\substitute{V}{\tB}{(\tA_1 \star \tA_2)} = \substitute{V}{\tB}{\tA_1}
  \star \substitute{V}{\tB}{\tA_2}$ for $\star \in \set{\idatatype, \ifunctype,
  \iuniontype}$.
\end{enumerate}
\end{lemma}

\begin{proof}
The three cases are by analysis of the defined positions.
\begin{enumerate}[(i)]
  \item The only defined position in $V$ is $\epsilon$. Then, for every $\pi$
  in $\tB$ we have $$(\substitute{V}{\tB}{V})(\pi) =
  (\substitute{V}{\tB}{V})(\epsilon\pi) = B(\pi)$$
  
  \item The only defined position in $a \neq V$ is $\epsilon$, thus we have
  $(\substitute{V}{\tB}{a})(\epsilon) = a(\epsilon) = a$. Any other position is
  undefined.
  
  \item Here we have $\tA = \tA_1 \star \tA_2$ with $\star \in \set{\idatatype,
  \ifunctype, \iuniontype}$. We proceed by analysing the defined positions of
  $\tA$.
  \begin{itemize}
    \item $\pi = \epsilon$. Then
    $$(\substitute{V}{\tB}{(\tA_1 \star \tA_2)})(\epsilon) =
    (\tA_1 \star \tA_2)(\epsilon) = \star = (\substitute{V}{\tB}{\tA_1} \star
    \substitute{V}{\tB}{\tA_2})(\epsilon)$$
    \item $\pi = i\pi'$. Here we have two possibilities:
    \begin{enumerate}
      \item either $\tA(\pi) \neq V$. Then $\tA_i(\pi') \neq V$ and we have $$
\begin{array}{r@{\quad=\quad}l@{\quad}l}
(\substitute{V}{\tB}{(\tA_1 \star \tA_2)})(\pi) & (\tA_1 \star \tA_2)(i\pi')         & \text{by Def.~\ref{def:treeSubstitution}} \\
                                                & \tA_i(\pi') \\
                                                & (\substitute{V}{\tB}{\tA_i})(\pi') & \text{by Def.~\ref{def:treeSubstitution}} \\
                                                & (\substitute{V}{\tB}{\tA_1} \star \substitute{V}{\tB}{\tA_2})(\pi)
\end{array} $$
      \item or $\tA(\pi) = V$. Then $\tA_i(\pi') = V$ and by definition of
      substitution we have, for every position $\pi''$ in $\tB$ $$
\begin{array}{r@{\quad=\quad}l}
(\substitute{V}{\tB}{(\tA_1 \star \tA_2)})(\pi\pi'') & \tB(\pi'') \\
                                                     & (\substitute{V}{\tB}{\tA_i})(\pi'\pi'') \\
                                                     & (\substitute{V}{\tB}{\tA_1} \star \substitute{V}{\tB}{\tA_2})(\pi\pi'')
\end{array} $$
    \end{enumerate}
  \end{itemize}
\end{enumerate}
\end{proof}

We show next that the substitution preserves the equivalent relation.

\begin{lemma}
\label{lem:substitutionOfEqtypesCo}
Let $\tA \eqtypeco \tA'$ and $\tB \eqtypeco \tB'$. Then
$\substitute{V}{\tB}{\tA} \eqtypeco \substitute{V}{\tB'}{\tA'}$.
\end{lemma}

\begin{proof}
Let $\S = \set{\pair{\substitute{V}{\tB}{\tA}}{\substitute{V}{\tB'}{\tA'}}
\mathrel| \tA \eqtypeco \tA', \tB \eqtypeco \tB'}$. We show that $\S \cup
\mathbin{\eqtypeco}$ is $\Phieqtypeco$-dense.

Let $\pair{\tC}{\tC'} \in \S \cup \mathbin{\eqtypeco}$. If $\pair{\tC}{\tC'}
\in \mathbin{\eqtypeco}$ the result is immediate by monotonicity of
$\Phieqtypeco$, since $\mathbin{\eqtypeco} =
\Phieqtypeco\!\left(\eqtypeco\right) \subseteq \Phieqtypeco\!\left(\S \cup
\mathbin{\eqtypeco}\right)$. Then we only present the case where
$\pair{\tC}{\tC'} \in \S$, $\tC = \substitute{V}{\tB}{\tA}$ and $\tC' =
\substitute{V}{\tB'}{\tA'}$ with $\tA \eqtypeco \tA'$ and $\tB \eqtypeco \tB'$.
Assume, without loss of generality $$
\begin{array}{r@{\quad\text{with}\quad}l}
\tA  = \maxuniontype{i \in 1..n}{\tA_i}  & \tA_i  \neq \iuniontype, i \in 1..n \\
\tA' = \maxuniontype{j \in 1..m}{\tA'_j} & \tA'_j \neq \iuniontype, j \in 1..m
\end{array} $$

\begin{enumerate}
  \item If $n = m = 1$ (\ie $\tA, \tA' \neq \iuniontype$), we analyze the shape
  of $\tA$:
  \begin{itemize}
    \item $\tA = a$. By Lem.~\ref{lem:equalityIsInvertible}, $\tA' = a$ and we
    have two possible cases. If $a \neq V$, by Lem.~\ref{lem:treeSubstitution}
    (ii), $\tC = a = \tC'$. If not, by Lem.~\ref{lem:treeSubstitution} (i),
    $\tC = \tB \eqtypeco \tB' = \tC'$. Both cases are immediate by definition
    of $\mathbin{\eqtypeco} \subseteq \Phieqtypeco\!\left(\S \cup
    \mathbin{\eqtypeco}\right)$.
    
    \item $\tA = \datatype{\tD}{\tA_1}$. By
    Lem.~\ref{lem:equalityIsInvertible}, $\tA' = \datatype{\tD'}{\tA'_1}$ with
    $\tD \eqtypeco \tD'$ and $\tA_1 \eqtypeco \tA'_1$. Then, by definition of
    $\S$, we have $\pair{\substitute{V}{\tB}{\tD}}{\substitute{V}{\tB'}{\tD'}}$
    and $\pair{\substitute{V}{\tB}{\tA_1}}{\substitute{V}{\tB'}{\tA'_1}} \in \S
    \cup \mathbin{\eqtypeco}$. Finally we conclude $\pair{\tC}{\tC'} \in
    \Phieqtypeco\!\left(\S \cup \mathbin{\eqtypeco}\right)$ since, by
    Lem.~\ref{lem:treeSubstitution} (iii), $$
\begin{array}{r@{\quad=\quad}c@{\quad=\quad}l}
\tC  & \substitute{V}{\tB}{(\datatype{\tD}{\tA_1})}    & \datatype{\substitute{V}{\tB}{\tD}}{\substitute{V}{\tB}{\tA_1}} \\
\tC' & \substitute{V}{\tB'}{(\datatype{\tD'}{\tA'_1})} & \datatype{\substitute{V}{\tB'}{\tD'}}{\substitute{V}{\tB'}{\tA'_1}}
\end{array} $$

    \item $\tA = \functype{\tA_1}{\tA_2}$. As before, by
    Lem.~\ref{lem:equalityIsInvertible}, we get $\tA =
    \functype{\tA'_1}{\tA'_2}$ with $\tA_1 \eqtypeco \tA'_1$ and $\tA_2
    \eqtypeco \tA'_2$. By definition $\S$ we have
    $\pair{\substitute{V}{\tB}{\tA_1}}{\substitute{V}{\tB'}{\tA'_1}}$ and
    $\pair{\substitute{V}{\tB}{\tA_2}}{\substitute{V}{\tB'}{\tA'_2}} \in \S
    \cup \mathbin{\eqtypeco}$. Thus, we conclude by
    Lem.~\ref{lem:treeSubstitution} (iii), $\pair{\tC}{\tC'} \in
    \Phieqtypeco\!\left(\S \cup \mathbin{\eqtypeco}\right)$.
  \end{itemize}
  
  \item If $n + m > 2$, by $\ruleEqcoUnion$ we have $$
\begin{array}{r@{\quad\text{s.t.}\quad}c@{\quad\text{for every}\quad}l}
\exists f : 1..n \to 1..m & \tA_i \eqtypeco \tA'_{f(i)} & i \in 1..n \\
\exists g : 1..m \to 1..n & \tA_{g(j)} \eqtypeco \tA'_j & j \in 1..m
\end{array} $$ Then,
  $\pair{\substitute{V}{\tB}{\tA_i}}{\substitute{V}{\tB'}{\tA'_{f(i)}}}$ and
  $\pair{\substitute{V}{\tB}{\tA_{g(j)}}}{\substitute{V}{\tB'}{\tA'_j}} \in \S
  \cup \mathbin{\eqtypeco}$ for every $i \in 1..n, j \in 1..m$. Once again we
  conclude by definition of $\Phieqtypeco$ and Lem.~\ref{lem:treeSubstitution}
  (iii), $\pair{\tC}{\tC'} \in
  \Phieqtypeco\!\left(\S \cup \mathbin{\eqtypeco}\right)$.
\end{enumerate}
\end{proof}

\subsubsection{Subtyping of trees}

In a similar way we have a coinductive characterization of subtyping over trees.

\begin{definition}
Infinite type subtyping, written $\subtypeco$, is defined by the coinductive
interpretation of the schemes in Fig.~\ref{fig:subtypingSchemesCo}. 
\end{definition}

The most interesting rule in Fig.~\ref{fig:subtypingSchemesCo} is
$\ruleSubcoUnion$. Here, for a maximal union type of the form $\maxuniontype{i
\in 1..n}{\tA_i}$ to be a subtype of a maximal union type $\maxuniontype{j \in
1..m}{\tB_j}$, one of the two must have at least one occurrence of the union
type construct ($n + m > 2$) and there must be a function $f : 1..n \to 1..m$
such that $\tA_i \subtypeco \tB_{f(i)}$ for each $i \in 1..n$. 

\begin{figure}[t] $$
\begin{array}{c}
\RuleCo{}
       {a \subtypeco a}
       {\ruleSubcoRefl}
\\
\\
\RuleCo{\tA' \subtypeco \tA \quad \tB \subtypeco \tB'}
       {\functype{\tA}{\tB} \subtypeco \functype{\tA'}{\tB'}}
       {\ruleSubcoFunc}
\qquad
\RuleCo{\tD \subtypeco \tD' \quad \tA \subtypeco \tA'}
       {\datatype{\tD}{\tA} \subtypeco \datatype{\tD'}{\tA'}}
       {\ruleSubcoComp}
\\
\\
\RuleCo{\tA_i \subtypeco \tB_{f(i)}
        \quad
        f : 1..n \to 1..m
        \quad
        \tA_i, \tB_j\neq \iuniontype
        \quad
        n + m > 2}
       {\maxuniontype{i \in 1..n}{\tA_i} \subtypeco \maxuniontype{j \in 1..m}{\tB_j}}
       {\ruleSubcoUnion}
\end{array} $$
\caption{Subtyping relation for infinite types}
\label{fig:subtypingSchemesCo}
\end{figure}

\begin{remark}
The rules are derived from those of Fig.~\ref{fig:subtypingSchemesMu}. More
precisely, rules $\ruleSubmuUnionRL$, $\ruleSubmuUnionRR$ and
$\ruleSubmuUnionL$ of Fig.~\ref{fig:subtypingSchemesMu} and the observation
that $\ruleSubmuUnionRL$ and $\ruleSubmuUnionRR$ can always be permuted past
$\ruleSubmuUnionL$.
\end{remark}

As above, the formal definition of the subtyping relation is given by the
associated function $\Phisubtypeco : \Parts{\Tree \times \Tree} \to
\Parts{\Tree \times \Tree}$ defined next: $$
\begin{array}{rcl}
\Phisubtypeco(\S) & =    & \set{\pair{a}{a} \mathrel| a \in \TypeVariable \cup \TypeConstant} \\
                  & \cup & \set{\pair{\datatype{\tD}{\tA}}{\datatype{\tD'}{\tA'}} \mathrel| \pair{\tD}{\tD'}, \pair{\tA}{\tA'} \in \S} \\
                  & \cup & \set{\pair{\functype{\tA}{\tB}}{\functype{\tA'}{\tB'}} \mathrel| \pair{\tA'}{\tA}, \pair{\tB}{\tB'} \in \S} \\
                  & \cup & \{\pair{\maxuniontype{i \in 1..n}{\tA_i}}{\maxuniontype{j \in 1..m}{\tB_j}} \mathrel| \tA_i, \tB_j \neq \iuniontype, n + m > 2 \\
                  &      & \qquad \exists f : 1..n \to 1..m \text{ s.t. } \pair{\tA_i}{\tB_{f(i)}} \in \S\}
\end{array} $$

Then $\mathbin{\subtypeco} = \nu\Phisubtypeco$. We now address some properties
of subtyping.

\begin{lemma}[Subtyping is a preorder]
$\subtypeco$ is a preorder (\ie reflexive and transitive).
\end{lemma}

\begin{proof}
This proof is similar to the one presented before for $\eqtypeco$.
\end{proof}

The following notion of invertibility (Lem.~\ref{lem:subtypingIsInvertible}) is
the main result of the present Section and an essential property to prove
Subject Reduction (Prop.~\ref{prop:subjectReduction}) and Progress
(Prop.~\ref{prop:progress}) for the type system proposed in
Sec.~\ref{sec:typingSystem}.

\begin{lemma}[Subtyping of non-union types is invertible]
\label{lem:subtypingIsInvertible}
Let $\tA, \tB \in \Tree$ be non-union types. Suppose $\tA \subtypeco \tB$.
\begin{enumerate}
  \item If $\tA = a$, then $\tB = a$.
  \item If $\tA = \datatype{\tD}{\tA'}$, then $\tB = \datatype{\tD'}{\tB'}$
  with $\tD \subtypeco \tD'$ and $\tA' \subtypeco \tB'$.
  \item If $\tA = \functype{\tA'}{\tA''}$, then $\tB = \functype{\tB'}{\tB''}$
  with $\tB' \subtypeco \tA'$ and $\tA'' \subtypeco \tB''$.
\end{enumerate}
\end{lemma}

\begin{remark}
In each of the three items of Lem.~\ref{lem:subtypingIsInvertible} the roles of
$\tA$ and $\tB$ can be reversed. 
\end{remark}


\begin{lemma}
\label{lem:eqImpliesSub}
$\tA \eqtypeco \tB \implies \tA \subtypeco \tB$.
\end{lemma}

\begin{proof}
We show that $\mathbin{\eqtypeco} = \Phieqtypeco\!\left(\eqtypeco\right)$ is
$\Phisubtypeco$-dense. Let $\pair{\tA}{\tB} \in \mathbin{\eqtypeco}$. By
Rem.~\ref{rem:maximalUnionTypes} we can consider maximal union types $$
\begin{array}{r@{\quad\text{with}\quad}l}
\tA = \maxuniontype{i \in 1..n}{\tA_i} & \tA_i \neq \iuniontype, i \in 1..n \\
\tB = \maxuniontype{j \in 1..m}{\tB_j} & \tB_j \neq \iuniontype, j \in 1..m
\end{array} $$ and we have two separate cases to analyze:
\begin{enumerate}
  \item If $n = m = 1$, then both $\tA$ and $\tB$ are non-union types. Now we
  proceed by analyzing the shape of $\tA$:
  \begin{itemize}
    \item $\tA = a$. Then, by definition of $\Phieqtypeco$, $\tB = a$ and the
    result is immediate since $\pair{a}{a} \in
    \Phisubtypeco\!\left(\eqtypeco\right)$ for every $a \in \TypeVariable \cup
    \TypeConstant$.
    
    \item $\tA = \datatype{\tD}{\tA'}$. Again, by definition of $\Phieqtypeco$,
    we have $\tB = \datatype{\tD'}{\tB'}$ with $\pair{\tD}{\tD'},
    \pair{\tA'}{\tB'} \in \mathbin{\eqtypeco}$. Then we conclude by definition
    of $\Phisubtypeco$, $\pair{\datatype{\tD}{\tA'}}{\datatype{\tD'}{\tB'}} \in
    \Phisubtypeco\!\left(\eqtypeco\right)$.
    
    \item $\tA = \functype{\tA'}{\tA''}$. Similarly, $\tB =
    \functype{\tB'}{\tB''}$ with $ \pair{\tA'}{\tB'}, \pair{\tA''}{\tB''} \in
    \mathbin{\eqtypeco}$. By symmetry $\pair{\tB'}{\tA'} \in
    \mathbin{\eqtypeco}$ and we conclude $\pair{\tA}{\tB} \in
    \Phisubtypeco\!\left(\eqtypeco\right)$.
  \end{itemize}
  
  \item If not (\ie $n + m > 2$), rule $\ruleEqcoUnion$ applies. Then $$
\begin{array}{r@{\quad\text{s.t.}\quad}c@{\quad\text{for every}\quad}l}
\exists f : 1..n \to 1..m & \pair{\tA_i}{\tB_{f(i)}} \in \mathbin{\eqtypeco} & i \in 1..n \\
\exists g : 1..m \to 1..n & \pair{\tA_{g(j)}}{\tB_j} \in \mathbin{\eqtypeco} & j \in 1..m
\end{array} $$ Thus, we conclude with the same function $f$, $\pair{A}{B}
  \in \Phisubtypeco\!\left(\eqtypeco\right)$.
\end{enumerate}
\end{proof}

To prove the correspondence of the coinductive formulation with the inductive
approach, it is convenient to work with finite trees (types). Thus, we
introduce a characterisation of the equivalence and subtyping relations in
terms of finite truncations of infinite trees.

We denote with $\card{\iuniontype}{\tA}$ the maximal number of adjacent union
type nodes, starting from the root of $\tA$: $$
\card{\iuniontype}{A} \eqdef \left\{
\begin{array}{l@{\quad\text{if}\quad}l}
0                                                         & \tA \neq \iuniontype \\
1 + \card{\iuniontype}{\tA_1} + \card{\iuniontype}{\tA_2} & \tA = \tA_1 \iuniontype \tA_2
\end{array}\right.
$$ Recall that, by definition of $\Tree$, a type cannot consist of infinitely
many consecutive occurrences of $\iuniontype$. Thus, the previous inductive
definition is well-founded, as well as the following:

\begin{definition}
\label{def:treeCut}
The \emphdef{truncation} of a tree $\tA$ at depth $k \in \Natural$ (notation
$\cut{\tA}{k}$) is defined inductively\footnote{Using the lexicographical
extension of the standard order to $\pair{k}{\card{\iuniontype}{\tA}}$.} as
follows: $$
\begin{array}{r@{\quad\eqdef\quad}l@{\quad}l}
\cut{\tA}{0}                          & \bullet \\
\cut{a}{k+1}                          & a                                   & \text{for $a \in \TypeVariable \cup \TypeConstant$} \\
\cut{(\tA_1 \star \tA_2)}{k+1}        & \cut{\tA_1}{k} \star \cut{\tA_2}{k} & \text{for $\star \in \set{\idatatype, \ifunctype}$} \\
\cut{(\uniontype{\tA_1}{\tA_2})}{k+1} & \uniontype{\cut{\tA_1}{k+1}}{\cut{\tA_2}{k+1}}
\end{array} $$ where $\bullet \in \TypeConstant$ is a distinguished type
constant used to identify the nodes where the tree was truncated.
\end{definition}

\begin{remark}
\label{rem:cutMaximalUnionTypes}
Given a maximal union type $\maxuniontype{i \in 1..n}{\tA_i}$, immediately from
the definition we have $\cut{(\maxuniontype{i \in 1..n}{\tA_i})}{k+1} =
\maxuniontype{i \in 1..n}{(\cut{\tA_i}{k+1})}$.
\end{remark}

\begin{lemma}
\label{lem:cutEquivalenceCo}
$\forall k \in \Natural.\cut{\tA}{k} \eqtypeco \cut{\tB}{k}$ iff $\tA \eqtypeco \tB$.
\end{lemma}

\begin{proof}
$\Rightarrow)$ We show that $\S \eqdef \set{\pair{\tA}{\tB} \mathrel| \forall
k \in \Natural. \cut{\tA}{k} \eqtypeco \cut{\tB}{k}}$ is $\Phieqtypeco$-dense.
Let $\pair{\tA}{\tB} \in \S$. Then, for every $k \in \Natural$ we have
$\cut{\tA}{k} \eqtypeco \cut{\tB}{k}$. Consider maximal union types $$
\begin{array}{r@{\quad\text{with}\quad}l}
\tA = \maxuniontype{i \in 1..n}{\tA_i} & \tA_i \neq \iuniontype, i \in 1..n \\
\tB = \maxuniontype{j \in 1..m}{\tB_j} & \tB_j \neq \iuniontype, j \in 1..m
\end{array} $$
\begin{enumerate}
  \item If $n = m = 1$ (\ie $\tA, \tB \neq \iuniontype$), we proceed by
  analyzing the shape of $\tA$:
  \begin{itemize}
    \item $\tA = a$. Then, $\cut{\tA}{k} = a$ for every $k > 0$ and, by
    Lem.~\ref{lem:equalityIsInvertible}, $\cut{\tB}{k} = a$. Hence, $\tB = a$
    and we conclude directly from the definition of $\Phieqtypeco$,
    $\pair{a}{a} \in \Phieqtypeco\!\left(\S\right)$.
    
    \item $\tA = \datatype{\tD}{\tA'}$. Similarly, we have $\cut{\tA}{k} =
    \datatype{\cut{\tD}{k-1}}{\cut{\tA'}{k-1}}$ for every $k > 0$. By
    Lem.~\ref{lem:equalityIsInvertible} once again, we get $\cut{\tB}{k} =
    \datatype{\tD'_k}{\tB'_k}$ with $\cut{\tD}{k-1} \eqtypeco \tD'_k$ and
    $\cut{\tA'}{k-1} \eqtypeco \tB'_k$. Note that for every $k$ we have
    different subtrees $\tD'_k$ and $\tB'_k$ but, since
    Lem.~\ref{lem:equalityIsInvertible} refers to tree equality (not
    equivalence) when determining the shape of $\tB$, it is immediate to see
    from the definition of the truncation that $\tB = \datatype{\tD'}{\tB'}$
    with $\tD'_k = \cut{\tD'}{k-1}$ and $\tB'_k = \cut{\tB'}{k-1}$ for every
    $k > 0$. Hence, $\cut{\tD}{k-1} \eqtypeco \cut{\tD'}{k-1}$ and
    $\cut{\tA'}{k-1} \eqtypeco \cut{\tB'}{k-1}$ for every $k > 0$. Then, by
    definition of $\S$, $\pair{\tD}{\tD'}, \pair{\tA'}{\tB'} \in \S$ and we
    conclude $\pair{\datatype{\tD}{\tA'}}{\datatype{\tD'}{\tB'}} \in
    \Phieqtypeco\!\left(\S\right)$.
    
    \item $\tA = \functype{\tA'}{\tA''}$. Analysis for this case is similar to
    the previous one. From $\cut{\tA}{k} =
    \functype{\cut{\tA'}{k-1}}{\cut{\tA''}{k-1}}$ we get $\tB =
    \functype{\tB'}{\tB''}$ with $\cut{\tA'}{k-1} \eqtypeco \cut{\tB'}{k-1}$
    and $\cut{\tA''}{k-1} \eqtypeco \cut{\tB''}{k-1}$ for every $k > 0$. Then
    we have $\pair{\tA'}{\tB'}, \pair{\tA''}{\tB''} \in \S$ and conclude
    $\pair{\functype{\tA'}{\tA''}}{\functype{\tB'}{\tB''}} \in
    \Phieqtypeco\!\left(\S\right)$.
  \end{itemize}
  
  \item If $n + m > 2$ we have $\cut{\tA}{k} =
  \maxuniontype{i \in 1..n}{(\cut{\tA_i}{k})}$ and $\cut{\tB}{k} =
  \maxuniontype{j \in 1..m}{(\cut{\tB_j}{k})}$ for every $k > 0$. From
  $\cut{\tA}{k} \eqtypeco \cut{\tB}{k}$, by $\ruleEqcoUnion$, we get $$
\begin{array}{r@{\quad\text{s.t.}\quad}c@{\quad\text{for every}\quad}l}
\exists f : 1..n \to 1..m & \cut{\tA_i}{k} \eqtypeco \cut{\tB_{f(i)}}{k} & i \in 1..n \\
\exists g : 1..m \to 1..n & \cut{\tA_{g(j)}}{k} \eqtypeco \cut{\tB_j}{k} & j \in 1..m
\end{array} $$
  Since $\cut{\tC}{0} = \bullet$ for every $\tC \in \Tree$, we have
  $\cut{\tA_i}{0} \eqtypeco \cut{\tB_{f(i)}}{0}$ and $\cut{\tA_{g(j)}}{0}
  \eqtypeco \cut{\tB_j}{0}$ by reflexivity. Thus, $\cut{\tA_i}{k} \eqtypeco
  \cut{\tB_{f(i)}}{k}$ and $\cut{\tA_{g(j)}}{k} \eqtypeco \cut{\tB_j}{k}$ for
  every $k \in \Natural$. Then, by definition of $\S$,
  $\pair{\tA_i}{\tB_{f(i)}}, \pair{\tA_{g(j)}}{\tB_j} \in \S$ for every $i \in
  1..n$, $j \in 1..m$. Finally, we conclude $\pair{\tA}{\tB} \in
  \Phieqtypeco\!\left(\S\right)$.
\end{enumerate}

$\Leftarrow)$ For this part of the proof we show that the converse relation $\S
\eqdef \set{\pair{\cut{\tA}{k}}{\cut{\tB}{k}} \mathrel| \tA \eqtypeco \tB, k
\in \Natural}$ is $\Phieqtypeco$-dense. Let $\pair{\cut{\tA}{k}}{\cut{\tB}{k}}
\in \S$. If $k = 0$, by definition of the truncation, $\cut{\tA}{k} = \bullet =
\cut{\tB}{k}$ and trivially $\pair{\bullet}{\bullet} \in
\Phieqtypeco\left(\S\right)$. We analyze next the cases where $k > 0$ given
that, by definition of $\S$, $\tA \eqtypeco \tB$. Once again we consider
maximal union types $$
\begin{array}{r@{\quad\text{with}\quad}l}
\tA = \maxuniontype{i \in 1..n}{\tA_i} & \tA_i \neq \iuniontype, i \in 1..n \\
\tB = \maxuniontype{j \in 1..m}{\tB_j} & \tB_j \neq \iuniontype, j \in 1..m
\end{array} $$ and analyze separately the cases where both $\tA$ and $\tB$ ar
non-union types.
\begin{enumerate}
  \item If $n = m = 1$ we a look at the shape of $A$:
  \begin{itemize}
    \item $\tA = a$. By Lem.~\ref{lem:equalityIsInvertible}, $\tB = a$ and
    $\cut{a}{k} = a$ for every $k > 0$. Then we conclude by definition of
    $\Phieqtypeco$, $\pair{a}{a} \in \Phieqtypeco\left(\S\right)$.
    
    \item $\tA = \datatype{\tD}{\tA'}$. By Lem.~\ref{lem:equalityIsInvertible},
    $\tB = \datatype{\tD'}{\tB'}$ with $\tD \eqtypeco \tD'$ and $\tA' \eqtypeco
    \tB'$. Then, by definition of $\S$,
    $\pair{\cut{\tD}{k-1}}{\cut{\tD'}{k-1}},
    \pair{\cut{\tA'}{k-1}}{\cut{\tB'}{k-1}} \in \S$ and we conclude
    $\pair{\cut{\tA}{k}}{\cut{\tB}{k}} =
    \pair{\datatype{\cut{\tD}{k-1}}{\cut{\tA'}{k-1}}}{\datatype{\cut{\tD'}{k-1}}{\cut{\tB'}{k-1}}}
    \in \Phieqtypeco\left(\S\right)$.
    
    \item $\tA = \functype{\tA'}{\tA''}$. Similarly to the previous case, we
    have $\tB = \functype{\tB'}{\tB''}$ with $\tA' \eqtypeco \tB'$ and $\tA''
    \eqtypeco \tB''$. Then $\pair{\cut{\tA'}{k-1}}{\cut{\tB'}{k-1}},
    \pair{\cut{\tA''}{k-1}}{\cut{\tB''}{k-1}} \in \S$ and
    $\pair{\cut{\tA}{k}}{\cut{\tB}{k}} =
    \pair{\functype{\cut{\tA'}{k-1}}{\cut{\tA''}{k-1}}}{\functype{\cut{\tB'}{k-1}}{\cut{\tB''}{k-1}}}
    \in \Phieqtypeco\left(\S\right)$.
  \end{itemize}
  
  \item If $n + m > 2$, by $\ruleEqcoUnion$ we have $$
\begin{array}{r@{\quad\text{s.t.}\quad}c@{\quad\text{for every}\quad}l}
\exists f : 1..n \to 1..m & \tA_i \eqtypeco \tB_{f(i)} & i \in 1..n \\
\exists g : 1..m \to 1..n & \tA_{g(j)} \eqtypeco \tB_j & j \in 1..m
\end{array} $$
  Then, by definition of $\S$, $\pair{\cut{\tA_i}{k}}{\cut{\tB_{f(i)}}{k}},
  \pair{\cut{\tA_{g(j)}}{k}}{\cut{\tB_j}{k}} \in \S$ for every $k > 0$. Thus,
  we conclude by resorting to Rem.~\ref{rem:cutMaximalUnionTypes},
  $\pair{\cut{\tA}{k}}{\cut{\tB}{k}} \in \Phieqtypeco\left(\S\right)$.
\end{enumerate}
\end{proof}

\begin{lemma}
\label{lem:cutSubtypingCo}
$\forall k \in \Natural.\cut{\tA}{k} \subtypeco \cut{\tB}{k}$ iff $\tA \subtypeco \tB$.
\end{lemma}

\begin{proof}
$\Rightarrow)$ Similarly to the previous lemma, we prove this part by showing
that $\S \eqdef \set{\pair{\tA}{\tB} \mathrel| \forall k \in \Natural.
\cut{\tA}{k} \subtypeco \cut{\tB}{k}}$ is $\Phisubtypeco$-dense. By hypothesis
we have $\cut{\tA}{k} \subtypeco \cut{\tB}{k}$ for every $k \in \Natural$. As
before we consider maximal union types and analyze separately the case for
non-union types $$
\begin{array}{r@{\quad\text{with}\quad}l}
\tA = \maxuniontype{i \in 1..n}{\tA_i} & \tA_i \neq \iuniontype, i \in 1..n \\
\tB = \maxuniontype{j \in 1..m}{\tB_j} & \tB_j \neq \iuniontype, j \in 1..m
\end{array} $$
\begin{enumerate}
  \item If $n = m = 1$ (\ie $\tA, \tB \neq \iuniontype$), we proceed by
  analyzing the shape of $\tA$:
  \begin{itemize}
    \item $\tA = a$. Then, $\cut{\tA}{k} = a$ for every $k > 0$ and, by
    Lem.~\ref{lem:subtypingIsInvertible}, $\cut{\tB}{k} = a$. Hence, $\tB = a$
    and we conclude directly from the definition of $\Phisubtypeco$,
    $\pair{a}{a} \in \Phisubtypeco\!\left(\S\right)$.
    
    \item $\tA = \datatype{\tD}{\tA'}$. Similarly, we have $\cut{\tA}{k} =
    \datatype{\cut{\tD}{k-1}}{\cut{\tA'}{k-1}}$ for every $k > 0$. By
    Lem.~\ref{lem:subtypingIsInvertible} once again, we get $\cut{\tB}{k} =
    \datatype{\tD'_k}{\tB'_k}$ with $\cut{\tD}{k-1} \subtypeco \tD'_k$ and
    $\cut{\tA'}{k-1} \subtypeco \tB'_k$. As in the previous lemma, in this case
    we have different subtrees $\tD'_k$ and $\tB'_k$ for every $k$ but, by
    resorting to tree equality on Lem.~\ref{lem:subtypingIsInvertible} and the
    definition of the truncation, we can assure that $\tB =
    \datatype{\tD'}{\tB'}$ with $\tD'_k = \cut{\tD'}{k-1}$ and $\tB'_k =
    \cut{\tB'}{k-1}$ for every $k > 0$. Hence, $\cut{\tD}{k-1} \subtypeco
    \cut{\tD'}{k-1}$ and $\cut{\tA'}{k-1} \subtypeco \cut{\tB'}{k-1}$ for every
    $k > 0$. Then, by definition of $\S$, $\pair{\tD}{\tD'}, \pair{\tA'}{\tB'}
    \in \S$ and we conclude $\pair{\datatype{\tD}{\tA'}}{\datatype{\tD'}{\tB'}}
    \in \Phisubtypeco\!\left(\S\right)$.
    
    \item $\tA = \functype{\tA'}{\tA''}$. Analysis for this case is similar to
    the previous one. From $\cut{\tA}{k} =
    \functype{\cut{\tA'}{k-1}}{\cut{\tA''}{k-1}}$ we get $\tB =
    \functype{\tB'}{\tB''}$ with $\cut{\tB'}{k-1} \subtypeco \cut{\tA'}{k-1}$
    and $\cut{\tA''}{k-1} \subtypeco \cut{\tB''}{k-1}$ for every $k > 0$. Note
    that, by Lem.~\ref{lem:subtypingIsInvertible}, subtyping order on the
    domains is inverted. Then we have $\pair{\tB'}{\tA'}, \pair{\tA''}{\tB''}
    \in \S$ and conclude $\pair{\functype{\tA'}{\tA''}}{\functype{\tB'}{\tB''}}
    \in \Phisubtypeco\!\left(\S\right)$.
  \end{itemize}
  
  \item If $n + m > 2$ we have $\cut{\tA}{k} = \maxuniontype{i \in
  1..n}{(\cut{\tA_i}{k})}$ and $\cut{\tB}{k} = \maxuniontype{j \in
  1..m}{(\cut{\tB_j}{k})}$ for every $k > 0$. From $\cut{\tA}{k} \subtypeco
  \cut{\tB}{k}$, by $\ruleSubcoUnion$, we get $$
\begin{array}{r@{\quad\text{s.t.}\quad}c@{\quad\text{for every}\quad}l}
\exists f : 1..n \to 1..m & \cut{\tA_i}{k} \subtypeco \cut{\tB_{f(i)}}{k} & i \in 1..n \\
\exists g : 1..m \to 1..n & \cut{\tA_{g(j)}}{k} \subtypeco \cut{\tB_j}{k} & j \in 1..m
\end{array} $$
  Since $\cut{\tC}{0} = \bullet$ for every $\tC \in \Tree$, we also have
  $\cut{\tA_i}{0} \subtypeco \cut{\tB_{f(i)}}{0}$ and $\cut{\tA_{g(j)}}{0}
  \subtypeco \cut{\tB_j}{0}$ by reflexivity. Thus, $\cut{\tA_i}{k} \subtypeco
  \cut{\tB_{f(i)}}{k}$ and $\cut{\tA_{g(j)}}{k} \subtypeco \cut{\tB_j}{k}$ for
  every $k \in \Natural$. Then, by definition of $\S$,
  $\pair{\tA_i}{\tB_{f(i)}}, \pair{\tA_{g(j)}}{\tB_j} \in \S$ for every $i \in
  1..n$, $j \in 1..m$. Finally, we conclude $\pair{\tA}{\tB} \in
  \Phisubtypeco\!\left(\S\right)$.
\end{enumerate}

$\Leftarrow)$ As before, we define $\S \eqdef
\set{\pair{\cut{\tA}{k}}{\cut{\tB}{k}} \mathrel| \tA \subtypeco \tB, k \in
\Natural}$ and show that is $\Phisubtypeco$-dense to prove this part of the
lemma. Again, if $k = 0$ the result is immediate, so lets focus on the case
where $k > 0$.

Let $\tA \subtypeco \tB$. We assume, without loss of generality, $\tA =
\maxuniontype{i \in 1..n}{\tA_i}$ and $\tB = \maxuniontype{j \in 1..m}{\tB_j}$
are maximal union types.

If $n + m > 2$ it is the case of $\ruleSubcoUnion$ and we have $\exists f :
1..n \to 1..m$ such that $\tA_i \subtypeco \tB_{f(i)}$ for every $i \in 1..n$.
Then, by definition we have $\pair{\cut{\tA_i}{k}}{\cut{\tB_{f(i)}}{k}} \in \S$
and conclude $\pair{\cut{\tA}{k}}{\cut{\tB}{k}} \in
\Phisubtypeco\left(\S\right)$.

On the other hand, if $n = 1 = m$ we analyze the form of $\tA$:
\begin{enumerate}
  \item $\tA = a$. By Lem.~\ref{lem:subtypingIsInvertible} we have $\tB = a$
  and the result is immediate.
  
  \item $\tA = \datatype{\tD}{\tA'}$. By Lem.~\ref{lem:subtypingIsInvertible},
  $\tB = \datatype{\tD'}{\tB'}$ with $\tD \subtypeco \tD'$ and $\tA' \subtypeco
  \tB'$. Then we have $\pair{\cut{\tD}{k-1}}{\cut{\tD'}{k-1}},
  \pair{\cut{\tA'}{k-1}}{\cut{\tB'}{k-1}} \in \S$ for every $k > 0$, and
  conclude by definition of $\Phisubtypeco$, $\pair{\cut{\tA}{k}}{\cut{\tB}{k}}
  \in \Phisubtypeco\left(\S\right)$.
  
  \item $\tA = \functype{\tA'}{\tA''}$. Similarly to the previous case we have
  $\tB = \functype{\tB'}{\tB''}$ with $\tB' \subtypeco \tA'$ and $\tA''
  \subtypeco \tB''$. Then we conclude by definition of $\S$ and $\Phisubtypeco$
  that $\pair{\cut{\tA}{k}}{\cut{\tB}{k}} =
  \pair{\functype{\cut{\tA'}{k-1}}{\cut{\tA''}{k-1}}}{\functype{\cut{\tB'}{k-1}}{\cut{\tB''}{k-1}}}
  \in \Phisubtypeco\left(\S\right)$.
\end{enumerate}
\end{proof}

\subsubsection{Correspondence between \texorpdfstring{$\mu$}{u}-types and infinite types}

Contractive $\mu$-types
characterize~\cite{journals/tcs/Courcelle83,DBLP:journals/toplas/AmadioC93,DBLP:journals/fuin/BrandtH98,Pierce:2002:TPL:509043}
a proper subset of $\Tree$ known as the \emphdef{regular trees} (trees whose
set of distinct subtrees is finite) and denoted $\TreeRegular$.
Given a contractive $\mu$-type $A$, $\toBTree{A}$ is the regular tree obtained
by completely unfolding all occurrences of $\rectype{V}{B}$ in $A$.
Def.~\ref{def:treeFunction} below extends that of~\cite{Pierce:2002:TPL:509043}
to union and data types. It is well-founded, relying on the lexicographical
extension of the standard order to $\pair{\length{\pi}}{\card{\irectype}{A}}$,
where $\card{\irectype}{A}$ is the number of occurrences of the $\irectype$
type constructor at the head position of $A$.

\begin{definition}
\label{def:treeFunction}
The function $\toBTree{\bullet} : \Type \to \TreeRegular$, mapping $\mu$-types to types, is defined inductively as follows: $$
\begin{array}{rcll}
\toBTree{a}(\epsilon)             & \eqdef & a \\
\toBTree{A_1 \star A_2}(\epsilon) & \eqdef & \star           & \quad \text{for $\star \in \set{\idatatype, \ifunctype, \iuniontype}$} \\
\toBTree{A_1 \star A_2}(i\pi)     & \eqdef & \toBTree{A_i}(\pi) & \quad \text{for $\star \in \set{\idatatype, \ifunctype, \iuniontype}$} \\
\toBTree{\rectype{V}{A}}(\pi)     & \eqdef & \toBTree{\substitute{V}{\rectype{V}{A}}{A}}(\pi) \\
\end{array} $$
\end{definition}

Commutation of $\toBTree{\bullet}$ with substitutions is as expected. 

\begin{lemma}
\label{lem:substitutionOfTrees}
$\toBTree{\substitute{V}{B}{A}} = \substitute{V}{\toBTree{B}}{\toBTree{A}}$.
\end{lemma}

\begin{proof}
We actualy prove the equivalente result $$\forall{k \in \Natural}.
\cut{\toBTree{\substitute{V}{B}{A}}}{k} =
\cut{(\substitute{V}{\toBTree{B}}{\toBTree{A}})}{k}$$ and conclude by
reflexivity of $\eqtypeco$ and Lem.~\ref{lem:cutEquivalenceCo}.

The proof is by induction on the lexicographical extension of the standard
order to
$\pair{h(\cut{\toBTree{\substitute{V}{B}{A}}}{k})}{\card{\irectype\iuniontype}{A}}$,
where $h : \TreeFinite \to \Natural$ is the height function for finite trees
and $\card{\irectype\iuniontype}{A}$ is the number of occurrences of both
$\irectype$ and $\iuniontype$ at the head of $A$.

We proceed by analyzing the possible forms of $A$ and assuming $k > 0$ since
the result for that case is immediate.
\begin{itemize}
  \item $A = V$: then $\cut{\toBTree{\substitute{V}{B}{V}}}{k} =
  \cut{\toBTree{B}}{k} = \cut{(\substitute{V}{\toBTree{B}}{V})}{k}$ by
  Lem.~\ref{lem:treeSubstitution}.
  
  \item $A = a \neq V$: then $\cut{\toBTree{\substitute{V}{B}{a}}}{k} =
  \cut{\toBTree{a}}{k} = a = \cut{(\substitute{V}{\toBTree{B}}{a})}{k}$ by
  definition of the interpretation and Lem.~\ref{lem:treeSubstitution}.
  
  \item $A = \datatype{D}{A'}$: then $$\kern-3em
\begin{array}{r@{\quad=\quad}l@{\quad}l}
\cut{\toBTree{\substitute{V}{B}{A}}}{k} & \cut{\toBTree{\datatype{\substitute{V}{B}{D}}{\substitute{V}{B}{A'}}}}{k} \\
                                        & \datatype{\cut{\toBTree{\substitute{V}{B}{D}}}{k-1}}{\cut{\toBTree{\substitute{V}{B}{A'}}}{k-1}} & \text{by Def.~\ref{def:treeFunction} and~\ref{def:treeCut}} \\
                                        & \datatype{\cut{(\substitute{V}{\toBTree{B}}{\toBTree{D}})}{k-1}}{\cut{(\substitute{V}{\toBTree{B}}{\toBTree{A'}})}{k-1}} & \text{by IH} \\
                                        & \cut{(\datatype{\substitute{V}{\toBTree{B}}{\toBTree{D}}}{\substitute{V}{\toBTree{B}}{\toBTree{A'}}})}{k} & \text{by Def.~\ref{def:treeCut}} \\
                                        & \cut{(\substitute{V}{\toBTree{B}}{\toBTree{\datatype{D}{A'}}})}{k} & \text{by Lem.~\ref{lem:treeSubstitution} and Def.~\ref{def:treeFunction}}
\end{array} $$
  
  \item $A = \functype{A'}{A''}$: this case is similar to the previous one.
  
  \item $A = \uniontype{A_1}{A_2}$: analysis for this case is similar to the
  previous ones but notice that we get the same $k$ when resorting to
  Def.~\ref{def:treeCut} (instead of $k-1$) before applying the inductive
  hypothesis. However, we are in conditions to apply it anyway since
  $$h(\cut{\toBTree{\substitute{V}{B}{A}}}{k}) \geq
  h(\cut{\toBTree{\substitute{V}{B}{A_i}}}{k}) \quad\text{but}\quad
  \card{\irectype\iuniontype}{A} > \card{\irectype\iuniontype}{A_i}$$ Hence, it
  is safe to conclude $\cut{\toBTree{\substitute{V}{B}{A}}}{k} =
  \cut{(\substitute{V}{\toBTree{B}}{\toBTree{A}})}{k}$.
  
  \item $A = \rectype{W}{A'}$: without loss of generality we can assume
  $\rename{V}{B} \avoids W$\footnote{We use the predicate $\sigma \avoids V$ to
  mean that there is no collition at all between $V$ and the variables in
  $\sigma$ (\ie $V \notin \dom{\sigma} \cap (\bigcup_{x \in \dom{\sigma}}
  \fv{\sigma x})$).}. Then $$
\begin{array}{r@{\quad=\quad}l@{\quad}l}
\cut{\toBTree{\substitute{V}{B}{A}}}{k} & \cut{\toBTree{\rectype{W}{\substitute{V}{B}{A'}}}}{k} \\
                                        & \cut{\toBTree{\substitute{W}{\rectype{W}{\substitute{V}{B}{A'}}}{\substitute{V}{B}{A'}}}}{k} & \text{by Def.~\ref{def:treeFunction}} \\
                                        & \cut{\toBTree{\substitute{V}{B}{\substitute{W}{A}{A'}}}}{k} \\
                                        & \cut{(\substitute{V}{\toBTree{B}}{\toBTree{\substitute{W}{A}{A'}}})}{k} & \text{by IH} \\
                                        & \cut{(\substitute{V}{\toBTree{B}}{\toBTree{A}})}{k} & \text{by Def.~\ref{def:treeFunction}}
\end{array} $$ Here we are in condition to apply the indutive hypothesis since
  $\card{\irectype\iuniontype}{A} >
  \card{\irectype\iuniontype}{\substitute{W}{A}{A'}} $ by contractiveness.
\end{itemize}
\end{proof}

The finite unfolding of a contractive $\mu$-type $A$ consists of recursively
replacing all occurrences of a bounded variable $V$ by $A$ itself a finite
number of times. We formalize a slightly more general variation of this idea in
the following lemma and prove its relation with $\toBTree{A}$.

\begin{lemma}
\label{lem:cutFiniteUnfolding}
Let $A = \rectype{V}{A'}$, $B$ any other $\mu$-type and $\sigma$ a
substitution. Define $$\unfoldf{A}{\sigma}{0} \eqdef B \qquad
\unfoldf{A}{\sigma}{n+1} \eqdef (\sigma \uplus
\rename{V}{\unfoldf{A}{\sigma}{n}})A'$$ Then, $\forall k \in \Natural.
\cut{\toBTree{\unfoldf{A}{\sigma}{k}}}{k} \eqtypeco \cut{\toBTree{\sigma
A}}{k}$.
\end{lemma}

\begin{proof}
By induction on $k$. We assume without loss of generality that $\sigma \avoids
V$.
\begin{itemize}
  \item $k = 0$. Then $\cut{\toBTree{B}}{0} = \bullet = \cut{\toBTree{\sigma
  A}}{0}$ by definition of the truncation.
  
  \item $k > 0$. By inductive hypothesis we have
  $\cut{\toBTree{\unfoldf{A}{\sigma}{k-1}}}{k-1} \eqtypeco
  \cut{\toBTree{\sigma A}}{k-1}$. Moreover, since $A = \rectype{V}{A'}$ is
  contractive, the first appearance of $V$ in $A'$ is at depth $n > 1$. So we
  have $k \leq k - 1 + n$ and, by Lem.~\ref{lem:substitutionOfEqtypesCo}
  and~\ref{lem:substitutionOfTrees}, we may conclude $$
\begin{array}{rcl@{\quad}l}
\cut{\toBTree{\unfoldf{A}{\sigma}{k}}}{k} & =         & \cut{\toBTree{(\sigma \uplus \rename{V}{\unfoldf{A}{\sigma}{k-1}})A'}}{k} \\
                                          & =         & \cut{(\substitute{V}{\toBTree{\unfoldf{A}{\sigma}{k-1}}}{\toBTree{\sigma A'}})}{k}            & \text{by Lem.~\ref{lem:substitutionOfTrees}} \\
                                          & =         & \cut{(\substitute{V}{\cut{\toBTree{\unfoldf{A}{\sigma}{k-1}}}{k-1}}{\toBTree{\sigma A'}})}{k} & k \leq k - 1 + n \\
                                          & \eqtypeco & \cut{(\substitute{V}{\cut{\toBTree{\sigma A}}{k-1}}{\toBTree{\sigma A'}})}{k}                 & \text{by Lem.~\ref{lem:substitutionOfEqtypesCo}} \\
                                          & =         & \cut{(\substitute{V}{\toBTree{\sigma A}}{\toBTree{\sigma A'}})}{k}                            & k \leq k - 1 + n \\
                                          & =         & \cut{\toBTree{(\sigma \uplus \rename{V}{\sigma A})A'}}{k}                                  & \quad\text{by Lem.~\ref{lem:substitutionOfTrees}} \\
                                          & =         & \cut{\toBTree{\sigma A}}{k}
\end{array} $$
\end{itemize}
\end{proof}

\begin{remark}
\label{rem:cutFiniteUnfolding}
It follows immediately from the previous result that for every $n \geq k$,
$\cut{\toBTree{\unfoldf{A}{\sigma}{n}}}{k} \eqtypeco
\cut{\toBTree{\sigma A}}{k}$.
\end{remark}

One of the main results of this section is the correspondence between the
equivalence relations $\eqtypemu$ and $\eqtypeco$ via the function
$\toBTree{\bullet}$. It follows from the lemma below that relates two
$\mu$-equivalent types with the truncation of their respective trees:

\begin{lemma}
\label{lem:cutEquivalenceMu}
$A \eqtypemu B$ iff $\forall k \in \Natural.\cut{\toBTree{A}}{k} \eqtypeco
\cut{\toBTree{B}}{k}$.
\end{lemma}

\begin{proof}
$\Rightarrow)$ This part of the proof is by induction on $A \eqtypemu B$
analyzing the last rule applied. Note that $\cut{\toBTree{A}}{0} = \bullet =
\cut{\toBTree{B}}{0}$ by definition of the truncation, so we only analyze the
cases where $k > 0$.
\begin{itemize}
  \item $\ruleEqmuRefl$: then $B = A$ and we conclude by reflexivity of
  $\eqtypeco$, $\cut{\toBTree{A}}{k} \eqtypeco \cut{\toBTree{A}}{k}$ for every
  $k > 0$.

  \item $\ruleEqmuTrans$: then $A \eqtypemu C$ and $C \eqtypemu B$. By
  inductive hypothesis $\cut{\toBTree{A}}{k} \eqtypeco \cut{\toBTree{C}}{k}$
  and $\cut{\toBTree{C}}{k} \eqtypeco \cut{\toBTree{B}}{k}$ for every $k > 0$.
  Then we conclude by transitivity of $\eqtypeco$.
  
  \item $\ruleEqmuSymm$: then $B \eqtypemu A$. By inductive hypothesis
  $\cut{\toBTree{B}}{k} \eqtypeco \cut{\toBTree{A}}{k}$ for every $k > 0$ and
  we conclude by symmetry of $\eqtypeco$.
  
  \item $\ruleEqmuFunc$: then $A = \functype{A'}{A''}, B = \functype{B'}{B''}$
  with $A' \eqtypemu B'$ and $A'' \eqtypemu B''$. By inductive hypothesis
  $\cut{\toBTree{A'}}{k} \eqtypeco \cut{\toBTree{B'}}{k}$ and
  $\cut{\toBTree{A''}}{k} \eqtypeco \cut{\toBTree{B''}}{k}$ for every $k > 0$.
  Then $$\cut{\toBTree{A}}{k} =
  \functype{\cut{\toBTree{A'}}{k-1}}{\cut{\toBTree{A''}}{k-1}} \eqtypeco
  \functype{\cut{\toBTree{B'}}{k-1}}{\cut{\toBTree{B''}}{k-1}} =
  \cut{\toBTree{B}}{k}$$
  
  \item $\ruleEqmuComp$: then $A = \datatype{D}{A'}, B = \datatype{D'}{B'}$
  with $A' \eqtypemu B'$ and $A'' \eqtypemu B''$. This case is similar to the
  previous one. We conclude directly from the inductive hypothesis and the
  definition of the truncation $$\cut{\toBTree{\datatype{D}{A'}}}{k} \eqtypeco
  \cut{\toBTree{\datatype{D'}{B'}}}{k}$$
  
  \item $\ruleEqmuUnionIdem$: then $A = \uniontype{B}{B}$. In this case we need
  to take into account that $B$ may be a union type as well and, when working
  with $\eqtypeco$, we must consider maximal union types. Let
  $\cut{\toBTree{A}}{k} = \maxuniontype{i \in 1..n}{\tA_i}$ and
  $\cut{\toBTree{B}}{k} = \maxuniontype{j \in 1..m}{\tB_j}$ with $\tA_j, \tB_j
  \neq \iuniontype$. It is immedate to see from the equality above that
  $n = 2*m$ and $\tA_j = \tA_{2*j} = \tB_j$ for every $j \in 1..m$. Finally we
  conclude by reflexivity of $\eqtypeco$ and $\ruleEqcoUnion$ $$
\begin{array}{rcl}
\cut{\toBTree{A}}{k} & =         & \maxuniontype{i \in 1..n}{\tA_i} \\
                     & =         & \uniontype{(\maxuniontype{j \in 1..m}{\tB_j})}{(\maxuniontype{j \in 1..m}{\tB_j})} \\
                     & \eqtypeco & \maxuniontype{j \in 1..m}{\tB_j} \\
                     & =         & \cut{\toBTree{B}}{k}
\end{array} $$
  
  \item $\ruleEqmuUnionComm$: then $A = \uniontype{C_1}{C_2}$ and $B =
  \uniontype{C_2}{C_1}$. As in the previous case consider $\cut{A}{k} =
  \maxuniontype{i \in 1..n}{\tA_i}$ and $\cut{B}{k} = \maxuniontype{j \in
  1..m}{\tB_j}$ with $\tA_i, \tB_j \neq \iuniontype$. Here $n = m > 1$, hence
  $n + m > 2$. Moreover, assuming $\tA_k$ is the last component of $C_1$ ($k
  \in 1..(n-1)$), we have $\tA_i = \tB_{i+k}$ if $i \leq n-k$, and $\tA_i =
  \tB_{i-(n-k)}$ if $i > n-k$. Thus, we conclude by reflexivity of $\eqtypeco$
  and $\ruleEqcoUnion$, $\cut{\toBTree{A}}{k} \eqtypeco \cut{\toBTree{B}}{k}$.
  
  \item $\ruleEqmuUnionAssoc$: then $A =
  \uniontype{C_1}{(\uniontype{C_2}{C_3})}$ and $B =
  \uniontype{(\uniontype{C_1}{C_2})}{C_3}$. Considering maximal union types as
  before we have $\cut{A}{k} = \maxuniontype{i \in 1..n}{\tA_i}$ and
  $\cut{B}{k} = \maxuniontype{j \in 1..m}{\tB_j}$ with $\tA_i, \tB_j \neq
  \iuniontype$ and $n = m > 2$. In this case we may conclude by resorting to
  the identity function in $1..n$, since $\tA_i = \tB_i$. Thus, by reflexivity
  and $\ruleEqcoUnion$, $\cut{\toBTree{A}}{k} \eqtypeco \cut{\toBTree{B}}{k}$.
  
  \item $\ruleEqmuUnion$: then $A = \uniontype{A_1}{A_2}, B =
  \uniontype{B_1}{B_2}$ with $A_1 \eqtypemu B_1$ and $A_2 \eqtypemu B_2$. By
  inductive hypothesis $\cut{\toBTree{A_1}}{k} \eqtypeco
  \cut{\toBTree{B_1}}{k}$ and $\cut{\toBTree{A_2}}{k} \eqtypeco
  \cut{\toBTree{B_2}}{k}$ for every $k \in \Natural$. Assume, without loss of
  generality $$
\begin{array}{r@{\quad\text{with}\quad}l}
\cut{\toBTree{A_1}}{k} = \maxuniontype{i \in 1..n}{\tA_i} & \tA_i \neq \iuniontype, i \in 1..n \\
\cut{\toBTree{B_1}}{k} = \maxuniontype{j \in 1..m}{\tB_j} & \tB_j \neq \iuniontype, j \in 1..m
\end{array} $$

  If $n + m > 2$, there exists $f : 1..n \to 1..m$, $g : 1..m \to 1..n$ such
  that $\tA_i \eqtypeco \tB_{f(i)}$ and $\tA_{g(j)} \eqtypeco \tB_j$. If not
  (\ie $n = m = 1$), we simply take $f = g = \mathit{id}$.
  
  Likewise, for $A_2$ and $B_2$ we have $$
\begin{array}{r@{\quad\text{with}\quad}l}
\cut{\toBTree{A_2}}{k} = \maxuniontype{i \in 1..n'}{\tA'_i} & \tA'_i \neq \iuniontype, i \in 1..n' \\
\cut{\toBTree{B_2}}{k} = \maxuniontype{j \in 1..m'}{\tB'_j} & \tB'_j \neq \iuniontype, j \in 1..m'
\end{array} $$ and there exists $f' : 1..n' \to 1..m'$, $g' : 1..m' \to 1..n'$
  such that $\tA'_i \eqtypeco \tB'_{f'(i)}$ and $\tA'_{g'(j)} \eqtypeco
  \tB'_j$.

  Finally, since $(n + n' + m + m') > 2$, we can apply $\ruleEqcoUnion$ to
  conclude $$
\begin{array}{rcl}
\cut{\toBTree{A}}{k} & =         & \uniontype{\cut{\toBTree{A_1}}{k}}{\cut{\toBTree{A_2}}{k}} \\
                     & =         & \uniontype{(\maxuniontype{i \in 1..n}{\tA_i})}{(\maxuniontype{i \in 1..n'}{\tA'_i})} \\
                     & \eqtypeco & \uniontype{(\maxuniontype{j \in 1..m}{\tB_j})}{(\maxuniontype{j \in 1..m'}{\tB'_j})} \\
                     & =         & \uniontype{\cut{\toBTree{B_1}}{k}}{\cut{\toBTree{B_2}}{k}} \\
                     & =         & \cut{\toBTree{B}}{k}
\end{array} $$
  
  \item $\ruleEqmuRec$: then $A = \rectype{V}{A'}, B = \rectype{V}{B'}$ with
  $A' \eqtypemu B'$. By inductive hypothesis $\cut{\toBTree{A'}}{k} \eqtypeco
  \cut{\toBTree{B'}}{k}$ and, by Lem.~\ref{lem:cutEquivalenceCo}, $\toBTree{A'}
  \eqtypeco \toBTree{B'}$.
  
  Now we consider the definition of $A_\sigma^n$ and $B_\sigma^n$ as in
  Lem.~\ref{lem:cutFiniteUnfolding} with $A_\sigma^0 \eqdef \sigma A'$ and
  $B_\sigma^0 \eqdef \sigma B'$. We claim that $\toBTree{A_{id}^n} \eqtypeco
  \toBTree{B_{id}^n}$ for every $n \in \Natural$. To prove this we proceed by
  induction on $n$
  \begin{itemize}
    \item $n = 0$. Then we have $\toBTree{A_{id}^0} = \toBTree{A'} \eqtypeco
    \toBTree{B'} = \toBTree{B_{id}^0}$ that holds by hypothesis.
    
    \item $n > 0$. By reflexivity
    $\substitute{V}{\toBTree{A_{id}^{n-1}}}{\toBTree{A'}} \eqtypeco
    \substitute{V}{\toBTree{A_{id}^{n-1}}}{\toBTree{A'}}$. Also, by inductive
    hypothesis, $\toBTree{A_{id}^{n-1}} \eqtypeco \toBTree{B_{id}^{n-1}}$ and,
    by hypothesis, $\toBTree{A'} \eqtypeco \toBTree{B'}$. Then we can apply
    Lem.~\ref{lem:substitutionOfEqtypesCo} and~\ref{lem:substitutionOfTrees},
    and conclude $$\toBTree{A_{id}^n} =
    \substitute{V}{\toBTree{A_{id}^{n-1}}}{\toBTree{A'}} \eqtypeco
    \substitute{V}{\toBTree{B_{id}^{n-1}}}{\toBTree{B'}} = \toBTree{B_{id}^n}$$
  \end{itemize}
  
  Finally, by Lem.~\ref{lem:cutEquivalenceCo}, $\cut{\toBTree{A_{id}^n}}{k}
  \eqtypeco \cut{\toBTree{B_{id}^n}}{k}$ for every $k, n \in \Natural$. Thus we
  conclude by Lem.~\ref{lem:cutFiniteUnfolding} $$\cut{\toBTree{A}}{k}
  \eqtypeco \cut{\toBTree{A_{id}^k}}{k} \eqtypeco \cut{\toBTree{B_{id}^k}}{k}
  \eqtypeco \cut{\toBTree{B}}{k}$$
  
  \item $\ruleEqmuFold$: then $A = \rectype{V}{A'}$ and $B =
  \substitute{V}{\rectype{V}{A'}}{A'}$. The result is immediate by definition
  of the interpretation, $\toBTree{A} = \toBTree{\rectype{V}{A'}} =
  \toBTree{\substitute{V}{\rectype{V}{A'}}{A'}} = \toBTree{B}$. Then
  $\cut{\toBTree{A}}{k} \eqtypeco \cut{\toBTree{B}}{k}$ for every $k \in
  \Natural$ by reflexivity of $\eqtypeco$.
  
  \item $\ruleEqmuContr$: then $B = \rectype{V}{B'}$ is contractive and $A
  \eqtypemu \substitute{V}{A}{B'}$. By inductive hypothesis and
  Lem.~\ref{lem:cutEquivalenceCo}, $\toBTree{A} \eqtypeco
  \toBTree{\substitute{V}{A}{B'}}$.
  
  As in the previous case we consider $B_\sigma^n$ from
  Lem.~\ref{lem:cutFiniteUnfolding}, this time with $B_\sigma^0 \eqdef \sigma
  A$. Now we show $\toBTree{A} \eqtypeco \toBTree{B_{id}^n}$ for every $n \in
  \Natural$, by induction on $n$
  
  \begin{itemize}
    \item $n = 0$. This case is immediate since $\toBTree{B_{id}^0} =
    \toBTree{A}$ by definition.
    
    \item $n > 0$. Then, by definition and Lem.~\ref{lem:substitutionOfTrees},
    $\toBTree{B_{id}^n} =
    \substitute{V}{\toBTree{B_{id}^{n-1}}}{\toBTree{B'}}$. By inductive
    hypothesis we know $\toBTree{A} \eqtypeco \toBTree{B_{id}^{n-1}}$ and, by
    Lem.~\ref{lem:substitutionOfEqtypesCo}, $\toBTree{B_{id}^n} \eqtypeco
    \substitute{V}{\toBTree{A}}{\toBTree{B'}}$. Finally we conclude by applying
    Lem.~\ref{lem:substitutionOfTrees} and transitivity of $\eqtypeco$ with
    hypothesis $\toBTree{A} \eqtypeco \toBTree{\substitute{V}{A}{B'}}$
    $$\toBTree{B_{id}^n} \eqtypeco \toBTree{\substitute{V}{A}{B'}} \eqtypeco
    \toBTree{A}$$
  \end{itemize}
  
  Then, by Lem.~\ref{lem:cutEquivalenceCo}, $\cut{\toBTree{A}}{k} \eqtypeco
  \cut{\toBTree{B_{id}^n}}{k}$ for every $k, n \in \Natural$. On the other
  hand, by Lem.~\ref{lem:cutFiniteUnfolding}, we know
  $\cut{\toBTree{B_{id}^k}}{k} \eqtypeco \cut{\toBTree{B}}{k}$. Thus, we
  conclude $$\cut{\toBTree{A}}{k} \eqtypeco \cut{\toBTree{B_{id}^k}}{k}
  \eqtypeco \cut{\toBTree{B}}{k}$$
\end{itemize}

$\Leftarrow)$ Let $\cut{\toBTree{A}}{k} \eqtypeco \cut{\toBTree{B}}{k}$ for
every $k \in \Natural$. Given $B = \rectype{V}{B'}$ it is immediate to see that
$\toBTree{\rectype{V}{B'}} = \toBTree{\substitute{V}{B}{B'}}$ while $B
\eqtypemu \substitute{V}{B}{B'}$, by definition of the interpretation and
$\ruleEqmuFold$ respectively. Moreover, since $\mu$-types are contractive, we
can assure that $\card{\irectype}{\substitute{V}{B}{B'}} <
\card{\irectype}{B}$. By a simple induction on $\card{\irectype}{B}$ we can
prove that for every $B \in \Type$ there exists $C \in \Type$ such that
$\card{\irectype}{C} = 0$, $B \eqtypemu C$ and $\toBTree{B} = \toBTree{C}$. It
is important to note that we are resorting to tree equality on this argument.
Thus, without loss of generality, we consider during the proof only the cases
where $\card{\irectype}{B} = 0$.

This proof is by induction on the lexicographical extension of the standard
order to $\pair{h(\cut{\toBTree{A}}{k})}{\card{\irectype}{A}}$, where $h :
\TreeFinite \to \Natural$ is the height function for finite trees. We proceed
by analyzing the possible forms of $A$.

Given $A, B \in \Type$ we can assume $$
\begin{array}{r@{\quad\text{with}\quad}l}
\toBTree{A} = \maxuniontype{i \in 1..n}{\tA_i} & \tA_i \neq \iuniontype, i \in 1..n \\
\toBTree{B} = \maxuniontype{j \in 1..m}{\tB_j} & \tB_j \neq \iuniontype, j \in 1..m
\end{array} $$ by Rem.~\ref{rem:maximalUnionTypes}. Moreover, since
$\card{\irectype}{B} = 0$ and by definition of the interpretation, we have
$B = \maxuniontype{j \in 1..m}{B_j}$ with $\toBTree{B_j} = \tB_j$ for every $j
\in 1..m$ (note that $B_j$ is a non-union type for every $j \in 1..m$).

Then, we can divide this proof in two cases, either
\begin{inparaenum}[(i)]
  \item $A$ and $B$ are both non-union types and thus $n = m = 1$; or
  \item at least one of them is a union type (\ie $n + m > 2$).
\end{inparaenum}

\begin{enumerate}[(i)]
  \item If $n = m = 1$. Here we analyze the shape of $A$:
  \begin{itemize}
    \item $A = a$. Then $\cut{\toBTree{A}}{k} = a$ for every $k > 0$ and, by
    Lem.~\ref{lem:equalityIsInvertible}, $\cut{\toBTree{B}}{k} = \cut{\tB_1}{k}
    = a$. Thus, by definition of the interpretation and tree truncation with
    the assumption $\card{\irectype}{B} = 0$, we have $B = a$ and conclude with
    $\ruleEqmuRefl$.
    
    \item $A = \datatype{D}{A'}$. Here we have $\cut{\toBTree{A}}{k} =
    \datatype{\cut{\toBTree{D}}{k-1}}{\cut{\toBTree{A'}}{k-1}}$ for every $k >
    0$ and, by Lem.~\ref{lem:equalityIsInvertible} once again,
    $\cut{\toBTree{B}}{k} = \datatype{\tB'_k}{\tB''_k}$ with
    $\cut{\toBTree{D}}{k-1} \eqtypeco \tB'_k$ and $\cut{\toBTree{A'}}{k-1}
    \eqtypeco \tB''_k$. With a similar analysis to the one made in
    Lem.~\ref{lem:cutEquivalenceCo}, by definition of the interpretation and
    tree truncation with the assumption $\card{\irectype}{B} = 0$, we can
    assure that $B = \datatype{D'}{B'}$ such that $\tB'_k =
    \cut{\toBTree{D'}}{k-1}$ and $\tB''_k = \cut{\toBTree{B'}}{k-1}$ for every
    $k > 0$. Then, we have $\cut{\toBTree{D}}{k-1} \eqtypeco
    \cut{\toBTree{D'}}{k-1}$ and $\cut{\toBTree{A'}}{k-1} \eqtypeco
    \cut{\toBTree{B'}}{k-1}$ and we can apply the inductive hypothesis to get
    $D \eqtypemu D'$ and $A' \eqtypemu B'$. Finally we conclude by
    $\ruleEqmuComp$, $\datatype{D}{A'} \eqtypemu \datatype{D'}{B'}$.

    \item $A = \functype{A'}{A''}$. Analysis for this case is similar to the
    previous one. From $\cut{\toBTree{A}}{k} =
    \functype{\cut{\toBTree{A'}}{k-1}}{\cut{\toBTree{A''}}{k-1}}$ we get $B =
    \functype{B'}{B''}$ with $\cut{\toBTree{A'}}{k-1} \eqtypeco
    \cut{\toBTree{B'}}{k-1}$ and $\cut{\toBTree{A''}}{k-1} \eqtypeco
    \cut{\toBTree{B''}}{k-1}$ for every $k > 0$. Then, by inductive hypothesis
    $A' \eqtypemu B'$ and $A'' \eqtypemu B''$. Thus we conclude with
    $\ruleEqmuFunc$, $\functype{A'}{A''} \eqtypemu \functype{B'}{B''}$.

    \item $A = \rectype{V}{A'}$ with $A'$ a non-union type. By definition of
    the interpretation we have $\cut{\toBTree{A}}{k} =
    \cut{\toBTree{\substitute{V}{A}{A'}}}{k} \eqtypeco \cut{\toBTree{B}}{k}$.
    Here we may apply the inductive hypothesis as
    $\card{\irectype}{\substitute{V}{A}{A'}} < \card{\irectype}{A}$. Then,
    $\substitute{V}{\rectype{V}{A'}}{A'} \eqtypemu B$. On the other hand,
    $\rectype{V}{A'} \eqtypemu \substitute{V}{\rectype{V}{A'}}{A'}$ by
    $\ruleEqmuFold$. Finally we conclude with $\ruleEqmuTrans$,
    $\rectype{V}{A'} \eqtypemu B$.
  \end{itemize}
  
  \item If $n + m > 2$. Then the last rule applied to derive
  $\cut{\toBTree{A}}{k} \eqtypeco \cut{\toBTree{B}}{k}$ is necessarily
  $\ruleEqcoUnion$. Then, there exists $f : 1..n \to 1..m, g : 1..m \to 1..n$
  such that $\cut{\tA_i}{k} \eqtypeco \cut{\toBTree{B_{f(i)}}}{k}$ and
  $\cut{\tA_{g(j)}}{k} \eqtypeco \cut{\toBTree{B_j}}{k}$ for every $i \in 1..n,
  j \in 1..m$.
  
  If $\card{\irectype}{A} \neq 0$, then $A = \rectype{V}{A'}$, $\toBTree{A} =
  \toBTree{\substitute{V}{A}{A'}}$ by definition and
  $\card{\irectype}{\substitute{V}{A}{A'}} < \card{\irectype}{A}$ by
  contractivity. Thus we can conclude directly from the inductive hypothesis
  with $\ruleEqmuFold$ and $\ruleEqmuTrans$ as before.
  
  If $\card{\irectype}{A} = 0$, by definition of the interpretation we have $A
  = \maxuniontype{i \in 1..n}{A_i}$ with $\toBTree{A_i} = \tA_i$ for every $i
  \in 1..n$. Hence, $\cut{\toBTree{A_i}}{k} \eqtypeco
  \cut{\toBTree{B_{f(i)}}}{k}$ and $\cut{\toBTree{A_{g(j)}}}{k} \eqtypeco
  \cut{\toBTree{B_j}}{k}$.
  
  Moreover, since $\tA_i, \tB_j \neq \iuniontype$, we are in the same situation
  as case (i) of this proof, so we can assure $A_i \eqtypemu B_{f(i)}$ and
  $A_{g(j)} \eqtypemu B_j$ for every $i \in 1..n, j \in 1..m$.
  
  Finally, we are under the hypothesis of Lem.~\ref{lem:unionEquivalence}, thus
  we conclude $\maxuniontype{i \in 1..n}{A_i} \eqtypemu \maxuniontype{j \in
  1..m}{B_j}$.
\end{enumerate}
\end{proof}

\begin{proposition}
\label{prop:eqtypeSoundnessAndCompleteness}
$A \eqtypemu B$ iff $\toBTree{A} \eqtypeco \toBTree{B}$.
\end{proposition}

\begin{proof}
This proposition follows from previous results shown on
Lem.~\ref{lem:cutEquivalenceCo} and~\ref{lem:cutEquivalenceMu}: $A \eqtypemu B$
iff $\forall k \in \Natural.\cut{\toBTree{A}}{k} \eqtypeco
\cut{\toBTree{B}}{k}$ iff $\toBTree{A} \eqtypeco \toBTree{B}$.
\end{proof}

To prove the correspondence between the subtyping relations we need to verify
that all variable assumptions in the subtyping context can be substituted by
convenient $\mu$-types before applying $\toBTree{\bullet}$.

\begin{lemma}
\label{lem:subtypeSoundnessWithHypothesis}
Let $\Sigma = \set{V_i \subtypemu W_i}_{i \in 1..n}$ be a subtyping context
and $\sigma$ a substitution such that $\dom{\sigma} = \set{V_i, W_i}_{i \in
1..n}$, $\sigma(V_i) = A_i$ and $\sigma(W_i) = B_i$ with $\dom{\sigma} \cap
\fv{\set{A_i, B_i}_{i \in 1..n}} = \varnothing$, $\toBTree{A_i} \subtypeco
\toBTree{B_i}$ and $A_i, B_i \in \Type$ for every $i \in 1..n$.
If $\sequTE{\Sigma}{A \subtypemu B}$, then $\toBTree{\sigma A} \subtypeco
\toBTree{\sigma B}$.
\end{lemma}

\begin{proof}
By induction on $\sequTE{\Sigma}{A \subtypemu B}$ analyzing the last rule
applied.
\begin{itemize}
  \item $\ruleSubmuRefl$: $A = B$ and the result is immediate by reflexivity of
  $\subtypeco$.
  
  \item $\ruleSubmuTrans$: $\sequTE{\Sigma}{A \subtypemu C}$ and
  $\sequTE{\Sigma}{C \subtypemu B}$ for some $C \in \Type$. By inductive
  hypothesis $\toBTree{\sigma A} \subtypeco \toBTree{\sigma C}$ and
  $\toBTree{\sigma C} \subtypeco \toBTree{\sigma B}$ for every $\sigma$
  satisfying the hypothesis of the lemma. Then we conclude by transitivity of
  $\subtypeco$.
  
  \item $\ruleSubmuHyp$: $A = V$ and $B = W$ with $\Sigma = \Sigma', V
  \subtypemu W$. Then $\sigma A = A_n$, $\sigma B = B_n$ and the result is
  immediate since, by hypothesis of the lemma, $\toBTree{A_i} \subtypeco
  \toBTree{B_i}$ for every $i \in 1..n$.
  
  \item $\ruleSubmuEq$: $\sequTE{}{A \eqtypemu B}$ and, since $\eqtypemu$ is a
  congruence, we have $\sequTE{}{\sigma A \eqtypemu \sigma B}$ for every
  substitution. So we can take $\sigma$ satisfying the hypothesis of the lemma.
  Then, by Prop.~\ref{prop:eqtypeSoundnessAndCompleteness}, $\toBTree{\sigma A}
  \eqtypeco \toBTree{\sigma B}$ and we conclude by Lem.~\ref{lem:eqImpliesSub},
  $\toBTree{\sigma A} \subtypeco \toBTree{\sigma B}$.
  
  \item $\ruleSubmuFunc$: $A = \functype{A'}{A''}$ and $B = \functype{B'}{B''}$
  with $\sequTE{\Sigma}{B' \subtypemu A'}$ and $\sequTE{\Sigma}{A'' \subtypemu
  B''}$. By inductive hypothesis we have $\toBTree{\sigma B'} \subtypeco
  \toBTree{\sigma A'}$ and $\toBTree{\sigma A''} \subtypeco \toBTree{\sigma
  B''}$. Then $$
\begin{array}{rcl}
\toBTree{\sigma A} & =          & \toBTree{\functype{\sigma A'}{\sigma A''}} \\
                   & =          & \functype{\toBTree{\sigma A'}}{\toBTree{\sigma A''}} \\
                   & \subtypeco & \functype{\toBTree{\sigma B'}}{\toBTree{\sigma B''}} \\
                   & =          & \toBTree{\functype{\sigma B'}{\sigma B''}} \\
                   & =          & \toBTree{\sigma B}
\end{array} $$
  
  \item $\ruleSubmuComp$: $A = \datatype{D}{A'}$ and $B = \datatype{D'}{B'}$
  with $\sequTE{\Sigma}{D \subtypemu D'}$ and $\sequTE{\Sigma}{A' \subtypemu
  B'}$. Similarly to the previous case we conclude from the inductive
  hypothesis that $\toBTree{\datatype{\sigma D}{\sigma A'}} \subtypeco
  \toBTree{\datatype{\sigma D'}{\sigma B'}}$.
  
  \item $\ruleSubmuUnionL$: $A = \uniontype{A'}{A''}$ with $\sequTE{\Sigma}{A'
  \subtypemu B}$ and $\sequTE{\Sigma}{A'' \subtypemu B}$. By inductive
  hypothesis $\toBTree{\sigma A'} \subtypeco \toBTree{\sigma B}$ and
  $\toBTree{\sigma A''} \subtypeco \toBTree{\sigma B}$. Let $$
\begin{array}{r@{\ =\ }l@{\qquad}r@{\ \neq\ }l}
\toBTree{\sigma A'}  & \maxuniontype{i \in 1..m}{\tA'_i}   & \tA'_i  & \iuniontype \\
\toBTree{\sigma A''} & \maxuniontype{j \in 1..m'}{\tA''_j} & \tA''_j & \iuniontype \\
\toBTree{\sigma B}   & \maxuniontype{k \in 1..l}{\tB_k}    & \tB_k   & \iuniontype
\end{array} $$
Now we need to consider the following situations:
  \begin{enumerate}
    \item $m = m' = l = 1$. Then we conclude directly from the inductive
    hypothesis by applying $\ruleSubcoUnion$,
    $\uniontype{\toBTree{\sigma A'}}{\toBTree{\sigma A''}} =
    \uniontype{\tA'_1}{\tA''_1} \subtypeco \tB_1 = \toBTree{\sigma B}$.
    
    \item $m + l > 2$. Then there exists $f : 1..m \to 1..l$ such that $\tA'_i
    \subtypeco \tB_{f(i)}$ and there are two possible cases:
    \begin{enumerate}
      \item $m' = l = 1$. Then $\tA'_i \subtypeco \tB_1$ (\ie $f$ is a constant
      function) and $\tA''_1 \subtypeco \tB_1$. Then we conclude by
      $\ruleSubcoUnion$ $$\uniontype{\toBTree{\sigma A'}}{\toBTree{\sigma A''}}
      = \uniontype{(\maxuniontype{i \in 1..m}{\tA'_i})}{\tA''_1} \subtypeco
      \tB_1 = \toBTree{\sigma B}$$
      
      \item $m' + l > 2$. Then there exists $g : 1..m' \to 1..l$ such that
      $\tA''_j \subtypeco \tB_{g(j)}$. Once again we conclude by
      $\ruleSubcoUnion$ $$
\begin{array}{rcl}
\uniontype{\toBTree{\sigma A'}}{\toBTree{\sigma A''}} & =          & \uniontype{(\maxuniontype{i \in 1..m}{\tA'_i})}{(\maxuniontype{j \in 1..m'}{\tA''_j})} \\
                                                      & \subtypeco & \maxuniontype{k \in 1..l}{\tB_k} \\
                                                      & =          & \toBTree{\sigma B}
\end{array} $$
    \end{enumerate}
    
    \item The only case left to analyze is $m = l = 1$ and $m' + l > 2$ that
    are similar to one where $m' = l = 1$ and $m + l > 2$.
  \end{enumerate}
  So we conclude that $\toBTree{\sigma A} =
  \uniontype{\toBTree{\sigma A'}}{\toBTree{\sigma A''}} \subtypeco
  \toBTree{\sigma B}$.
  
  \item $\ruleSubmuUnionRL$: $B = \uniontype{B'}{B''}$ with $\sequTE{\Sigma}{A
  \subtypemu B'}$. By inductive hypothesis $\toBTree{\sigma A} \subtypeco
  \toBTree{\sigma B'}$. Let $$
\begin{array}{r@{\ =\ }l@{\qquad}r@{\ \neq\ }l}
\toBTree{\sigma A}   & \maxuniontype{i \in 1..m}{\tA_i}    & \tA_i   & \iuniontype \\
\toBTree{\sigma B'}  & \maxuniontype{j \in 1..l}{\tB'_j}   & \tB'_j  & \iuniontype \\
\toBTree{\sigma B''} & \maxuniontype{k \in 1..l'}{\tB''_k} & \tB''_k & \iuniontype
\end{array} $$
Here there are two possible situations:
  \begin{enumerate}
    \item $m = l = 1$. Then $\tA_1 \subtypeco \tB'_1$ and we conclude by
    $\ruleSubcoUnion$ $$\toBTree{\sigma A} = \tA_1 \subtypeco
    \uniontype{\tB'_1}{(\maxuniontype{k \in 1..l'}{\tB''_k})} =
    \uniontype{\toBTree{\sigma B'}}{\toBTree{\sigma B''}}$$
    \item $m + l > 2$. Then there exists $f : 1..m \to 1..l$ such that $\tA_i
    \subtypeco \tB'_{f(i)}$. We are again in a situation where all the
    conditions for $\ruleSubcoUnion$ hold $$
\begin{array}{rcl}
\toBTree{\sigma A} & =          & \maxuniontype{i \in 1..m}{\tA_i} \\
                   & \subtypeco & \uniontype{(\maxuniontype{j \in 1..l}{\tB'_j})}{(\maxuniontype{k \in 1..l'}{\tB''_k})} \\
                   & =          & \uniontype{\toBTree{\sigma B'}}{\toBTree{\sigma B''}}
\end{array} $$
  \end{enumerate}
  So we conclude that $\toBTree{\sigma A} \subtypeco \uniontype{\toBTree{\sigma B'}}{\toBTree{\sigma B''}} = \toBTree{\sigma B}$.
  
  \item $\ruleSubmuUnionRR$: this case is similar to the previous one, with
  $B = \uniontype{B'}{B''}$ and $\sequTE{\Sigma}{A \subtypemu B''}$.
  
  \item $\ruleSubmuRec$: $A = \rectype{V}{A'}, B = \rectype{W}{B'}$ with
  $\sequTE{\Sigma, V \subtypemu W}{A' \subtypemu B'}, W \notin \fv{A'}$ and $V
  \notin \fv{B'}$. Let $\sigma$ be a substitution satisfying the hypothesis of
  the lemma $$
 \begin{array}{cl}
(1) & \dom{\sigma} = \set{V_i, W_i}_{i \in 1..n} \\
(2) & \sigma(V_i) = A_i \text{ and } \sigma(W_i) = B_i \\
(3) & \set{V_i, W_i}_{i \in 1..n} \cap \fv{\set{A_i, B_i}_{i \in 1..n}} = \varnothing \\
(4) & \toBTree{A_i} \subtypeco \toBTree{B_i}
\end{array} $$
  Now consider $\unfoldf{A}{\sigma}{m}$ and $\unfoldf{B}{\sigma}{m}$ as in
  Lem.~\ref{lem:cutFiniteUnfolding}, recall $$
\begin{array}{rcl}
\unfoldf{A}{\sigma}{0} & \eqdef & \bullet \\
\unfoldf{B}{\sigma}{0} & \eqdef & \bullet
\end{array}
\hspace{.5cm}
\begin{array}{rcl}
\unfoldf{A}{\sigma}{m+1} & \eqdef & (\sigma \uplus \rename{V}{\unfoldf{A}{\sigma}{m}})A' \\
\unfoldf{B}{\sigma}{m+1} & \eqdef & (\sigma \uplus \rename{W}{\unfoldf{B}{\sigma}{m}})B'
\end{array} $$ and also the substitution $\sigma_m = (\sigma \uplus
  \rename{V}{\unfoldf{A}{\sigma}{m}} \uplus
  \rename{W}{\unfoldf{B}{\sigma}{m}})$ for each $m \in \Natural$. Notice that
  $\sigma_m A' = \unfoldf{A}{\sigma}{m+1}$ since $W \notin \fv{A'}$. Similarly,
  $\sigma_m B' = \unfoldf{B}{\sigma}{m+1}$.
  
  It is immediate to see that $\sigma_0$ satisfies the hypothesis of the lemma
  for the extended context $\Sigma, V \subtypemu W$, taking $A_{n+1} =
  \unfoldf{A}{\sigma}{0} = \bullet = \unfoldf{B}{\sigma}{0} = B_{n+1}$. This
  allow us to apply the inductive hypothesis and conclude that
  $\toBTree{\unfoldf{A}{\sigma}{1}} = \toBTree{\sigma_0 A'} \subtypeco
  \toBTree{\sigma_0 B'} = \toBTree{\unfoldf{B}{\sigma}{1}}$, and once again we
  are under the hypothesis of the lemma, this time with $\sigma_1$. Thus,
  directly from the inductive hypothesis (applied as many times as needed) we
  have $\toBTree{\unfoldf{A}{\sigma}{m}} \subtypeco
  \toBTree{\unfoldf{B}{\sigma}{m}}$ for every $m \in \Natural$.
  
  Then, by Lem.~\ref{lem:cutSubtypingCo},
  $\cut{\toBTree{\unfoldf{A}{\sigma}{m}}}{k} \subtypeco
  \cut{\toBTree{\unfoldf{B}{\sigma}{m}}}{k}$ for every $k \in \Natural$.
  Moreover, by Lem.~\ref{lem:cutFiniteUnfolding} we have
  $\cut{\toBTree{\sigma A}}{k} \eqtypeco
  \cut{\toBTree{\unfoldf{A}{\sigma}{k}}}{k}$ and
  $\cut{\toBTree{\unfoldf{B}{\sigma}{k}}}{k} \eqtypeco \cut{\toBTree{\sigma
  B}}{k}$. Finally, by Lem.~\ref{lem:eqImpliesSub} and transitivity of
  subtyping we get $\cut{\toBTree{\sigma A}}{k} \subtypeco \cut{\toBTree{\sigma
  B}}{k}$ and conclude with Lem.~\ref{lem:cutSubtypingCo}.
\end{itemize}
\end{proof}

Finally, as mentioned above, the following proposition and
Lem.~\ref{lem:subtypingIsInvertible} allows us to prove
Prop.~\ref{prop:subtypingIsInvertible}.

\begin{proposition}
\label{prop:subtypeSoundnessAndCompleteness}
$A \subtypemu B$ iff $\toBTree{A} \subtypeco \toBTree{B}$.
\end{proposition}

\begin{proof}
$\Rightarrow)$ This part of the proof follows directly from
Lem.~\ref{lem:subtypeSoundnessWithHypothesis}, taking $\Sigma$ an empty
subtyping context and thus $\sigma$ results in the identity substitution. Hence
from $A \subtypemu B$ we get $\toBTree{A} \subtypeco \toBTree{B}$.

$\Leftarrow)$ For the converse we prove the equivalent result: if $\forall k
\in \Natural.\cut{\toBTree{A}}{k} \subtypeco \cut{\toBTree{B}}{k}$ then, $A
\subtypemu B$. And finally conclude by Lem.~\ref{lem:cutSubtypingCo}.

Let $\cut{\toBTree{A}}{k} \subtypeco \cut{\toBTree{B}}{k}$ for every $k \in
\Natural$. As in the proof for Lem.~\ref{lem:cutEquivalenceMu}, we only
consider the cases where $\card{\irectype}{B} = 0$ and proceed by induction on
the lexicographical extension of the standard order to
$\pair{h(\cut{\toBTree{A}}{k})}{\card{\irectype}{A}}$, analyzing the possible
forms of $A$.\
\begin{itemize}
  \item $A = a$. By definition of of the interpretation and tree truncation we
  have $\cut{\toBTree{A}}{k} = a$ for every $k > 0$. Now, by definition of
  $\subtypeco$, only two rules apply:
  \begin{itemize}
    \item $\ruleSubcoRefl$: in this case we have $\cut{\toBTree{B}}{k} = a =
    B$, by definition of the interpretation, and we conclude with
    $\ruleSubmuRefl$.
    
    \item $\ruleSubcoUnion$: by definition of the interpretation once again, we
    have $B = \maxuniontype{1 \in 1..n}{B_i}$ and $$\cut{\toBTree{a}}{k}
    \subtypeco \maxuniontype{i \in 1..n}{\cut{\toBTree{B_i}}{k}}$$ with
    $\cut{\toBTree{a}}{k} \subtypeco \cut{\toBTree{B_j}}{k} \neq \iuniontype$
    for some $j \in 1..n$, $n > 1$. Now the only applicable rule is
    $\ruleSubcoRefl$, thus $\cut{\toBTree{B_j}}{k} = a = B_j$. Then, by
    $\ruleSubmuRefl$, $\ruleSubmuUnionRL$ and $\ruleSubmuUnionRR$, we conclude
    $A \subtypemu \maxuniontype{i \in 1..n}{B_i}$.
  \end{itemize}
  
  \item $A = \datatype{D}{A'}$. As before, by definition of the interpretation
  and tree truncation with $k > 0$, $\cut{\toBTree{A}}{k} =
  \datatype{\cut{\toBTree{D}}{k-1}}{\cut{\toBTree{A'}}{k-1}} \subtypeco
  \cut{\toBTree{B}}{k}$. The only two possible cases here are:
  \begin{itemize}
    \item $\ruleSubcoComp$: by definition of the interpretation and tree
    truncation once again, we have $B = \datatype{D'}{B'}$ with
    $\cut{\toBTree{D}}{k-1} \subtypeco \cut{\toBTree{D'}}{k-1}$ and
    $\cut{\toBTree{A'}}{k-1} \subtypeco \cut{\toBTree{B'}}{k-1}$. Then, by
    inductive hypotesis, $D \subtypemu D'$ and $A' \subtypemu B'$. Finally we
    conclude by $\ruleSubmuComp$, $\datatype{D}{A'} \subtypemu
    \datatype{D'}{B'}$.
    
    \item $\ruleSubcoUnion$: with a similar analysis as the case
    $\ruleSubcoUnion$ for $A = a$, we have $B = \maxuniontype{i \in 1..n}{B_i}$
    and $$\cut{\toBTree{\datatype{D}{A'}}}{k} \subtypeco \maxuniontype{i \in
    1..n}{\cut{\toBTree{B_i}}{k}}$$ with $\cut{\toBTree{\datatype{D}{A'}}}{k}
    \subtypeco \cut{\toBTree{B_j}}{k} \neq \iuniontype$ for some $j \in 1..n$,
    $n > 1$. Then, by definition of $\subtypeco$, it is necessarily the case
    $B_j = \datatype{D'}{B'}$ with $\cut{\toBTree{D}}{k} \subtypeco
    \cut{\toBTree{D'}}{k}$ and $\cut{\toBTree{A'}}{k} \subtypeco
    \cut{\toBTree{B'}}{k}$. Now, as in the previous case, we have
    $\datatype{D}{A'} \subtypemu B_j$ by inductive hypothesis. Finally, with
    $\ruleSubmuUnionRL$ and $\ruleSubmuUnionRR$, we conclude $\datatype{D}{A'}
    \subtypemu \maxuniontype{i \in 1..n}{B_i}$.
  \end{itemize}
  
  \item $A = \functype{A'}{A''}$. The only two applicable rules here are
  $\ruleSubcoFunc$ and $\ruleSubcoUnion$. Both cases are similar to the ones
  exposed for $\idatatype$, concluding directly from the inductive hypothesis
  and the application of $\ruleSubmuFunc$ in the former while
  $\ruleSubmuUnionRL$ and $\ruleSubmuUnionRR$ are used in the latter.
  
  \item $A = \maxuniontype{i \in 1..n}{A_i}$ with $A_i$ a non-union type for
  every $i \in 1..n$, $n > 1$. This case is slightly simpler than the others as
  the only applicable rule is $\ruleSubcoUnion$. Let $B = \maxuniontype{j \in
  1..m}{B_j}$ with $B_j$ a non-union type for $j \in 1..m$. Note that $m$ is
  not necessarily greater then 1. By definition of the interpretation and tree
  truncation we have, from $\ruleSubcoUnion$, $\exists f : 1..n \to 1..m$ such
  that $\cut{\toBTree{A_i}}{k} \subtypeco \cut{\toBTree{B_{f(i)}}}{k}$ for
  every $i \in 1..n$. Then, by inductive hypothesis, $A_i \subtypemu B_{f(i)}$
  for every $i \in 1..n$. Now, by properly applying $\ruleSubmuUnionRL$ and
  $\ruleSubmuUnionRR$ on each case, we get $A_i \subtypemu B$ for every $i \in
  1..n$. Finally we conclude by multiple applications of $\ruleSubmuUnionL$,
  $\maxuniontype{i \in 1..n}{A_i}$.
  
  \item $A = \rectype{V}{A'}$. Then $\cut{\toBTree{A}}{k} =
  \cut{\toBTree{\substitute{V}{\rectype{V}{A'}}{A'}}}{k} \subtypeco
  \cut{\toBTree{B}}{k}$. By inductive hypothesis, with
  $\card{\irectype}{\substitute{V}{A}{A'}} < \card{\irectype}{A}$, we have
  $\substitute{V}{\rectype{V}{A'}}{A'} \subtypemu B$. On the other hand, by
  $\ruleEqmuFold$ and $\ruleSubmuEq$, we get $\rectype{V}{A'} \subtypemu
  \substitute{V}{\rectype{V}{A'}}{A'}$ and we conclude by $\ruleSubmuTrans$,
  $\rectype{V}{A'} \subtypemu B$.
\end{itemize}
\end{proof}

\begin{proposition}
\label{prop:subtypingIsInvertible}
\begin{enumerate}
  \item If $\datatype{D}{A} \subtypemu \datatype{D'}{A'}$, then $D \subtypemu
  D'$ and $A \subtypemu A'$.
  
  \item If $\functype{A}{B} \subtypemu \functype{A'}{B'}$, then $A' \subtypemu
  A$ and $B \subtypemu B'$.
\end{enumerate}
\end{proposition}

\begin{proof}
This result follows immediately from Lem.~\ref{lem:subtypingIsInvertible} and
Prop.~\ref{prop:subtypeSoundnessAndCompleteness}.
\end{proof}

\subsubsection{Further properties on \texorpdfstring{$\mu$}{u}-types}

We conclude the section with a simple but useful result on the preservation of
the structure of non-union types by means of subtyping. Define the set of
\emphdef{union contexts} as the expressions generated by the following grammar
$$\U ::= \Box \mathrel| \uniontype{\U}{A} \mathrel| \uniontype{A}{\U}$$

\begin{lemma}
\label{lem:noMuOnHeadPosition}
For every type $A \in \Type$ there exists $A' \in \Type$ such that $A \eqtypemu
A'$ and $\card{\irectype}{A'} = 0$. Moreover, if $\card{\irectype}{A} = 0$ then
$A$ and $A'$ have the same outermost type constructor.
\end{lemma}

\begin{proof}
By induction in $\card{\irectype}{A}$.
\begin{itemize}
  \item $\card{\irectype}{A} = 0$: the result is immediate taking $A' = A$.
  Notice that the second part of the statement holds trivially.
  
  \item $\card{\irectype}{A} > 0$: then $A = \rectype{V}{A''}$ and by rule
  $\ruleEqmuFold$ $A \eqtypemu \substitute{V}{A}{A''}$. Since $\mu$-types are
  contractive we have $\card{\irectype}{\substitute{V}{A}{A''}} <
  \card{\irectype}{A}$. Then, by inductive hypothesis, there exists $A' \in
  \Type$ such that $A \eqtypemu A'$, $\card{\irectype}{A'} = 0$. Finally we
  conclude by rule $\ruleEqmuTrans$.
\end{itemize}
\end{proof}

\begin{lemma}\label{lem:supertypesOfNonUnionTypes}
If $\U[A] \subtypemu B$ and $A$ is a non-union type, then there exists a
non-union type $A' \in \Type$ such that
\begin{inparaenum}[(i)]
  \item $B \eqtypemu \U'[A']$;
  \item $A \subtypemu A'$; and
  \item $A$ and $A'$ have the same outermost type constructor.
\end{inparaenum}
\end{lemma}


\begin{proof}
By induction on the union context $\U$. Without loss of generality we can
assume $\card{\irectype}{A} = 0$, by Lem.~\ref{lem:noMuOnHeadPosition}.
\begin{itemize}
  \item $\U = \Box$. We have $A \subtypemu B$. By
  Prop.~\ref{prop:subtypeSoundnessAndCompleteness}, $\toBTree{A} \subtypeco
  \toBTree{B}$ where $\toBTree{A} \neq \iuniontype$ by hypothesis. Let
  $\toBTree{B} = \maxuniontype{i \in 1..n}{\tB_i}$ with $\tB_i \neq
  \iuniontype$ for $i \in i..n$. Note that $\tB_i$ is a subtree of the regular
  tree $\toBTree{B}$, thus it is regular too. Then, for every $i \in 1..n$
  there exists $C_i \in \Type$ such that $\toBTree{C_i} = \tB_i$. Moreover,
  taking $C = \maxuniontype{i \in 1..n}{C_i}$ we have $\toBTree{C} =
  \toBTree{B}$, hence $C \eqtypemu B$ by
  Prop.~\ref{prop:eqtypeSoundnessAndCompleteness}.
  \begin{itemize}
    \item If $n = 1$ (\ie $\toBTree{B} = B_1 \neq \iuniontype$) the only
    applicable rules are $\ruleSubcoRefl$, $\ruleSubcoFunc$ or
    $\ruleSubcoComp$, hence both trees have the same type constructor on the
    root. Applying Lem.~\ref{lem:noMuOnHeadPosition} on $B$ yields a type $A'$
    such that $B \eqtypemu A'$, thus proving the first item with $U' = \Box$.
    This type $A'$ has the same outermost type constructor as $B$, which we
    already saw is the same as $A$, hence proving item (iii). We are left to
    prove the second item. This follows from $A \subtypemu B$, $B \eqtypemu A'$
    by rules $\ruleEqmuTrans$ and $\ruleSubmuEq$.
    
    \item If $n > 1$, then the only applicable rule is $\ruleSubcoUnion$ and we
    have $\toBTree{A} \subtypeco \toBTree{C_j} = \tB_j \neq \iuniontype$ and,
    by Prop.~\ref{prop:subtypeSoundnessAndCompleteness}, $A \subtypemu C_j$ for
    some $j \in 1..n$. Note that both trees must have the same constructor in
    the root since neither of them is a union type
    (Lem.~\ref{lem:subtypingIsInvertible}).
    Then we take the union context $$\U' =
    \uniontype{\uniontype{\uniontype{\uniontype{C_1}{\ldots}}{\Box_j}}{\ldots}}{C_n}$$
    and, by Lem.~\ref{lem:noMuOnHeadPosition}, there exists $A' \in \Type$ such
    that $A' \eqtypemu C_j$, $\card{\irectype}{A'} = 0$ and has the same
    outermost type constructor than $C$. Finally, we have $$B \eqtypemu C
    \eqtypemu \U[A']$$ while $A \subtypemu A'$ and both have the same outermost
    type constructor.
  \end{itemize}

  \item $\U =
  \uniontype{\uniontype{\uniontype{\uniontype{C_1}{\ldots}}{\Box_k}}{\ldots}}{C_m}$
  with $m > 1$, where $C_k$ with $k \in 1..m$ is the position of $\Box$ within $\U$
  (\ie $C_k = A$ in $\U[A]$). We can assume without loss of generality that $C_j$
  is a non-union type for every $j \in 1..m$.
  
  From $\U[A] \subtypemu B$ and
  Prop.~\ref{prop:subtypeSoundnessAndCompleteness} we have $\toBTree{\U[A]}
  \subtypeco \toBTree{B}$. By definition $$\toBTree{\U[A]} =
  \uniontype{\uniontype{\uniontype{\uniontype{\toBTree{C_1}}{\ldots}}{\toBTree{A}}}{\ldots}}{\toBTree{C_m}}$$
  with $\toBTree{C_j} \neq \iuniontype$ for every $j \in 1..m$.
  
  Assume once again $\toBTree{B} = \maxuniontype{i \in 1..n}{\tB_i}$ with $B_i
  \neq \iuniontype$ for $i \in 1..n$. The only subtyping rule that applies
  here is $\ruleSubcoUnion$ since $m > 1$, hence $n + m > 2$. Then there exists
  $f : 1..m \to 1..n$ such that $\toBTree{C_j} \subtypeco \tB_{f(j)}$ for every
  $j \in 1..m$.

  Notice that $\U = \uniontype{\U''}{C_n}$ or $\U = \uniontype{C_1}{\U''}$ for
  some proper union context $\U''$. Hence, by construction $$\U'' =
  \uniontype{\uniontype{\uniontype{\uniontype{C_1}{\ldots}}{\Box_k}}{\ldots}}{C_{m-1}}
  \quad\text{or}\quad \U'' =
  \uniontype{\uniontype{\uniontype{\uniontype{C_2}{\ldots}}{\Box_k}}{\ldots}}{C_m}$$
  In either case, by rule $\ruleSubcoUnion$, we have $\toBTree{\U''[A]}
  \subtypeco \toBTree{B}$, hence $\U''[A] \subtypemu B$ by
  Prop.~\ref{prop:subtypeSoundnessAndCompleteness}.
  
  Finally, we can apply the inductive hypothesis to conclude that $B \eqtypemu
  \U'[A']$ with $A' \in \Type$ a non-union type such that $A \subtypemu A'$ and
  both have the same outermost type constructor. 
\end{itemize}
\end{proof}


\subsection{Typing Schemes}
\label{sec:typinsSchemes}

\begin{figure}[t] $$
\begin{array}{c}
\multicolumn{1}{l}{\textbf{Patterns}}
\\
\\
\Rule{\theta(x) = A}{\sequP{\theta}{\matchable{x} : A}}{\rulePMatch}
\quad
\Rule{}{\sequP{\theta}{\constterm{c} : \consttype{c}}}{\rulePConst} 
\quad
\Rule{\sequP{\theta}{p : D} \quad \sequP{\theta}{q : A}}{\sequP{\theta}{\dataterm{p}{q} : \datatype{D}{A}}}{\rulePComp}
\\
\\
\multicolumn{1}{l}{\textbf{Terms}}
\\
\\
\Rule{\Gamma(x) = A}{\sequT{\Gamma}{x : A}}{\ruleTVar}
\quad
\Rule{}{\sequT{\Gamma}{\constterm{c} : \consttype{c}}}{\ruleTConst}
\quad
\Rule{\sequT{\Gamma}{r : D} \quad \sequT{\Gamma}{u : A}}{\sequT{\Gamma}{\dataterm{r}{u} : \datatype{D}{A}}}{\ruleTComp} \\
\\
\Rule{
\begin{array}{c}
  \lista{\sequC{\theta_i}{p_i : A_i}}_{i \in 1..n} \text{ compatible}\vspace{.1cm}
  \\
  (\sequP{\theta_i}{p_i : A_i})_{i\in 1..n}
  \quad
  (\dom{\theta_i} = \fm{p_i})_{i\in 1..n}
  \quad
  (\sequT{\Gamma, \theta_i}{s_i : B})_{i\in 1..n}
\end{array}
     }
     {\sequT{\Gamma}{(\absterm{p_i}{s_i}{\theta_i})_{i \in 1..n} : \functype{\maxuniontype{i \in 1..n}{A_i}}{B}}}
     {\ruleTAbs}
\\
\\
\Rule{\sequT{\Gamma}{r : \functype{\maxuniontype{i \in 1..n}{A_i}}{B}}
      \quad
      \sequT{\Gamma}{u : A_k}
      \quad
      k \in 1..n
      }
      {\sequT{\Gamma}{\appterm{r}{u} : B}}
      {\ruleTApp} 
\quad
\Rule{\sequT{\Gamma}{s : A}
      \quad
      \sequTE{}{A \subtypemu A'}
     }
     {\sequT{\Gamma}{s : A'}}
     {\ruleTSubs}
\end{array} $$
\caption{Typing rules for patterns and terms}\label{fig:typingSchemesForPatternsAndTerms}
\end{figure}

A \emphdef{typing context} $\Gamma$ (or $\theta$) is a partial function from term variables to $\mu$-types; $\Gamma(x) = A$ means that $\Gamma$ maps $x$ to $A$.  We have two typing judgements, one for patterns $\sequP{\theta}{p : A}$ and one for terms $\sequT{\Gamma}{s : A}$.
Accordingly, we have two sets of typing rules: Fig.~\ref{fig:typingSchemesForPatternsAndTerms}, top and bottom. We write $\sequPDeriv{\theta}{p:A}$ to indicate that the typing judgement $\sequP{\theta}{p:A}$ is derivable (likewise for $\sequTDeriv{\Gamma}{s:A}$).
The typing schemes speak for themselves except for two of them which we now comment. The first is $\ruleTApp$.
Note that we do not require the $A_i$ to be non-union types. This allows examples such as (\ref{eq:terms:ejemploReduccion}) to be typable (the outermost instance of $\ruleTApp$ is with $n=1$ and $A_1=\bool=\uniontype{\consttype{true}}{\consttype{false}}$).
Regarding $\ruleTAbs$ it requests a number of conditions. First of all, each of the patterns $p_i$ must be typable under the typing context $\theta_i$, $i \in 1..n$. Also, the set of free matchables in each $p_i$ must be exactly the domain of $\theta_i$. Another condition, indicated by $(\sequT{\Gamma, \theta_i}{s_i : B})_{i\in 1..n}$, is that the bodies of each of the branches $s_i$, $i\in 1..n$, be typable under the context extended with the corresponding $\theta_i$. More noteworthy is the condition that the list $\lista{\sequC{\theta_i}{p_i : A_i}}_{i \in 1..n}$ be \emph{compatible}, which we now discuss in further detail.

\subsection{Compatibility}
\label{sec:compatibility}

Let us say that a pattern $p$ \emphdef{subsumes} a pattern $q$, written \matches{p}{q} if there exists a substitution $\sigma$
s.t. $\sigma p = q$.
Consider an abstraction $(\absterm{p}{s}{\theta} \icaseterm \absterm{q}{t}{\theta'})$ and two judgements $\sequP{\theta}{p : A}$ and $\sequP{\theta'}{q : B}$. We consider two cases depending on whether $p$ subsumes $q$ or not.

As already mentioned in example (\ref{eq:example:compatib:i}) of the introduction, if $p$ subsumes $q$, then the branch \absterm{q}{t}{\theta'} will never be evaluated since the argument will already match $p$. Indeed, for any term $u$ of type $B$ in matchable form, the application will reduce to $\match{p}{u}\,s$. Thus, in this case, in order to ensure SR we demand that $B \subtypemu A$.

Suppose $p$ does not subsume $q$ (\ie $\nmatches{p}{q}$). We analyze the cause of failure of subsumption in order to determine whether requirements on $A$ and $B$ must be put forward. In some cases no requirements are necessary. For example in: 
\begin{equation}
\matchable{f} \ifuncterm_{\set{f:\functype{A}{B}}}
\begin{array}[t]{cllll}
(          & \dataterm{\constterm{c}}{\matchable{z}} & \ifuncterm_{\set{z:A}} & {\dataterm{\constterm{c}}{(\appterm{f}{z})}} \\
\icaseterm & \dataterm{\constterm{d}}{\matchable{y}} & \ifuncterm_{\set{y:B}} & {\dataterm{\constterm{d}}{y}})
\end{array} 
\label{eq:example:compat:noOverlap}
\end{equation}
\noindent no relation between $A$ and $B$ is required since the branches are mutually disjoint. In
other cases, however, $A\subtypemu B$ is required; we seek to characterize them.
We focus on those cases where $p$ fails to subsume $q$, and $\pi\in\pos{p}\cap\pos{q}$ is an
offending position in both patterns. The following table exhaustively lists them:
\begin{center}
\begin{tabular}{l|c|c|l}
    & $\at{p}{\pi}$                         & $\at{q}{\pi}$        & \\
\hline
(a) & \multirow{3}{*}{\constterm{c}}        & \matchable{y}        & restriction required \\
(b) &                                       & \constterm{d}        & no overlapping ($\nmatches{q}{p}$) \\
(c) &                                       & $\appterm{q_1}{q_2}$ & no overlapping \\
\hline
(d) & \multirow{3}{*}{$\appterm{p_1}{p_2}$} & \matchable{y}        & restriction required \\
(e) &                                       & \constterm{d}        & no overlapping
\end{tabular}
\end{center}
\noindent In cases (b), (c) and (e), no extra condition on the types of $p$ and $q$ is necessary either,
since their respective sets of possible arguments are disjoint; example~(\ref{eq:example:compat:noOverlap})
corresponds to the first of these. The cases where $A$ and $B$ must be related are (a) and (d): for
those we require $B\subtypemu A$.
The first of these has already been illustrated in the introduction~(\ref{eq:example:compatib:i}), the second one is
illustrated as follows:
\begin{equation}
\matchable{f} \ifuncterm_{\set{f:\functype{D}{\functype{A}{C}}}} \matchable{g} \ifuncterm_{\set{g:\functype{B}{C}}}
\begin{array}[t]{clll}
(          & \dataterm{\matchable{x}}{\matchable{y}} & \ifuncterm_{\set{x:D,y:A}} & \appterm{\appterm{f}{x}}{y} \\
\icaseterm & \matchable{z}                           & \ifuncterm_{\set{z:B}}     & \appterm{g}{z})
\end{array}
\end{equation}
The problematic situation is when $B = \datatype{D'}{B'}$, \ie the type of $z$ is another
compound, which may have no relation at all with $\datatype{D}{A}$. Compatibility ensures $B\subtypemu \datatype{D}{A}$.

We now formalize these ideas. 

\begin{definition}
\label{def:constructorAtPosition}
Given a pattern $\sequP{\theta}{p:A}$ and $\pi\in\pos{p}$, we say $A$ \emph{admits a symbol
  $\odot$ (with $\odot\in \TypeVariable \cup \TypeConstant \cup \set{\ifunctype,\idatatype}$)}
\emph{at position} $\pi$ iff $\odot \in \lookup{A}{\pi}$, where: $$
\begin{array}{c}
\begin{array}{rcll}
\lookup{a}{\epsilon}                 & \eqdef & \set{a} \\
\lookup{(A_1 \star A_2)}{\epsilon}   & \eqdef & \set{\star},       & \star\in \set{\ifunctype, \idatatype} \\
\lookup{(A_1 \star A_2)}{i\pi}       & \eqdef & \lookup{A_i}{\pi}, & \star\in \set{\ifunctype, \idatatype}, i\in \set{1,2} \\
\lookup{(\uniontype{A_1}{A_2})}{\pi} & \eqdef & \lookup{A_1}{\pi} \cup \lookup{A_2}{\pi} \\
\lookup{(\rectype{V}{A'})}{\pi}      & \eqdef & \lookup{(\substitute{V}{\rectype{V}{A'}}{A'})}{\pi}
\end{array}
\end{array} $$
\end{definition}
Note that $\sequPDeriv{\theta}{p:A}$ and contractiveness of $A$, implies $\lookup{A}{\pi}$ is well-defined for $\pi\in\pos{p}$.

Whenever subsumption between two patterns fails, any mismatching position is a leaf in the
syntactic tree of one of the patterns. Otherwise, both of them would have a type application constructor in that
position and there would be no failure of subsumption.  

\begin{definition}
The \emph{maximal positions} in a set of positions $P$ are: $$\mpos{P} \eqdef
\set{\pi \in P \mathrel| \nexists \pi'\in P. \pi'=  \pi\pi''\land \pi''\neq\epsilon}$$

The \emph{mismatching positions} between two patterns are: $$\cpos{p}{q} \eqdef \set{\pi \mathrel| \pi \in \mpos{\pos{p} \cap \pos{q}} \land \nmatches{\at{p}{\pi}}{\at{q}{\pi}}}$$
\end{definition}

\begin{definition}
\label{def:compatibility}
We say $\sequC{\theta}{p : A}$ is \emph{compatible} with $\sequC{\theta'}{q :
  B}$, written $\compatible{\sequC{\theta}{p : A}}{\sequC{\theta'}{q :
  B}}$, iff the following two conditions hold:
\begin{enumerate}
 \item $\matches{p}{q} \implies B \subtypemu A$.
 \item $\nmatches{p}{q} \implies \left(\forall \pi \in \cpos{p}{q}.
\lookup{A}{\pi} \cap \lookup{B}{\pi} \neq \varnothing
\right) \implies B \subtype A$.
\end{enumerate}

A list of patterns $\lista{\sequC{\theta_i}{p_i : A_i}}_{i \in 1..n}$ is compatible if $\forall i, j \in 1..n. i < j \implies \compatible{\sequC{\theta_i}{p_i : A_i}}{\sequC{\theta_j}{p_j : A_j}}$.
\end{definition}

As a further example, suppose we wish to apply $\mathsf{upd}$ (\cf
(\ref{eq:intro:upd})) to data structures holding values of different types:
say \constterm{vl} prefixed values are numbers and $\constterm{vl2}$ prefixed
values are functions over numbers. Note that $\mathsf{upd}$ cannot be typed as
it stands. The reason is that the last branch would have to handle values of
functional type and hence would receive type
$\uniontype{\consttype{cons}}{\uniontype{\consttype{node}}{\uniontype{\consttype{nil}}{\uniontype{\consttype{vl2}}{(\functype{\nat}{\nat})}}}}$.
This fails to be a datatype due to the presence of the component of functional
type. As a consequence, $x\, y$ cannot be typed since it requires an
applicative type $\idatatype$. The remedy is to add an additional branch to
$\mathsf{upd}$ capable of handling values prefixed by $\constterm{vl2}$:
\begin{equation}
\mathsf{upd'} = \absterm{\matchable{f}}{\absterm{\matchable{g}}{
\begin{array}[t]{clll}
(          & \dataterm{\constterm{vl}}{\matchable{z}}  & \ifuncterm_{\set{z:A_1}}                 & \dataterm{\constterm{vl}}{(\appterm{f}{z})} \\
\icaseterm & \dataterm{\constterm{vl2}}{\matchable{z}} & \ifuncterm_{\set{z:\functype{A_2}{A_3}}} & \dataterm{\constterm{vl2}}{({g}\appterm{z})} \\
\icaseterm & \dataterm{\matchable{x}}{\matchable{y}}   & \ifuncterm_{\set{x:C, y:D}}              & \dataterm{(\appterm{\appterm{\mathsf{upd'}}{f}}{x})}{(\appterm{\appterm{\mathsf{upd'}}{f}}{y})} \\
\icaseterm & {\matchable{w}} & \ifuncterm_{\set{w:E}} & {w})
\end{array}
}{\set{g:\functype{(\functype{A_2}{A_3})}{B}}}
}{\set{f:\functype{A_1}{B}}}
\label{eq:intro:ejemploDelia:extendido}
\end{equation}
The type of $\mathsf{upd'}$ is 
$\functype{\functype{{(\functype{A_1}{B})}}{(\functype{(\functype{A_2}{A_3})}{B})}}{(\functype{F_{A_1,\functype{A_2}{A_3}}}{F_{B,B}})}$,
where $F_{X,Y}$ is
\begin{center}
$\rectype{\alpha}{\uniontype{\uniontype{\uniontype{(\datatype{\consttype{vl}}{X})}{(\datatype{\consttype{vl2}}{Y})}}{(\datatype{\alpha}{\alpha})}}{(\uniontype{\consttype{cons}}{\uniontype{\consttype{node}}{\consttype{nil}}})}}$
\end{center}
This is quite natural: the type system establishes a clear distinction between
semi-structured data, susceptible to path polymorphism, and ``unstructured''
data represented here by base and functional types.

\subsection{Basic Metatheory of Typing}
\label{sec:typingMetatheory}

We present some technical lemmas that will be useful in the proof of safety and
type-checking.

The following four lemmas are straightforward adaptations of the standard
Generation Lemma and Basis Lemma to our system, considering patterns and terms
separately.

\begin{lemma}[Generation Lemma for Patterns]
\label{lem:generationForPatterns}
Let $\theta$ be a typing context and $A$ a type.
\begin{enumerate}[(i)]
  \item If $\sequPDeriv{\theta}{\matchable{x} : A}$ then $x : A \in \theta$.
  \item If $\sequPDeriv{\theta}{\constterm{c} : A}$ then $A \eqtypemu \consttype{c}$.
  \item If $\sequPDeriv{\theta}{\dataterm{p}{q} : A}$ then $\exists D, A'$ such that $A \eqtypemu \datatype{D}{A'}$, $\sequPDeriv{\theta}{p : D}$ and $\sequPDeriv{\theta}{q : A'}$.
\end{enumerate}
\end{lemma}

\begin{proof}
By simple analysis of the applicable rules for each term constructor. Note that
here there's only one applicable rule in each case.
\end{proof}

\begin{lemma}[Basis Lemma for Patterns]
\label{lem:basisForPatterns}
Let $\theta$ be a typing context, $p$ a pattern and $A$ a type such that $\sequPDeriv{\theta}{p : A}$.
\begin{enumerate}[(i)]
  \item Let $\Delta \supseteq \theta$ be another typing context, then $\sequPDeriv{\theta}{p : A}$.
  \item $\fm{p} \subseteq \dom{\theta}$.
  \item $\sequPDeriv{\res{\theta}{\fm{p}}}{p : A}$.
\end{enumerate}
\end{lemma}

\begin{proof}
The three cases are by induction on $p$ using the Generation Lemma for Patterns.
\end{proof}

\begin{lemma}[Generation Lemma]
\label{lem:generation}
Let $\Gamma$ be a typing context and $A$ a type.
\begin{enumerate}[(i)]
  \item If $\sequTDeriv{\Gamma}{x : A}$ then $\exists A'$ s.t. $A' \subtypemu A$ and $x : A' \in \Gamma$.
  \item If $\sequTDeriv{\Gamma}{\constterm{c} : A}$ then $\consttype{c} \subtypemu A$.
  \item If $\sequTDeriv{\Gamma}{\appterm{r}{u} : A}$ then:
  \begin{enumerate}[(a)]
    \item either $\exists D, A'$ s.t. $\datatype{D}{A'} \subtypemu A$, $\sequTDeriv{\Gamma}{r : D}$ and $\sequTDeriv{\Gamma}{u : A'}$;
    \item or $\exists A_1, \ldots, A_n, A',k\in 1..n$ s.t. $A' \subtypemu A$, $\sequTDeriv{\Gamma}{r : \functype{\maxuniontype{i \in 1..n}{A_i}}{A'}}$, and $\sequTDeriv{\Gamma}{u : A_k}$.
  \end{enumerate}
  \item If $\sequTDeriv{\Gamma}{(\absterm{p_i}{s_i}{\theta_i})_{i \in 1..n} : A}$ then $\exists A_1,
    \ldots, A_n, B$ s.t. $\functype{\maxuniontype{i \in 1..n}{A_i}}{B} \subtypemu A$, $\lista{\sequC{\theta_i}{p_i : A_i}}_{i \in 1..n}$ is compatible, $\dom{\theta_i} = \fm{p_i}$, $\sequPDeriv{\theta_i}{p_i : A_i}$ and $\sequTDeriv{\Gamma, \theta_i}{s_i : B}$ for every $i \in 1..n$.
\end{enumerate}
\end{lemma}

\begin{proof}
By induction on the derivation of $\sequT{\Gamma}{s : A}$ analyzing the last rule applied.
\begin{itemize}
  \item $\ruleTVar$: then $s = x$ with $x : A' \in \Gamma$. We take $A = A'$ and (i) holds by reflexivity of subtyping.

  \item $\ruleTConst$: then $s = \constterm{c}$ and $A = \consttype{c}$. Again by reflexivity we conclude that (ii) holds.

  \item $\ruleTComp$: then $s = \dataterm{r}{u}$ and $A = \datatype{D}{A'}$ with $\sequTDeriv{\Gamma}{r : D}$ and $\sequTDeriv{\Gamma}{u : A'}$. By reflexivity of subtyping we get $\datatype{D}{A'} \subtypemu A$ and conclude that (iii.a) holds.

  \item $\ruleTAbs$: then $s = (\absterm{p_i}{s_i}{\theta_i})_{i \in 1..n}$ and $A = \functype{\maxuniontype{i \in 1..n}{A_i}}{B}$ with $\dom{\theta_i} = \fm{p_i}$,  $\lista{\sequC{\theta_i}{p_i : A_i}}_{i \in 1..n}$ compatible, $\sequPDeriv{\theta_i}{p_i : A_i}$ and $\sequTDeriv{\Gamma, \theta_i}{s_i : B}$ for every $i \in 1..n$. Here (iv) holds by reflexivity of subtyping.

  \item $\ruleTApp$: then $s = \appterm{r}{u}$ with $\sequTDeriv{\Gamma}{r : \functype{\maxuniontype{i \in 1..n}{A_i}}{A}}$ and $\sequTDeriv{\Gamma}{u : A_k}$ for some $k \in 1..n$. We conclude reflexivity with $A' = A$ that (iii.b) holds.

  \item $\ruleTSubs$: then $\sequTDeriv{\Gamma}{s : A''}$ with $A'' \subtypemu A$. Now we analyze the form of the term $s$ to see which of the cases of the lemma holds for each term constructor:
  \begin{enumerate}[(i)]
    \item $s = x$. By inductive hypothesis $\exists A'$ such that $A' \subtypemu A''$ and $x : A' \in \Gamma$. Then, by transitivity of subtyping, $A' \subtypemu A''$ and we conclude that (i) holds.
    
    \item $s = \constterm{c}$. By inductive hypothesis $\consttype{c} \subtypemu A''$ and by transitivity of subtyping $\consttype{c} \subtypemu A$. Hence (ii) holds.
    
    \item $s = \appterm{r}{u}$. By inductive hypothesis we have two options:
    \begin{enumerate}[(a)]
      \item either $\exists D, A'$ such that $\datatype{D}{A'} \subtypemu A''$, $\sequTDeriv{\Gamma}{r : D}$ and $\sequTDeriv{\Gamma}{u : A'}$. By transitivity we have $\datatype{D}{A'} \subtypemu A$ and we are in the case that (iii.a) holds.
      
      \item or $\exists A_1, \ldots, A_n, A'$ such that $A \subtypemu A''$, $\sequTDeriv{\Gamma}{r : \functype{\maxuniontype{i \in 1..n}{A_i}}{A'}}$, and $\sequTDeriv{\Gamma}{u : A_k}$ for some $k \in 1..n$. Again by transitivity $A' \subtypemu A$ and (iii.b) holds.
    \end{enumerate}
    
    \item $s = (\absterm{p_i}{s_i}{\theta_i})_{i \in 1..n}$. By inductive hypothesis $\exists A_1, \ldots, A_n, B$ such that $\functype{\maxuniontype{i \in 1..n}{A_i}}{B} \subtypemu A''$, $\lista{\sequC{\theta_i}{p_i : A_i}}_{i \in 1..n}$ is compatible, $\dom{\theta_i} = \fm{p_i}$, $\sequPDeriv{\theta_i}{p_i : A_i}$ and $\sequTDeriv{\Gamma, \theta_i}{s_i : B}$ for every $i \in 1..n$. Then we conclude by transitivity of subtyping that $\functype{\maxuniontype{i \in 1..n}{A_i}}{B} \subtypemu A$ and (iv) holds.
  \end{enumerate}
\end{itemize}
\end{proof}

\begin{lemma}[Basis Lemma]
\label{lem:basis}
Let $\Gamma$ be a typing context, $s$ a term and $A$ a type such that $\sequTDeriv{\Gamma}{s : A}$.
\begin{enumerate}[(i)]
  \item Let $\Delta \supseteq \Gamma$ be another typing context, then $\sequTDeriv{\Delta}{s : A}$.
  \item $\fv{s} \subseteq \dom{\Gamma}$.
  \item $\sequTDeriv{\res{\Gamma}{\fv{s}}}{s : A}$.
\end{enumerate}
\end{lemma}

\begin{proof}
The three cases are by induction on $s$ using the Generation Lemma.
\begin{enumerate}[(i)]
  \item $\sequTDeriv{\Delta}{s : A}$.
  \begin{itemize}
    \item $s = x$. By Lem.~\ref{lem:generation} (i) $\exists A'$ such that $A' \subtypemu A$ and $x : A' \in \Gamma$. Then $\Delta = \Delta', x : A'$ and by $\ruleTVar$ and $\ruleTSubs$ we get $\sequTDeriv{\Delta}{x : A}$.

    \item $s = \constterm{c}$. By Lem.~\ref{lem:generation} (ii) $\consttype{c} \subtypemu A$ and we conclude by $\ruleTConst$ and $\ruleTSubs$ $\sequTDeriv{\Delta}{\constterm{c} : A}$.

    \item $s = \appterm{r}{u}$. By Lem.~\ref{lem:generation} (iii) we have two possible cases:
    \begin{enumerate}[(a)]
      \item either $\exists D, A'$ such that $\datatype{D}{A'} \subtypemu A$, $\sequTDeriv{\Gamma}{r : D}$ and $\sequTDeriv{\Gamma}{u : A'}$. By inductive hypothesis $\sequTDeriv{\Delta}{r : D}$ and $\sequTDeriv{\Delta}{u : A'}$. Then by $\ruleTComp$ and $\ruleTSubs$ $\sequTDeriv{\Delta}{\dataterm{r}{u} : A}$.
      
      \item or $\exists A_1, \ldots, A_n, A'$ such that $A' \subtypemu A$, $\sequTDeriv{\Gamma}{r : \functype{\maxuniontype{i \in 1..n}{A_i}}{A'}}$, and $\sequTDeriv{\Gamma}{u : A_k}$ for some $k \in 1..n$. Applying the inductive hypothesis we get $\sequTDeriv{\Delta}{r : \functype{\maxuniontype{i \in 1..n}{A_i}}{A'}}$ and $\sequTDeriv{\Delta}{u : A_k}$, so we conclude by $\ruleTApp$ and $\ruleTSubs$ that $\sequTDeriv{\Delta}{\appterm{r}{u} : A}$.
    \end{enumerate}

    \item $s = (\absterm{p_i}{s_i}{\theta_i})_{i \in 1..n}$. By Lem.~\ref{lem:generation} (iv) $\exists A_1, \ldots, A_n, B$ such that $\functype{\maxuniontype{i \in 1..n}{A_i}}{B} \subtypemu A$, $\lista{\sequC{\theta_i}{p_i : A_i}}_{i \in 1..n}$ is compatible, $\dom{\theta_i} = \fm{p_i}$, $\sequPDeriv{\theta_i}{p_i : A_i}$ and $\sequTDeriv{\Gamma, \theta_i}{s_i : B}$ for every $i \in 1..n$. Without loss of generality we can assume $\dom{\Delta} \cap \dom{\theta_i} = \varnothing$ for all $i \in 1..n$. Then $\Delta, \theta_i$ is also a typing context and by inductive hypothesis $\sequTDeriv{\Delta, \theta_i}{s_i : B}$ for all $i \in 1..n$. Then by $\ruleTAbs$ and $\ruleTSubs$ we conclude $\sequTDeriv{\Delta}{(\absterm{p_i}{s_i}{\theta_i})_{i \in 1..n} : A}$.
  \end{itemize}

  \item $\fv{s} \subseteq \dom{\Gamma}$.
  \begin{itemize}
    \item $s = x$. By Lem.~\ref{lem:generation} (i) $\exists A'$ such that $A' \subtypemu A$ and $x : A' \in \Gamma$. Then $\fv{s} = \set{x} \subseteq \dom{\Gamma}$.

    \item $s = \constterm{c}$. Then $\fv{s} = \varnothing \subseteq \dom{\Gamma}$.

    \item $s = \appterm{r}{u}$. By Lem.~\ref{lem:generation} (iii) $\exists B, B'$ such that $\sequTDeriv{\Gamma}{r : B}$ and $\sequTDeriv{\Gamma}{u : B'}$. By inductive hypothesis $\fv{r} \subseteq \dom{\Gamma}$ and $\fv{u} \subseteq \dom{\Gamma}$. Then $\fv{s} = \fv{r}\cup\fv{u} \subseteq \dom{\Gamma}$.
		
    \item $s = (\absterm{p_i}{s_i}{\theta_i})_{i \in 1..n}$. By Lem.~\ref{lem:generation} (iv) $\exists A_1, \ldots, A_n, B$ such that $\functype{\maxuniontype{i \in 1..n}{A_i}}{B} \subtypemu A$, $\dom{\theta_i} = \fm{p_i}$ and $\sequTDeriv{\Gamma, \theta_i}{s_i : B}$ for every $i \in 1..n$. By inductive hypothesis $\fv{s_i} \subseteq \dom{\Gamma, \theta_i} = \dom{\Gamma} \uplus \fm{p_i}$ and we have $\fv{s_i} \setminus \fm{p_i} \subseteq \dom{\Gamma}$ for every $i \in 1..n$. Then $\fv{s} = \bigcup_{i \in 1..n}{\fv{s_i} \setminus \fm{p_i}} \subseteq \dom{\Gamma}$.
  \end{itemize}

  \item $\sequTDeriv{\res{\Gamma}{\fv{s}}}{s : A}$.
  \begin{itemize}
    \item $s = x$. By Lem.~\ref{lem:generation} (i) $\exists A'$ such that $A' \subtypemu A$ and $x : A' \in \Gamma$. Then by $\ruleTVar$ and $\ruleTSubs$ $\sequTDeriv{\res{\Gamma}{\fv{s}} = x : A'}{x : A}$.

    \item $s = \constterm{c}$. By Lem.~\ref{lem:generation} (ii) $\consttype{c} \subtypemu A$ and we conclude by $\ruleTConst$ and $\ruleTSubs$ $\sequTDeriv{\res{\Gamma}{\fv{s}}}{\constterm{c} : A}$.

    \item $s = \appterm{r}{u}$. By Lem.~\ref{lem:generation} (iii) we have two possible cases:
    \begin{enumerate}[(a)]
      \item either $\exists D, A'$ such that $\datatype{D}{A'} \subtypemu A$, $\sequTDeriv{\Gamma}{r : D}$ and $\sequTDeriv{\Gamma}{u : A'}$. By inductive hypothesis $\sequTDeriv{\res{\Gamma}{\fv{r}}}{r : D}$ and $\sequTDeriv{\res{\Gamma}{\fv{u}}}{u : A'}$. Since $\Gamma$ is a typing context, $\res{\Gamma}{\fv{r}} \subseteq \Gamma$ and $\res{\Gamma}{\fv{u}} \subseteq \Gamma$, then $\res{\Gamma}{\fv{r}} \cup \res{\Gamma}{\fv{u}} = \res{\Gamma}{\fv{\dataterm{r}{u}}}$ is also a typing context. Now, by Lem.~\ref{lem:basis} (i), $\sequTDeriv{\res{\Gamma}{\fv{s}}}{r : D}$ and $\sequTDeriv{\res{\Gamma}{\fv{s}}}{u : A'}$. Then we conclude by applying $\ruleTComp$ and $\ruleTSubs$.
      
      \item or $\exists A_1, \ldots, A_n, A'$ such that $A' \subtypemu A$, $\sequTDeriv{\Gamma}{r : \functype{\maxuniontype{i \in 1..n}{A_i}}{A'}}$, and $\sequTDeriv{\Gamma}{u : A_k}$ for some $k \in 1..n$. By inductive hypothesis $\sequTDeriv{\res{\Gamma}{\fv{u}}}{r : \functype{\maxuniontype{i \in 1..n}{A_i}}{A'}}$ and $\sequTDeriv{\res{\Gamma}{\fv{u}}}{u : A_k}$. Again we have $\res{\Gamma}{\fv{r}} \cup \res{\Gamma}{\fv{u}} = \res{\Gamma}{\fv{\appterm{r}{u}}}$ a typing context  and we can apply case (i) of this same lemma to get $\sequTDeriv{\res{\Gamma}{\fv{s}}}{r : \functype{\maxuniontype{i \in 1..n}{A_i}}{A'}}$ and $\sequTDeriv{\res{\Gamma}{\fv{s}}}{u : A_k}$. Finally we conclude by $\ruleTApp$ and $\ruleTSubs$.
    \end{enumerate}

    \item $s = (\absterm{p_i}{s_i}{\theta_i})_{i \in 1..n}$. By Lem.~\ref{lem:generation} (iv) $\exists A_1, \ldots, A_n, B$ such that $\functype{\maxuniontype{i \in 1..n}{A_i}}{B} \subtypemu A$, $\lista{\sequC{\theta_i}{p_i : A_i}}_{i \in 1..n}$ is compatible, $\dom{\theta_i} = \fm{p_i}$, $\sequPDeriv{\theta_i}{p_i : A_i}$ and $\sequTDeriv{\Gamma, \theta_i}{s_i : B}$ for every $i \in 1..n$. By inductive hypothesis $\sequTDeriv{\res{(\Gamma, \theta_i)}{\fv{s_i}}}{s_i : B}$ and it's immediate to see that $\res{(\Gamma, \theta_i)}{\fv{s_i}} \subseteq \res{\Gamma}{\fv{s_i}} \uplus \theta_i$, since $\dom{\Gamma} \cap \dom{\theta_i} = \varnothing$. Moreover, as $\dom{\theta_i} = \fm{p_i}$, we have $\res{\Gamma}{\fv{s_i}} = \res{\Gamma}{(\fv{s_i} \setminus \fm{p_i})} \subseteq \res{\Gamma}{\bigcup_{j \in 1..n}{(\fv{s_j} \setminus \fm{p_j})}}$. Then $\res{\Gamma}{\fv{s}} \uplus \theta_i \supseteq \res{\Gamma}{\fv{s_i}} \uplus \theta_i \supseteq \res{(\Gamma, \theta_i)}{\fv{s_i}}$ is also a typing context and, by Lem.~\ref{lem:basis} (i), we get $\sequTDeriv{\res{\Gamma}{\fv{s}}, \theta_i}{s_i : B}$ for every $i \in 1..n$. Finally we apply $\ruleTAbs$ and $\ruleTSubs$ to conclude $\sequTDeriv{\res{\Gamma}{\fv{s}}}{(\absterm{p_i}{s_i}{\theta_i})_{i \in 1..n} : A}$.
  \end{itemize}
\end{enumerate}
\end{proof}

The following lemma is useful to deduce the shape of the type when we know the
term is a data structure. Essentially it states that every data structure that
can be given a type, can also be typed with a more specific non-union datatype.

\begin{lemma}[Typing for Data Structures]
\label{lem:typingForDataStructures}
Suppose  $\sequTDeriv{\Gamma}{d : A}$, for $d$ a data structure.  Then $\exists D$ datatype such that $D$ is a non-union type, $D \subtypemu A$ and $\sequTDeriv{\Gamma}{d : D}$. Moreover,
\begin{enumerate}
\item If $d = \constterm{c}$, then $D \eqtypemu \consttype{c}$.
\item If $d = \dataterm{d'}{t}$, then $\exists D', A'$ such that $D \eqtypemu \datatype{D'}{A'}$, $\sequTDeriv{\Gamma}{d' : D'}$ and $\sequTDeriv{\Gamma}{t : A'}$.
\end{enumerate}
\end{lemma}

\begin{proof}
By induction on $d$.
\begin{itemize}
  \item $d = \constterm{c}$. By Lem.~\ref{lem:generation} (ii) $D = \consttype{c} \subtypemu A$.

  \item $d = \dataterm{d'}{t}$. By Lem.~\ref{lem:generation} (iii) there are two possible cases:
  \begin{enumerate}[(a)]
    \item either $\exists D', A'$ such that $\datatype{D'}{A'} \subtypemu A$, $\sequTDeriv{\Gamma}{d' : D'}$ and $\sequTDeriv{\Gamma}{t : A'}$. Then the property holds with $D = \datatype{D'}{A'}$, since by $\ruleTComp$ we can derive $\sequTDeriv{\Gamma}{\dataterm{d'}{t} : D}$.
    
    \item or $\exists A_1, \ldots, A_n, A'$ such that $A' \subtypemu A$, $\sequTDeriv{\Gamma}{d' : \functype{\maxuniontype{i \in 1..n}{A_i}}{A'}}$, and $\sequTDeriv{\Gamma}{t : A_k}$ for some $k \in 1..n$. By inductive hypothesis applied to $d'$ we get that $\exists D'$ datatype such that $D'$ is not a union type and $D' \subtypemu \functype{\maxuniontype{i \in 1..n}{A_i}}{A'}$. But, by Lem.~\ref{lem:supertypesOfNonUnionTypes}, both of them have the same outermost type constructor, which leads to a contradiction. Hence this case does not apply.
  \end{enumerate}
\end{itemize}
\end{proof}

Some results on compatibility follow, the crucial one being Lem.~\ref{lem:compatibility}. 
This next lemma shows that maching failure is enough to guarantee that the type of the
argument is not a subtype of that of the pattern.

\begin{lemma}
\label{lem:mismatchingTypes}
Given $\sequTDeriv{\Gamma}{u : B}$, $\sequPDeriv{\theta}{p : A}$. If $\match{p}{u} = \fail$, then $B \not\subtypemu A$.
\end{lemma}

\begin{proof}
By induction on $p$. We only analyse the cases where $\match{p}{u} = \fail$, otherwise the implication holds trivially.
\begin{itemize}
  \item $p = \constterm{c}$: then $u$ is a matchable form and $u \neq \constterm{c}$. By Lem.~\ref{lem:generationForPatterns} (ii), $A = \consttype{c}$.
  \begin{enumerate}
    \item $u = \constterm{d} \neq \constterm{c}$: by Lem.~\ref{lem:generation} (ii), $\consttype{d} \subtypemu B$. Then, if $B \subtypemu \consttype{c}$ we would have $\consttype{d} \subtypemu \consttype{c}$ by transitivity, which is clearly not possible by invertibility of subtyping for non-union types. Hence, it cannot be the case that $B \subtypemu A$.
    \item $u = \dataterm{u_1}{u_2}$: by Lem.~\ref{lem:typingForDataStructures}, $\exists D',B'$ such that $\datatype{D'}{B'} \subtypemu B$. Again, if $B \subtypemu \consttype{c}$ we have a contradiction.
    \item $u = (\absterm{q_j}{u_j}{})_{j \in 1..m}$: by Lem.~\ref{lem:generation} (iv), there exists $B_1, \ldots, B_m, B'$ such that $\functype{\maxuniontype{j \in 1..m}{B_j}}{B'} \subtypemu B$. Thus, we conclude by contradiction as in the previous case.
  \end{enumerate}
  
  \item $p = \dataterm{p_1}{p_2}$: here, by Lem.~\ref{lem:generationForPatterns} (iii), $\exists D, A'$ such that $A = \datatype{D}{A'}$ with $\sequPDeriv{\theta}{p_1 : D}$ and $\sequPDeriv{\theta}{p_2 : A'}$. There are three possible cases of mismatch:
  \begin{enumerate}
    \item $u = \constterm{d} \neq \constterm{c}$: similarly to the previous cases, by Lem.~\ref{lem:generation} (ii) we have $\consttype{d} \subtypemu B$ which leads to a contradiction if $B \subtypemu \datatype{D}{A'}$.
    \item $u = \dataterm{u_1}{u_2}$: then the mismatch was internal. Thus, we have $\match{p_i}{u_i} = \fail$ for at least one of the two possibilities. By Lem.~\ref{lem:typingForDataStructures}, $\exists D',B'$ such that $\datatype{D'}{B'} \subtypemu B$ with $\sequTDeriv{\Gamma}{u_1 : D'}$ and $\sequTDeriv{\Gamma}{u_2 : B'}$. Then, by inductive hypothesis, we have $D \not\subtypemu D'$, or $A' \not\subtypemu B'$, or both.

    Now suppose $B \subtypemu A \eqtypemu \datatype{D'}{A'}$. By transitivity we have $\datatype{D'}{B'} \subtypemu \datatype{D'}{A'}$ and by invertibility of subtyping for non-union types both $D \subtypemu D'$ and $A' \subtypemu B'$ should hold. Thus, we conclude $B \not\subtypemu A$.
    \item $u = (\absterm{q_j}{u_j}{})_{j \in 1..m}$: as before, by Lem.~\ref{lem:generation} (iv), we have $B_1, \ldots, B_m, B'$ such that $\functype{\maxuniontype{j \in 1..m}{B_j}}{B'} \subtypemu B$ and conclude by contradiction with $B \subtypemu \datatype{D}{A'}$.
  \end{enumerate}
\end{itemize}
\end{proof}

Define $\Psicomp{\sequC{\theta}{p : A}}{\sequC{\theta'}{q : B}} \eqdef \forall \pi \in \cpos{p}{q}.\lookup{A}{\pi} \cap \lookup{B}{\pi} \neq \varnothing$, so that compatibility can alternatively be characterized as: $$\compatible{\sequC{\theta}{p : A}}{\sequC{\theta'}{q : B}} \quad\text{iff}\quad \Psicomp{\sequC{\theta}{p : A}}{\sequC{\theta'}{q : B}} \implies B \subtypemu A$$

The Compatibility Lemma should be interpreted in the context of an abstraction.
Assume an argument $u$ of type $B$ is passed to a function where there are (at
least) two branches, defined by patterns $p$ and $q$, the latter having the
same type as $u$. If the argument matches the first pattern of (potentially) a
different type $A$, then $\Psicomp{\sequC{\theta}{p : A}}{\sequC{\theta'}{q : B}}$
must hold. Since patterns within an abstraction must be compatible, we get
$B \subtypemu A$ and thus $\sequTDeriv{\Gamma}{u : A}$ too.

\begin{lemma}[Compatibility Lemma]
\label{lem:compatibility}
Suppose $\sequTDeriv{\Gamma}{u : B}$, $\sequPDeriv{\theta}{p : A}$, $\sequPDeriv{\theta'}{q : B}$ and $\match{p}{u}$ is successful. Then, $\Psicomp{\sequC{\theta}{p : A}}{\sequC{\theta'}{q : B}}$ holds.
\end{lemma}

\begin{proof}
By induction on $p$.
\begin{itemize}
  \item $p = \matchable{x}$: then the result is immediate since $\matches{\matchable{x}}{q}$ for every pattern $q$.
  
  \item $p = \constterm{c}$: if $\matches{\constterm{c}}{q}$ the result is immediate. So lets analize the case where $\nmatches{\constterm{c}}{q}$ (\ie $q \neq \constterm{c}$). We have $u = \constterm{c}$ by matching success and $\consttype{c} \subtypemu B$ by Lem.~\ref{lem:generation} (ii). Assume $B = \maxuniontype{i \in 1..n}{B_i}$ with $B_i \neq \iuniontype$, then $\consttype{c} \subtypemu B_j$ for some $j \in 1..n$. Moreover, by invertibility of subtyping of non-union types, $B_j = \consttype{c}$. On the other hand, by Lem.~\ref{lem:generationForPatterns} (ii), $A = \consttype{c}$. Then, $\lookup{A}{\epsilon} \cap \lookup{B}{\epsilon} \neq \varnothing$ and we conclude since $\cpos{p}{q} = \set{\epsilon}$.
  
  \item $p = \dataterm{p_1}{p_2}$: again, lets see the cases where $\nmatches{p}{q}$. By matching success we have $u = \dataterm{u_1}{u_2}$ a data structure with $\match{p}{u} = \match{p_1}{u_1} \uplus \match{p_2}{u_2}$ both successful. Moreover, by Lem.~\ref{lem:typingForDataStructures}, $\exists D',B'$ such that $\datatype{D'}{B'} \subtypemu B$ with $\sequTDeriv{\Gamma}{u_1 : D'}$ and $\sequTDeriv{\Gamma}{u_2 : B'}$. Now we analize the shape of $q$:
  \begin{enumerate}
    \item $q = \matchable{y}$: as before, assume $B = \maxuniontype{i \in 1..n}{B_i}$ with $B_i \neq \iuniontype$ for every $i \in 1..n$. Then, by definition and invertibility of subtyping for non-union types, from $\datatype{D'}{B'} \subtypemu B$ we have $B_j = \datatype{D'_j}{B'_j}$ for some $j \in 1..n$. Again, by Lem.~\ref{lem:generationForPatterns} (iii), $\exists D, A'$ such that $A = \datatype{D}{A'}$ and we conclude with $\lookup{A}{\epsilon} \cap \lookup{B}{\epsilon} \neq \varnothing$, given that $\cpos{p}{q} = \set{\epsilon}$.
    
    \item $q = \constterm{d}$: by Lem.~\ref{lem:generationForPatterns} (ii) we have $B = \consttype{d}$ which leads to a contradiction with $\datatype{D'}{B'} \subtypemu B$. Hence, this case is not possible.
    
    \item $q = \dataterm{q_1}{q_2}$: by Lem.~\ref{lem:generationForPatterns} (iii), $\exists D'', B''$ such that $B = \datatype{D''}{B''}$ with $\sequPDeriv{\theta'}{q_1 : D''}$ and $\sequPDeriv{\theta'}{q_2 : B''}$. Then, by invertibility of subtyping for non-union types, we get $D' \subtypemu D''$ and $B' \subtypemu B''$. Thus, $\sequTDeriv{\Gamma}{u_1 : D''}$ and $\sequTDeriv{u_2}{B''}$ by subsumption. On the other hand, by Lem.~\ref{lem:generationForPatterns} (iii), $\exists D, A'$ such that $A = \datatype{D}{A'}$ with $\sequPDeriv{\theta}{p_1 : D}$ and $\sequPDeriv{\theta}{p_2 : A'}$. Then, by inductive hypothesis, both $\Psicomp{\sequC{\theta}{p_1 : D'}}{\sequC{\theta'}{q_1 : D''}}$ and $\Psicomp{\sequC{\theta}{p_2 : A'}}{\sequC{\theta'}{q_2 : B''}}$ hold. Finally, since both patterns are compounds every mismatching position is internal, thus we can assure that $\Psicomp{\sequC{\theta}{p : A}}{\sequC{\theta'}{q : B}}$ holds too.
  \end{enumerate}
\end{itemize}
\end{proof}

Let $\Gamma, \theta$ be typing contexts, $\sigma$ a substitution. We write
$\sequT{\Gamma}{\sigma:\theta}$ to indicate that $\dom{\sigma}=\dom{\theta}$
and $\sequT{\Gamma}{\sigma(x) : \theta(x)}$, for all $x \in \dom{\sigma}$.
Likewise we use $\sequTDeriv{\Gamma}{\sigma : \theta}$ if each judgment is
derivable. Two auxiliary results before addressing SR.

The following lemma assures that the substitution yielded by a successful
match preserves the types of the variables in the pattern.

\begin{lemma}[Type of Successful Match]
\label{lem:typeOfSuccessfulMatch}
Suppose $\match{p}{u} = \sigma$ is successful, $\dom{\theta} = \fm{p}$, $\sequPDeriv{\theta}{p : A}$ and $\sequTDeriv{\Gamma}{u : A}$. Then $\sequTDeriv{\Gamma}{\sigma : \theta}$.
\end{lemma}

\begin{proof}
By induction on $p$.
\begin{itemize}
  \item $p = \matchable{x}$. Then $\sigma = \rename{x}{u}$ and, by Lem.~\ref{lem:generationForPatterns} (i), $x : A \in \theta$. Then $\theta = \set{x : A}$ and $\sequTDeriv{\Gamma}{\sigma : \theta}$ that holds by hypothesis.

  \item $p = \constterm{c}$. The property holds trivially as $\dom{\sigma} = \varnothing = \dom{\theta}$.

  \item $p = \dataterm{p_1}{p_2}$. Then, as the matching was successful, $u = \dataterm{u_1}{u_2}$ is a data structure and $\sigma = \match{p_1}{u_1} \uplus \match{p_2}{u_2} = \sigma_1 \uplus \sigma_2$. By Lem.~\ref{lem:generationForPatterns} (iii), $\exists D, A'$ such that $A = \datatype{D}{A'}$, $\sequPDeriv{\theta}{p_1 : D}$ and $\sequPDeriv{\theta}{p_2 : A'}$. Then, by Lem.~\ref{lem:basisForPatterns} (iii), $\sequPDeriv{\theta_1}{p_1 : D}$ and $\sequPDeriv{\theta_2}{p_2 : A'}$ with $\theta_1 = \res{\theta}{\fm{p_1}}$ and $\theta_2 = \res{\theta}{\fm{p_2}}$.

  On the other hand, by Lem.~\ref{lem:typingForDataStructures}, $\exists D', A''$ such that $\datatype{D'}{A''} \subtypemu A$, $\sequTDeriv{\Gamma}{u_1 : D'}$ and $\sequTDeriv{\Gamma}{u_2 : A''}$. From $\datatype{D'}{A''} \subtype \datatype{D}{A'} \eqtypemu A$ we get, by Prop.~\ref{prop:subtypingIsInvertible}, $D' \subtypemu D$ and $A'' \subtypemu A'$. Then we can derive  $\sequTDeriv{\Gamma}{u_1 : D}$ and $\sequTDeriv{\Gamma}{u_2 : A'}$ by applying $\ruleTSubs$.

  Finally we can apply the inductive hypothesis on both side of the derivation and we get $\sequTDeriv{\Gamma}{\sigma_1 : \theta_1}$ and $\sequTDeriv{\Gamma}{\sigma_2 : \theta_2}$. As $\sigma_1$ and $\sigma_2$ are disjoint then $\theta_1$ and $\theta_2$ are as well, and we can assure that $\sequTDeriv{\Gamma}{\sigma : \theta}$.
\end{itemize}
\end{proof}

Finally, we recall to the standard Substitution Lemma for type systems. It may also
be interpreted in the context of an abstraction. Given $\absterm{p}{s}{\theta}$, where
$\theta$ has the type assignments for the variables in $p$, every substitution that
preserves $\theta$ will also preserve the type of $s$ once $\theta$ is abstracted.

\begin{lemma}[Substitution Lemma]
\label{lem:substitution}
Suppose $\sequTDeriv{\Gamma,\theta}{s : A}$ and $\sequTDeriv{\Gamma}{\sigma : \theta}$. Then $\sequTDeriv{\Gamma}{\sigma s : A}$.
\end{lemma}

\begin{proof}
By induction on $s$.
\begin{itemize}
  \item $s = x$. By Lem.~\ref{lem:generation} (i), $\exists A'$ such that $A' \subtypemu A$ and $x : A' \in \Gamma, \theta$. If $x \in \dom{\sigma}$, as $\dom{\sigma} = \dom{\theta}$, $x : A' \in \theta$ and by hypothesis $\sequTDeriv{\Gamma}{\sigma(x) : \theta(x)}$. Then by $\ruleTSubs$ we get $\sequTDeriv{\Gamma}{\sigma x : A}$. If not, $x : A' \in \Gamma$ and $\sigma x = x$, then by $\ruleTVar$ and $\ruleTSubs$ we conclude $\sequTDeriv{\Gamma}{\sigma x : A}$.

  \item $s = \constterm{c}$. By Lem.~\ref{lem:generation} (ii), $A \subtypemu \consttype{c}$ and, as $\sigma\constterm{c} = \constterm{c}$, by $\ruleTConst$ and $\ruleTSubs$ we have $\sequTDeriv{\Gamma}{\sigma\constterm{c} : A}$.

  \item $s = \appterm{r}{u}$. By Lem.~\ref{lem:generation} (iii) we have two cases:
  \begin{enumerate}[(a)]
    \item either $\exists D, A'$ such that $\datatype{D}{A'} \subtypemu A$, $\sequTDeriv{\Gamma, \theta}{r : D}$ and $\sequTDeriv{\Gamma, \theta}{u : A'}$. By inductive hypothesis $\sequTDeriv{\Gamma}{\sigma r : D}$ and $\sequTDeriv{\Gamma}{\sigma u : A'}$. As $\dataterm{\sigma r}{\sigma u} = \sigma(\dataterm{r}{u})$ by $\ruleTComp$ and $\ruleTSubs$ we get $\sequTDeriv{\Gamma}{\sigma(\dataterm{r}{u}) : A}$.
    \item or $\exists A_1, \ldots, A_n, A'$ such that $A' \subtypemu A$, $\sequTDeriv{\Gamma, \theta}{r : \functype{\maxuniontype{i \in 1..n}{A_i}}{A'}}$, and $\sequTDeriv{\Gamma, \theta}{u : A_j}$ for some $j \in 1..n$. Similarly to the previous case, we apply the inductive hypothesis to get $\sequTDeriv{\Gamma}{\sigma r : \functype{\maxuniontype{i \in 1..n}{A_i}}{A'}}$ and $\sequTDeriv{\Gamma}{\sigma u : A_j}$. Then we conclude by $\ruleTApp$ and $\ruleTSubs$ that $\sequTDeriv{\Gamma}{\sigma(\dataterm{r}{u}) : A}$.
  \end{enumerate}

  \item $s = (\absterm{p_i}{s_i}{\theta_i})_{i \in 1..n}$. By Lem.~\ref{lem:generation} (iv), $\exists A_1, \ldots, A_n, B$ such that $\functype{\maxuniontype{i \in 1..n}{A_i}}{B} \subtypemu A$, $\lista{\sequC{\theta_i}{p_i : A_i}}_{i \in 1..n}$ is compatible, $\dom{\theta_i} = \fm{p_i}$, $\sequPDeriv{\theta_i}{p_i : A_i}$ and $\sequTDeriv{\Gamma, \theta, \theta_i}{s_i : B}$ for every $i \in 1..n$. Without loss of generality we can assume $ \sigma \avoids \theta_i$\footnote{Here we mean $\sigma \avoids x$ for every $x \in \dom{\theta_i}$.} and $\Gamma, \theta_i$ is a basis. Then $\sigma s = (\absterm{p_i}{\sigma s_i}{\theta_i})_{i \in 1..n}$ and, by Lem.~\ref{lem:basis} (i), $\sequTDeriv{\Gamma, \theta_i}{\sigma : \theta}$. By inductive hypothesis we get $\sequTDeriv{\Gamma, \theta_i}{\sigma s_i : B}$ for every $i \in 1..n$. Finally, by $\ruleTAbs$ and $\ruleTSubs$, we conclude $\sequTDeriv{\Gamma}{\sigma s : A}$.
\end{itemize}
\end{proof}


\section{Safety}
\label{sec:safety}

Subject Reduction (Prop.~\ref{prop:subjectReduction}) and Progress (Prop.~\ref{prop:progress}) are addressed next.

\begin{proposition}[Subject Reduction]
\label{prop:subjectReduction}
If $\sequTDeriv{\Gamma}{s : A}$ and $s \reduce s'$, then $\sequTDeriv{\Gamma}{s' : A}$.
\end{proposition}

\begin{proof}
By induction on $s$.
\begin{itemize}
  \item $s = x$ or $s = \constterm{c}$. The property holds trivially as there is no $s'$ such that $s \reduce s'$.

  \item $s = \appterm{r}{u}$. Here we may consider three possibilities:
  \begin{enumerate}
    \item $r \reduce r'$. By Lem.~\ref{lem:generation} (iii) we have two cases:
    \begin{enumerate}[(a)]
      \item either $\exists D, A'$ such that $\datatype{D}{A'} \subtypemu A$, $\sequTDeriv{\Gamma}{r : D}$ and $\sequTDeriv{\Gamma}{u : A'}$. By inductive hypothesis $\sequTDeriv{\Gamma}{r' : D}$. Then, by $\ruleTComp$ and $\ruleTSubs$, we have $\sequTDeriv{\Gamma}{s' : A}$.
      \item or $\exists A_1, \ldots, A_n, A'$ such that $A' \subtypemu A$, $\sequTDeriv{\Gamma}{r : \functype{\maxuniontype{i \in 1..n}{A_i}}{A'}}$, and $\sequTDeriv{\Gamma}{u : A_k}$ for some $k \in 1..n$. By inductive hypothesis $\sequTDeriv{\Gamma}{r' : \functype{\maxuniontype{i \in 1..n}{A_i}}{A'}}$ and by applying $\ruleTApp$ and $\ruleTSubs$ we conclude $\sequTDeriv{\Gamma}{s' : A}$.
    \end{enumerate}

    \item $u \reduce u'$. This case is similar to the previous one as by Lem.~\ref{lem:generation} we have the same two possible cases:
    \begin{enumerate}[(a)]
      \item either $\exists D, A'$ such that $\datatype{D}{A'} \subtypemu A$, $\sequTDeriv{\Gamma}{r : D}$ and $\sequTDeriv{\Gamma}{u : A'}$. By inductive hypothesis $\sequTDeriv{\Gamma}{u' : A'}$. Then, by $\ruleTComp$ and $\ruleTSubs$, we have $\sequTDeriv{\Gamma}{s' : A}$.
      \item or $\exists A_1, \ldots, A_n, A'$ such that $A' \subtypemu A$, $\sequTDeriv{\Gamma}{r : \functype{\maxuniontype{i \in 1..n}{A_i}}{A'}}$, and $\sequTDeriv{\Gamma}{u : A_k}$ for some $k \in 1..n$. By inductive hypothesis $\sequTDeriv{\Gamma}{u' : A_k}$ and by applying $\ruleTApp$ and $\ruleTSubs$ we conclude $\sequTDeriv{\Gamma}{s' : A}$.
    \end{enumerate}

    \item $r = (\absterm{p_i}{s_i}{\theta_i})_{i \in 1..n}$ and $s' = \match{p_k}{u}s_k$ for some $k \in 1..n$ such that $\match{p_k}{u} = \sigma$ and $\match{p_i}{u} = \fail$ for every $i < k$. Assume, towards an absurd, that Lem.~\ref{lem:generation} (iii.a) holds for $s$. Then, $\exists D, A'$ such that $\datatype{D}{A'} \subtypemu A$, $\sequTDeriv{\Gamma}{r : D}$ and $\sequTDeriv{\Gamma}{u : A'}$. But, by Lem.~\ref{lem:generation} (iv) applied to $\sequTDeriv{\Gamma}{r : D}$, $\exists A_1, \ldots, A_n, B$ such that $\functype{\maxuniontype{i \in 1..n}{A_i}}{B} \subtypemu D$ and, by Lem.~\ref{lem:supertypesOfNonUnionTypes}, $\exists \U$ such that $D \eqtypemu \U[B']$ with $\functype{\maxuniontype{i \in 1..n}{A_i}}{B} \subtypemu B'$ which is a contradiction since $D$ is a data type. Thus, it must be the case that Lem.~\ref{lem:generation} (iii.b) holds for $s$.

    Then, $\exists C_1, \ldots, C_m, A'$ such that $A' \subtypemu A$, $\sequTDeriv{\Gamma}{r : \functype{\maxuniontype{j \in 1..m}{C_m}}{A'}}$ and:
    \begin{equation}
      \label{eq:sr:i}
      \sequTDeriv{\Gamma}{u : C_{k'}}
    \end{equation}
    for some $k' \in 1..m$. Applying once again Lem.~\ref{lem:generation} (iv), this time to $\sequTDeriv{\Gamma}{r : \functype{\maxuniontype{j \in 1..m}{C_m}}{A'}}$, we get $\exists A_1, \ldots, A_n, B$ such that: 
    \begin{equation}
      \label{eq:sr:ii}
      \functype{\maxuniontype{i \in 1..n}{A_i}}{B} \subtypemu \functype{\maxuniontype{j \in 1..m}{C_m}}{A'}
    \end{equation}
    $\dom{\theta_i} = \fm{p_i}$, $\lista{\sequC{\theta_i}{p_i : A_i}}_{i \in 1..n}$ is compatible, $\sequPDeriv{\theta_i}{p_i : A_i}$ and $\sequTDeriv{\Gamma, \theta_i}{s_i : B}$ for every $i \in 1..n$.

    From~(\ref{eq:sr:ii}) and Prop.~\ref{prop:subtypingIsInvertible} we have $B \subtypemu A'$ and
    \begin{equation}
      \label{eq:sr:iii}
      \maxuniontype{j \in 1..m}{C_m} \subtypemu \maxuniontype{i \in 1..n}{A_i}
    \end{equation}

    We want to show that $\sequTDeriv{\Gamma}{u : A_k}$. For that we need to distinguish two cases:
    \begin{enumerate}
      \item If $u$ is in matchable form, we have two possibilities:
      \begin{enumerate}
        \item $u$ is a data structure: then, by Lem.~\ref{lem:typingForDataStructures}, there exists a non-union datatype $D$ such that $D \subtypemu C_{k'}$ and $\sequTDeriv{\Gamma}{u : D}$.
        \item $u$ is an abstraction: then, by Lem.~\ref{lem:generation} (iv), $\exists C', C''$ such that $\functype{C'}{C''} \subtypemu C_{k'}$ and $\sequTDeriv{\Gamma}{u : \functype{C'}{C''}}$.
      \end{enumerate}
      Then, in both cases there exists a non-union type, say $C$, such that $C \subtypemu C_{k'}$ and $\sequTDeriv{\Gamma}{u : C}$. Then, from (\ref{eq:sr:iii}) we get: $$C \subtypemu \maxuniontype{i \in 1..n}{A_i}$$ and, since $C$ is non-union, $C \subtypemu A_l$ for some $l \in 1..n$. Hence, by subsumption $\sequTDeriv{\Gamma}{u : A_l}$.
  
      If $k = l$ we are done, so assume $k \neq l$. Recall the conditions for the reduction rule, where $\match{p_i}{u} = \fail$ for every $i < k$. Then, by Lem.~\ref{lem:mismatchingTypes}, we have $A_l \not\subtypemu A_i$. Thus, it must be the case that $k < l$. By Lem.~\ref{lem:compatibility} with hypothesis $\sequTDeriv{\Gamma}{u : A_l}$, $\sequPDeriv{\theta_k}{p_k : A_k}$, $\sequPDeriv{\theta_l}{p_l : A_l}$ and $\match{p_k}{u} = \sigma$ we get that $\Psicomp{\sequC{\theta_k}{p_k : A_k}}{\sequC{\theta_l}{p_l : A_l}}$ holds. Additionally, we already saw that the list $\lista{\sequC{\theta_i}{p_i : A_i}}_{i \in 1..n}$ is compatible, thus $\compatible{\sequC{\theta_k}{p_k : A_k}}{\sequC{\theta_l}{p_l : A_l}}$ and by definition $A_l \subtypemu A_k$. Finally we conclude by subsumption once again, $\sequTDeriv{\Gamma}{u : A_k}$.
  
      \item If $u$ is not in matchable form, then $p_k = \matchable{x}$ and by the premises of the reductions rule we need $\match{p_i}{u} = \fail$ for every $i < k$. Thus, necessarily $k = 1$. Moreover, since $\matches{\matchable{x}}{p_i}$ for every $i \in 1..n$, by compatibility we have $A_i \subtypemu A_k$. Then, from (\ref{eq:sr:iii}) we get $$C_{k'} \subtypemu \maxuniontype{j \in 1..m}{C_j} \subtypemu \maxuniontype{i \in i..n}{A_i} \subtypemu A_k$$ Thus, by subsumption, $\sequTDeriv{\Gamma}{u : A_k}$.
    \end{enumerate}

    Finally, in either case we have $\sequTDeriv{\Gamma}{u : A_k}$. Now Lem.~\ref{lem:typeOfSuccessfulMatch} and~\ref{lem:substitution} with $\sequTDeriv{\Gamma, \theta_k}{s_k : B}$ entails $\sequTDeriv{\Gamma}{s' : B}$ and we conclude by subsumption, $\sequTDeriv{\Gamma}{s' : A}$ (recall $B \subtypemu A' \subtypemu A$).
  \end{enumerate}
  
  \item $s = (\absterm{p_i}{s_i}{\theta_i})_{i \in 1..n}$. Then $s' = \absterm{p_1}{s_1}{\theta_1} \icaseterm \ldots \icaseterm \absterm{p_k}{s'_k}{\theta_n} \icaseterm \ldots \icaseterm \absterm{p_n}{s_n}{\theta_n}$ with $s_k \reduce s'_k$. By Lem.~\ref{lem:generation} (iv), $\exists A_1, \ldots, A_n, B$ s.t. $\functype{\maxuniontype{i \in 1..n}{A_i}}{B} \subtypemu A$, $\lista{\sequC{\theta_i}{p_i : A_i}}_{i \in 1..n}$ is compatible, $\dom{\theta_i} = \fm{p_i}$, $\sequPDeriv{\theta_i}{p_i : A_i}$ and $\sequTDeriv{\Gamma, \theta_i}{s_i : B}$ for every $i \in 1..n$. By inductive hypothesis $\sequTDeriv{\Gamma, \theta_k}{s'_k : A_k}$ and by applying $\ruleTAbs$ and $\ruleTSubs$ we conclude $\sequTDeriv{\Gamma}{s' : A}$.
\end{itemize}
\end{proof}

Let the set of \emphdef{values} be defined as $v \Coloneq \dataterm{x}{v_1 \ldots v_n} \mathrel| \dataterm{\constterm{c}}{v_1 \ldots v_n} \mathrel| (\absterm{p_i}{s_i}{\theta_i})_{i \in 1..n}$.
The following auxiliary property guarantees the success of matching for well-typed closed values (note that values are already in matchable form).

\begin{lemma}[Successful Match for Closed Values]
\label{lem:successfulMatchForClosedValues}
Suppose $\sequTDeriv{}{v : A}$ and $\sequPDeriv{\theta}{p : A}$ where $v$ is a value. Then, $\match{p}{v}$ is successful.
\end{lemma}

\begin{proof}
By induction on $p$. Note that $v$ cannot be a variable since it is typed on the empty context and, by Lem.~\ref{lem:basis}, $\fv{v} \subseteq \varnothing$. Hence, it is a closed term. Then $v$ is either a data structure or a case.
\begin{itemize}
  \item $p = \matchable{x}$. The property holds trivially with the substitution $\rename{x}{v}$.
	
  \item $p = \constterm{c}$. By Lem.~\ref{lem:generationForPatterns} (ii), $A = \consttype{c}$. Suppose $v = (\absterm{q_i}{s_i}{\theta_i})_{i \in 1..n}$. By Lem.~\ref{lem:generation} (iv), $\exists A_1, \ldots, A_n, B$ such that $\functype{\maxuniontype{i \in 1..n}{A_i}}{B} \subtypemu \consttype{c}$ and, by Lem.~\ref{lem:supertypesOfNonUnionTypes}, $\exists \U, A'$ such that $\consttype{c} \eqtypemu \U[A']$, $\functype{\maxuniontype{i \in 1..n}{A_i}}{B} \subtypemu A'$ and they both have the same outermost type constructor. This leads to a contradiction. Hence $v$ is not a case.

  Then it must be a data structure. By Lem.~\ref{lem:typingForDataStructures}, $\exists D$ such that $D$ is a non-union type, $D \subtypemu \consttype{c}$ and $\sequTDeriv{}{r : D}$. Furthermore, case (2) of the lemma does not hold since $A \eqtypemu \consttype{c}$. Then, by case (1), $v = \constterm{c}$ and $D \eqtypemu \consttype{c}$. Finally we can assure that $\match{p}{v} = \match{\constterm{c}}{\constterm{c}}$ is successful.

  \item $p = \dataterm{p_1}{p_2}$. By Lem.~\ref{lem:generationForPatterns} (iii), $\exists D, A'$ such that $A = \datatype{D}{A'}$, $\sequPDeriv{\theta}{p_1 : D}$ and $\sequPDeriv{\theta}{p_2 : A'}$. Similarly to the previous case we may conclude that if $v = (\absterm{q_i}{s_i}{\theta_i})_{i \in 1..n}$ there exists a functional type $B$ such that $\datatype{D}{A'} \eqtypemu \U[B]$ which leads to a contradiction. Hence we are again in the case that $v$ is a data structure.

  By Lem.~\ref{lem:typingForDataStructures}, $\exists D'$ such that $D'$ is a non-union type, $D' \subtypemu \datatype{D}{A'}$ and $\sequTDeriv{}{v : D'}$. Moreover, we can assure that case (2) of the lemma holds, so we have $v = \dataterm{v_1}{v_2}$ and $\exists D'', A''$ such that $D' \eqtypemu \datatype{D''}{A''}$, $\sequTDeriv{}{v_1 : D''}$ and $\sequTDeriv{}{v_2 : A''}$. Now by Prop.~\ref{prop:subtypingIsInvertible} with $\datatype{D''}{A''} \subtypemu \datatype{D}{A'}$ we get $D'' \subtypemu D$ and $A'' \subtypemu A'$, and by $\ruleTSubs$ $\sequTDeriv{}{v_1: D}$ and $\sequTDeriv{}{v_2: A'}$.

  Then we can apply the inductive hypothesis and to deduce that both $\match{p_1}{v_1}$ and $\match{p_2}{v_2}$ are successful. Finally by linearity of patterns we can safely conclude that $\match{p}{v} = \match{p_1}{v_1} \uplus \match{p_2}{v_2}$ is also successful.
\end{itemize}
\end{proof}

\begin{proposition}[Progress]
\label{prop:progress}
If $\sequTDeriv{}{s : A}$ and $s$ is not a value, then $\exists s'$ s.t. $s \reduce s'$.
\end{proposition}

\begin{proof}
By induction on $s$ analyzing the subterm of $s$ that is not yet a value.
\begin{itemize}
  \item $s = x$, $s = \constterm{c}$ or $s = (\absterm{p_i}{s_i}{\theta_i})_{i \in 1..n}$. The property holds trivially as $s$ is already a value.

  \item $s = \appterm{r}{u}$. Here we have three possible cases:
  \begin{enumerate}
    \item $r$ is not yet a value. Then, by Lem.~\ref{lem:generation} (iii), $\exists A_1, A_2$ such that $\sequTDeriv{}{r : A_1}$ and $\sequTDeriv{}{u : A_2}$. By inductive hypothesis $\exists r'$ such that $r \reduce r'$ and we conclude with $s' = \appterm{r'}{u}$.
    \item $r$ is a value and $u$ is not. Again by Lem.~\ref{lem:generation} (iii), $\exists A_1, A_2$ such that $\sequTDeriv{}{r : A_1}$ and $\sequTDeriv{}{u : A_2}$. By inductive hypothesis $\exists u'$ such that $u \reduce u'$ and we conclude with $s' = \appterm{r}{u'}$.
    \item $r = (\absterm{p_i}{s_i}{\theta_i})_{i \in 1..n}$ with $u$ already a value. As for SR, by Lem.~\ref{lem:generation} (iii.b), we have that $\exists C_1, \ldots, C_m, A'$ such that $A' \subtypemu A$, $\sequTDeriv{}{r :\functype{\maxuniontype{j \in 1..m}{C_j}}{A'}}$ and
    \begin{equation}
      \label{eq:prog:i}
      \sequTDeriv{}{u : C_{k'}}
    \end{equation}
    for some $k' \in 1..m$. And, by Lem.~\ref{lem:generation} (iv) on $\sequTDeriv{}{r :\functype{\maxuniontype{j \in 1..m}{C_j}}{A'}}$, $\exists A_1, \ldots, A_n, B$ such that
    \begin{equation}
      \label{eq:prog:ii}
      \functype{\maxuniontype{i \in 1..n}{A_i}}{B} \subtypemu \functype{\maxuniontype{j \in 1..m}{C_j}}{A'}
    \end{equation}
    $\dom{\theta_i} = \fm{p_i}, \lista{\sequC{\theta_i}{p_i : A_i}}_{i \in 1..n}$ is compatible, $\sequPDeriv{\theta_i}{p_i : A_i}$ and $\sequTDeriv{\theta_i}{s_i : B}$ for every $i \in 1..n$.

    From (\ref{eq:prog:ii}) and Prop.~\ref{prop:subtypingIsInvertible} we have $B \subtypemu A'$ and
    \begin{equation}
      \label{eq:prog:iii}
      \maxuniontype{j \in 1..m}{C_j} \subtypemu \maxuniontype{i \in 1..n}{A_i}
    \end{equation}
    Additionally, by (\ref{eq:prog:i}) and Lem.~\ref{lem:basis} we know that $u$ is a closed value, \ie a data structure or an abstraction. Hence, $u$ is in matchable form and matching agains every pattern $p_i$ is decided. Then, we have to possibilities as in the proof for SR:
    \begin{enumerate}
      \item $u$ is a data structure: by Lem.~\ref{lem:typingForDataStructures}, there exists a non-union datatype $D$ such that $D \subtypemu C_{k'}$ and $\sequTDeriv{\Gamma}{u : D}$.
      \item $u$ is an abstraction: by Lem.~\ref{lem:generation} (iv), $\exists C', C''$ such that $\functype{C'}{C''} \subtypemu C_{k'}$ and $\sequTDeriv{\Gamma}{u : \functype{C'}{C''}}$.
    \end{enumerate}
    In both cases we can assume there is a non-union type, say $C$, such that $C \subtypemu C_{k'}$ and $\sequTDeriv{}{u : C}$. Then, from (\ref{eq:prog:iii}) we get $C \subtypemu \maxuniontype{i \in 1..n}{A_i}$ and $C \subtypemu A_k$ for some $k \in 1..n$, as before. Thus, by subsumption, $\sequTDeriv{}{u : A_k}$. Finally, with $\sequPDeriv{\theta_k}{p_k : A_k}$ we are under the hypothesis of Lem.~\ref{lem:successfulMatchForClosedValues}, and we conclude by taking $s' = \match{p_k}{u}s_k$.
  \end{enumerate}
\end{itemize}
\end{proof}


\section{Conclusions}

A type system is proposed for a calculus that supports path polymorphism and
two fundamental properties are addressed, namely Subject Reduction and
Progress. The type system includes type application, constants as types,
union and recursive types. Both properties rely crucially on a notion of
pattern \emph{compatibility} and on invertibility of subtyping of $\mu$-types.
This last result is proved via a coinductive semantics for the finite
$\mu$-types. Regarding future work an outline of possible avenues follows.

\begin{itemize}
  \item There exists extensive work on type-checking for recursive
  types~\cite{Eikelder:1991,DBLP:journals/toplas/AmadioC93,DBLP:conf/fossacs/JhaPZ02},
  including some efficient algorithms for both
  equivalence~\cite{DBLP:journals/iandc/PalsbergZ01} and
  subtyping~\cite{DBLP:journals/mscs/KozenPS95}. We are currently adapting
  these ideas to $\capp$.
  
  \item We already mentioned the addition of parametric polymorphism
  (presumably in the style of
  F$_{<:}$~\cite{Cardelli1991,Pierce:2002:TPL:509043,Colazzo200571}). We
  believe this should not present major difficulties. 
  
  \item Strong normalization requires devising a notion of positive/negative
  occurrence in the presence of strong $\mu$-type equality, which is known not
  to be obvious~\cite[page 515]{BDS:2013}.
  
  \item A more ambitious extension is that of \emph{dynamic patterns}, namely
  patterns that may be computed at run-time, \ppc\ being the prime example of a
  calculus supporting this feature.
\end{itemize}


\bibliographystyle{entcs}
\bibliography{biblio}

\end{document}
